\documentclass[10pt,a4paper]{article}

\usepackage{amsmath,amsgen,latexsym}
\usepackage{amstext,amssymb,amsfonts,latexsym}
\usepackage{theorem}
\usepackage{pifont}
\usepackage[dvips]{graphics,epsfig}
\usepackage[dvips]{graphicx}

 \newcommand{\bs}{\bigskip}
 \newcommand{\ms}{\medskip}
 \newcommand{\n}{\noindent}
 \newcommand{\s}{\smallskip}
 \newcommand{\hs}[1]{\hspace*{ #1 mm}}
 \newcommand{\vs}[1]{\vspace*{ #1 mm}}

 \newcommand{\setempty}{\mathrm{\O}}
 \newcommand{\real}{\mathbb{R}}
 \newcommand{\nat}{\mathbb{N}}
 
 \newcommand{\integer}{\mathbb{Z}}
 \newcommand{\rational}{\mathbb{Q}}

 \newcommand{\co}{\mathrm{co}\mbox{-}}

 \newcommand{\ie}{\textrm{i.e.},\hspace*{2mm}}
 \newcommand{\eg}{\textrm{e.g.},\hspace*{2mm}}
 
 \newcommand{\etalc}{\textrm{et al.}}

 \newcommand{\DD}{{\cal D}}

 \newcommand{\SSS}{{\cal S}}
 
 \newcommand{\PP}{{\cal P}}

 \newcommand{\dl}{\mathrm{L}}
 \newcommand{\nl}{\mathrm{NL}}
 \newcommand{\p}{\mathrm{P}}
 \newcommand{\np}{\mathrm{NP}}

 \newcommand{\poly}{\mathrm{poly}}

 \newcommand{\fl}{\mathrm{FL}}
 \newcommand{\fp}{\mathrm{FP}}

 \def\bbox{\vrule height6pt width6pt depth1pt}

\theoremstyle{plain}
\theoremheaderfont{\bfseries}
 \newtheorem{theorem}{Theorem}[section]
 \newtheorem{lemma}[theorem]{Lemma}
 \newtheorem{proposition}[theorem]{Proposition}
 \newtheorem{corollary}[theorem]{Corollary}

 {\theorembodyfont{\rmfamily}
  }
 {\theorembodyfont{\rmfamily} }
 {\theorembodyfont{\rmfamily} }

 \newenvironment{proof}{\par \noindent
            {\bf Proof. \hs{2}}}{\hfill$\Box$ \vspace*{3mm}}

 \newenvironment{proofof}[1]{\vspace*{5mm} \par \noindent
         {\bf Proof of #1.\hs{2}}}{\hfill$\Box$ \vspace*{3mm}}

 \newcommand{\ceilings}[1]{\lceil #1 \rceil}
 \newcommand{\floors}[1]{\lfloor #1 \rfloor}
 \newcommand{\pair}[1]{\langle #1 \rangle}

\newcommand{\ignore}[1]{}

\newcommand{\cent}{{|}\!\!\mathrm{c}}
\newcommand{\dollar}{\$}

\newcommand{\APreduces}{\leq_{\mathrm{AP}}}

\newcommand{\EXreduces}{\leq_{\mathrm{EX}}}

\newcommand{\sAPreduces}{\leq_{\mathrm{sAP}}}

\newcommand{\apx}{\mathrm{APX}}
\newcommand{\apxp}{\mathrm{APXP}}
\newcommand{\apxl}{\mathrm{APXL}}
\newcommand{\po}{\mathrm{PO}}
\newcommand{\npo}{\mathrm{NPO}}
\newcommand{\lo}{\mathrm{LO}}
\newcommand{\nlo}{\mathrm{NLO}}
\newcommand{\ptas}{\mathrm{PTAS}}
\newcommand{\lsas}{\mathrm{LSAS}}
\newcommand{\ncas}[1]{\mathrm{NC}^{ #1 }\mathrm{AS}}
\newcommand{\acas}[1]{\mathrm{AC}^{ #1 }\mathrm{AS}}
\newcommand{\pbo}{\mathrm{PBO}}

\newcommand{\apxnc}[1]{\mathrm{APXNC}^{ #1 }}
\newcommand{\apxac}[1]{\mathrm{APXAC}^{ #1 }}
\newcommand{\nco}[1]{\mathrm{NC}^{ #1 }\mathrm{O}}
\newcommand{\aco}[1]{\mathrm{AC}^{ #1 }\mathrm{O}}

\newcommand{\ac}[1]{\mathrm{AC}^{ #1 }}
\newcommand{\nc}[1]{\mathrm{NC}^{ #1 }}
\newcommand{\tc}[1]{\mathrm{TC}^{ #1 }}
\newcommand{\sac}[1]{\mathrm{SAC}^{ #1 }}

\newcommand{\dlogtime}{\mathrm{DLOGTIME}}

\newcommand{\auxp}{\mathrm{auxP}}
\newcommand{\auxfp}{\mathrm{auxFP}}

\newcommand{\auxl}{\mathrm{auxL}}
\newcommand{\auxfl}{\mathrm{auxFL}}
\newcommand{\auxfnc}[1]{\mathrm{auxFNC}^{ #1 }}

\newcommand{\fac}[1]{\mathrm{FAC}^{ #1 }}
\newcommand{\fnc}[1]{\mathrm{FNC}^{ #1 }}

\newcommand{\minnl}{\mathrm{MinNL}}
\newcommand{\maxnl}{\mathrm{MaxNL}}
\newcommand{\minp}{\mathrm{MinP}}
\newcommand{\maxp}{\mathrm{MaxP}}
\newcommand{\optp}{\mathrm{OptP}}
\newcommand{\optl}{\mathrm{OptL}}

\begin{document}
%%%%%%%%%%%%%%%%%%
%%%%%%%%%%%%%%%%%%

\begin{center}
{\Large {\bf Uniform-Circuit and Logarithmic-Space Approximations of \s\\
Refined Combinatorial Optimization Problems}}\footnote{A preliminary report  appeared in the Proceedings of the 7th International Conference on Combinatorial Optimization and Applications (COCOA 2013), Chengdu, China, December 12--14, 2013, Lecture Notes in Computer Science, Springer-Verlag, vol.8287, pp.318--329, 2013.} \bs\\
{\sc Tomoyuki Yamakami}\footnote{Present Affiliation: Graduate School of Engineering, University of Fukui, 3-9-1 Bunkyo, Fukui 910-8507, Japan} \bs\\
\end{center}

\pagestyle{plain}
%%%%%%%%%%%%%%%%%%

\begin{quote}
\n{\bf Abstract:}
A significant progress has been made in the past three decades over the study of combinatorial NP optimization problems and their associated optimization and approximate classes, such as NPO, PO, APX (or APXP), and PTAS. Unfortunately, a collection of problems that are simply placed inside the P-solvable optimization class PO never have been studiously analyzed regarding their exact computational complexity. To improve this situation, the existing framework based on polynomial-time computability  needs to be expanded and further refined for an insightful analysis of various approximation algorithms targeting optimization problems within PO.
In particular, we deal with those problems characterized in terms of logarithmic-space computations and uniform-circuit computations.
We are focused on nondeterministic logarithmic-space (NL)  optimization problems or NPO problems.
Our study covers a wide range of optimization and approximation classes, dubbed as, NLO, LO, APXL, and LSAS as well as new classes NC$^{1}$O, APXNC$^{1}$, NC$^{1}$AS, and AC$^{0}$O, which are founded on uniform families of Boolean circuits.
Although many NL decision problems can be naturally converted into NL optimization (NLO) problems, few NLO problems have been studied vigorously.
We thus provide a number of new NLO problems falling into
those low-complexity classes.
With the help of NC$^{1}$ or AC$^{0}$ approximation-preserving reductions, we also identify the most difficult problems (known as complete problems) inside those classes. Finally, we demonstrate a number of collapses and separations among those refined optimization and approximation classes with or without unproven complexity-theoretical assumptions.

\s

\n{\bf Keywords:}
optimization problem, approximation-preserving reduction, approximation algorithm, NC$^1$ circuit, AC$^{0}$ circuit, logarithmic space, complete problem
\end{quote}

%%%%%%%%%%%%%%%%%
%%%%%%%%%%%%%%%%%
\sloppy

\section{Refined Combinatorial Optimization Problems}\label{sec:intro}

\subsection{NL Optimization Problems}

Many combinatorial problems can be understood as sets of constraints (or requirements), which specify certain relations between {\em admissible instances}  and {\em feasible solutions}. Of such problems, a {\em combinatorial optimization problem}, in particular, asks to find an ``optimal'' solution that satisfies  certain constraints specified by  each given admissible instance, where the optimality usually takes a form of either ``maximization'' or ``minimization'' according to a predetermined ordering over all feasible solutions.
When finding  such optimal solutions is costly, we often resort to look for  solutions that are close enough to the desired optimal solutions.
A significant progress had been made in a field of fundamental research on these combinatorial optimization problems during 1990s and its trend has continued promoting our understandings of the approximability of the problems.  In particular, {\em NP optimization problems} (or {\em NPO problems}, in short) have been a centerfold of our interests because of their direct connection to NP (nondeterministic polynomial time) decision problems.

NPO problems are naturally derived from NP decision problems. As a typical NP problem, let us consider the {\em CNF Boolean formula satisfiability problem} (SAT) of determining whether a satisfying assignment exists for a given Boolean formula in conjunctive normal form. It is easy to convert {\sc SAT} to its corresponding  optimization problem, {\sc Maximum Weighted Satisfiability}, of finding a satisfying assignment having the maximal weight. This problem is an $\npo$ problem.
As is customary, the notation $\npo$ also denotes the collection of such optimization problems. Of those $\npo$ problems, those that can be solved exactly in polynomial time form a ``tractable'' optimization class $\po$, whereas an approximation class $\apx$ (which is hereafter denoted by $\apxp$ to emphasize its feature of ``polynomial time'' in comparison with ``logarithmic space'' and ``circuits'') consists of $\npo$ problems whose optimal solutions are {\em relatively approximated} within constant factors in polynomial time.
Another optimization problem, {\sc Maximum Cut}, of finding a partition of a given graph into two disjoint sets that maximize the number of crossing edges falls into this approximation class $\apxp$.

Up to now, a large number of $\npo$ problems have been nicely  classified into those classes of optimization problems (see, e.g., \cite[Compendium]{ACG+03}).
Among those optimization and approximation classes, $\po$ is the smallest class and has been proven to contain a number of intriguing optimization problems, including a minimization problem, {\sc Min Weight-st-Cut}, of finding a minimal $s$-$t$ cut of a given directed graph. In a study on NPO problems, the use of {\em approximation-preserving reductions} helps us identify the most difficult optimization problems in a given class of optimization problems and
many natural problems have been classified as the computationally hardest problems for $\npo$, $\po$, or $\apxp$. Those problems are known as ``complete'' problems. {\sc Maximum Weighted Satisfiability} and {\sc Maximum Cut} are respectively proven to be complete for $\nlo$ and $\apxp$.

The above classification of optimization problems is all described from a single viewpoint of ``polynomial-time''
computability and approximability and, as a result,
a systematic discussion on optimization problems inside $\po$ has been vastly neglected although $\po$ contains numerous intriguing problems of various complexities. For instance, the {\em minimum path weight problem} ({\sc Min Path-Weight}) is to find in a given directed graph $G$ a path $\SSS=(v_1,v_2,\ldots,v_k)$ with $k\geq2$ from given vertex  $v_1$ to another vertex $v_k$ having its (biased) path weight having binary representation of the form  $bin(w(v_1))bin(w(v_2))\cdots bin(w(v_k))$, where $bin(a)$ denotes the binary representation of a nonnegative integer $a$. This minimization problem {\sc Min Path-Weight} belongs to $\po$.
Another example is the {\em maximum Boolean formula value problem} ({\sc Max BFVP}) of finding a maximal subset of a given set of Boolean formulas that are satisfied by a given truth assignment. This simple problem also resides inside $\po$; however, it apparently looks much easier to solve than {\sc Min Path-Weight}.  This circumstantial evidence leads us to ponder that there might exist a finer and richer structure inside $\po$.
Consequently, we may raise a natural question of whether it is possible to find  such a finer structure within $\po$.

To achieve this goal, we first seek to develop a {\em new, finer framework}---a low-complexity world of optimization problems---and reexamine the computational complexity of such optimization problems within this new framework.
For this purpose, we need to reshape the existing framework of expressing optimization complexity classes by clarifying the scope and complexity of verification processes used for solutions using objective (or measure) functions.  While {\sc Min Weight-st-Cut} is known to be one of the most difficult problems in $\po$ under $\p$-reductions (even under $\nc{1}$-reductions, shown in Proposition \ref{min-st-cut-is-po}), the computational complexity of {\sc Min Path-Weight} seems to be significantly lower than {\sc Min Weight-st-Cut} residing in $\po$.
To study the fine structures inside $\po$, we wish to shift our interest from a paradigm of polynomial-time optimization to much lower-complexity optimization, notably logarithmic-space or uniform-circuit optimization.

In the past decades, {\em logarithmic-space} (or {\em log-space}) computation has exhibited intriguing features, which are often different from those of polynomial-time computation. A notable result is the closure property of $\nl$ (nondeterministic logarithmic space) under complementation \cite{Imm88,Sze88}.

{\`A}lvarez and Jenner \cite{AJ93,AJ95} first studied optimization problems from a viewpoint of log-space computability and discussed a class $\mathrm{OptL}$ of functions that compute optimal solutions using only a logarithmic amount of memory storage. In contrast, along the line of a study on NP optimization problems, Tantau \cite{Tan07} investigated {\em nondeterministic logarithmic-space (NL) optimization
problems} or {\em NLO problems}.
Intuitively, an NLO problem $Q$  is asked to to find its optimal solutions among all possible feasible solutions of size polynomial in input size $n$, provided that,
(i) we can check, using only $O(\log{n})$ memory space, whether any  given solution candidate $y$ is indeed a solution of the problem $Q$ and, if so, (ii) we can calculate the objective value of $y$ using $O(\log{n})$ memory space.
We simply write $\nlo$ for the collection of all $\nlo$ problems. It turns out that significant differences actually exist between two optimization classes $\npo$ and $\nlo$. One of the  crucial differences is caused by the way that an underlying Turing machine produces its output strings on its output tape. When a log-space machine writes such a string, the machine must produce it {\em obliviously} because the output string is usually longer than the machine's memory size. In short, log-space computation cannot remember polynomially-many symbols.
As a result, unlike $\npo$ problems, such machines do not seem to implement a typical approximation-preserving reduction between minimization problems and maximization problems inside $\nlo$ (see Section \ref{sec:PBP}).
When we discuss $\nlo$ problems, we need to heed the size of objective functions. An optimization problem is {\em polynomially bounded} if its objective (or measure) function outputs only polynomially-large integers.

Throughout this paper, we shall target those intriguing NLO problems. As unfolded  in later sections, NLO problems occupy a substantial portion of PO and they include numerous important and natural problems.
The aforementioned problems {\sc Min Path-Weight} and {\sc Max BFVP} are typical examples of the NLO problems. As other examples, the class NLO contains a restricted knapsack problem, called {\sc Max 2BCU-Knapsack}, and a restricted algebraic problem, called {\sc Max AGen} (see Section \ref{sec:approximation-class} for their definitions).
When we refer to PO, APXP, and PTAS in the existing framework based on NPO problems, we need to clarify their underlying framework; therefore, we intend to use new notations $\po_{\npo}$ (instead of $\po$), $\ptas_{\npo}$ (instead of $\ptas$), and $\apxp_{\npo}$ (instead of $\apxp$), when we discuss exact solvability and approximability of ``$\npo$  problems.''

%%%%%%
\subsection{Optimization Problems Inside NLO}

By shifting the paradigm of optimization problems, we wish to look into a world of NLO problems and to unearth rich and complex structures underlying in this  world.
Of all $\nlo$ problems, those that cane be {\em $\dl$-solvable} (\ie solvable exactly by multi-tape deterministic Turing machines using logarithmic space)
form an optimization class $\lo_{\nlo}$. If we restrict input graphs of {\sc Min Path-Weight} onto undirected forests, then the resulted problem, called {\sc Min Forest-Path-Weight}, belongs to $\lo_{\nlo}$.
Using uniform families of $\nc{1}$-circuits and $\ac{0}$-circuits in place of log-space Turing machines used in the existing notion of AP reduction, respectively, we can introduce two extra optimization classes $\nco{1}_{\nlo}$ and $\ac{0}\mathrm{O}_{\nlo}$,
where
$\nc{1}$ refers to $O(\log{n})$-depth polynomial-size circuits of bounded fan-in AND and OR gates and AC$^{0}$ indicates constant-depth polynomial-size circuits of unbounded fan-in AND and OR gates.

In analogy with $\apxp_{\npo}$, another refined approximation class $\apxl_{\nlo}$ is introduced using log-space approximation algorithms for NLO problems. Between $\apxl_{\nlo}$ and $\lo_{\nlo}$ exists a special class of optimization problems that have log-space approximation schemes. We call this class $\lsas_{\nlo}$, similar to $\ptas_{\npo}$.
In a similar way, we define $\apxnc{k}_{\nlo}$, $\apxac{k}_{\nlo}$, $\ncas{k}_{\nlo}$, and $\acas{k}_{\nlo}$ for each index $k\in\{0,1\}$.

To compare the complexity of $\nlo$ problems, we consider {\em approximation-preserving (AP) reduction}, {\em exact (EX) reduction}, and {\em strong AP (sAP) reduction} using logarithmic space or by $\nc{1}$-circuits (or even $\ac{0}$-circuits).
Using those weak reductions, we shall present  in Section \ref{sec:completeness}--\ref{sec:PBP} a number of concrete optimization problems that are complete for the aforementioned refined classes of optimization problems.  As discussed in Section \ref{sec:why-NC1}, those weak reductions are necessary for low-complexity optimization problems, because strong reductions tend to obscure the essential characteristics of ``complete'' problems.
Because of their fundamental nature, approximation classes are quite sensitive to the use of weak reductions. To use such reductions, we need to guarantee the existence of certain approximation bounds that must be easy to estimate.

Unlike NPO problems, a special attention is required for ``complete'' problems among NLO problems. Because of its logarithmic space-constraint, at this moment, it is unknown that complete problems actually exist in $\nlo$. What we do know is the existence of complete problems for the class $\maxnl$ of all maximization NLO problems (or the class $\minnl$ of all minimization NLO problems) as shown in  Section \ref{sec:completeness}. More specifically, we manage to demonstrate that {\sc Min Path-Weight} is indeed complete for $\minnl$. A similar situation is observed also for $\apxl_{\nlo}$. In contrast, the class $\lo_{\nlo}$ of $\dl$-solvable NLO problems possesses complete problems. When we limit our attention to polynomially-bounded NLO problems, each of $\nlo$, $\apxl_{\nlo}$, $\lo_{\nlo}$, $\apxnc{1}_{\nlo}$, and $\nco{1}_{\nlo}$ actually owns complete problems (Section \ref{sec:PBP}).

Among the aforementioned refined classes, we shall  also prove relationships concerning collapses and separations in Section \ref{sec:complexity-OP}.
If we limit our optimization problems onto $\nlo$, then $\po_{\nlo}$,   $\ptas_{\nlo}$, $\apxp_{\nlo}$, and $\aco{1}_{\nlo}$ all coincide with $\nlo$ (Lemmas \ref{PO=APXP=NLO} and \ref{char-with-AC1}(1)).
For polynomially-bounded $\nlo$ problems, in contrast, we can characterize them in terms of $\lo$ problems if their underlying log-space Turing machines are further allowed to access $\nl$ oracles.
Following \cite{Tan07}, $\dl\neq\nl$ if and only if the polynomially-bounded subclasses of $\nlo$, $\apxl_{\nlo}$, $\lsas_{\nlo}$, and $\lo_{\nlo}$ are all distinct (Theorem \ref{refined-relation}(1)). Similarly, we can show that $\nc{1}\neq\dl$ if and only if  the polynomially-bounded subclasses of  $\lo_{\nlo}$, $\apxnc{1}_{\nlo}$, $\ncas{1}_{\nlo}$, and $\nco{1}_{\nlo}$ are all different  (Theorem \ref{refined-relation}(2)). For much lower-complexity optimization problems, we can separate $\aco{0}_{\nlo}$, $\acas{0}_{\nlo}$,  $\apxac{0}_{\nlo}$, and $\nco{1}_{\nlo}$ one from another (Theorem \ref{AC0-separation}). Those separations directly follow from  the well-known separation $\ac{0}\neq \nc{1}$ \cite{Ajt83,FSS84}.

To help the reader overview intrinsic relationships among the aforementioned optimization (complexity) classes, we include Figure~\ref{fig:inclusion-map}, which illustrates class containments and class separations obtained in Section \ref{sec:complexity-OP}. The last section provides a short list of open problems. 

%%%
%%%
%\ignore{
\begin{figure}[t]
 \begin{center}
 \includegraphics[width=9.5cm]{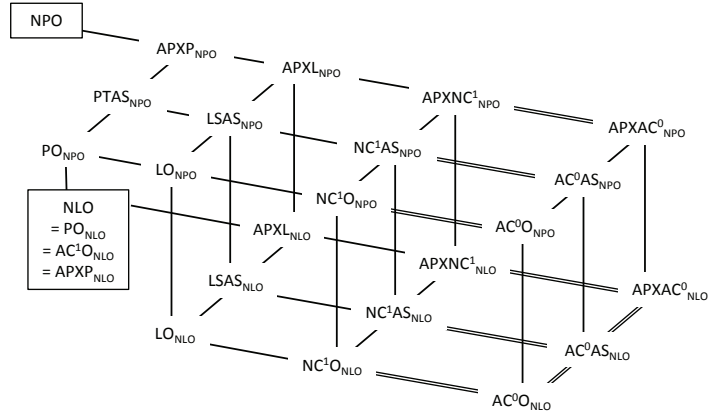}
 \end{center}
\caption{A hierarchy of optimization and approximation classes. Single lines indicate simple class inclusions, whereas double lines indicate proper inclusions.}\label{fig:inclusion-map}
\end{figure}
%}
%%%
%%%

In the subsequent section, we shall  provide a set of basic terminology on the approximation complexity of optimization problems.

%%%%%%%%%%%%%%%%%
\section{Optimization and Approximation Preliminaries}\label{sec:preliminaries}

We aim at {\em refining} an existing framework for studying combinatorial optimization problems of, in particular, low computational complexity. Throughout this paper, the notation $\nat$ denotes the set of all {\em natural numbers} (i.e., nonnegative integers) and $\nat^{+}$ indicates $\nat-\{0\}$. Moreover, $\rational$ (resp., $\real$) indicates the set of all {\em rational numbers} (resp., {\em real numbers}). Two special notations $\rational^{>1}$ and $\rational^{\geq1}$ respectively express the sets $\{q\in\rational\mid q>1\}$ and $\{q\in\rational\mid q\geq 1\}$.
Given two numbers $m,n\in\nat$ with $m\leq n$, an {\em integer interval}  $[m,n]_{\integer}$ is the set $\{m,m+1,m+2,\ldots,n\}$. In the case of $[1,n]_{\integer}$ for $n\geq1$, we abbreviate it as $[n]$.
A {\em(multi-variate) polynomial} is always assumed to have nonnegative integer coefficients. We also  assume that all logarithms are {\em  to base $2$}.

For any set $A$, $\PP(A)$ denotes the {\em power set} of $A$, i.e., the set of all subsets of $A$. Given two sequences $\SSS=(s_1,s_2,\ldots,s_k)$ and $\SSS'=(t_1,t_2,\ldots,t_m)$, the notation $\SSS*\SSS'$ denotes a concatenated sequence $(s_1,s_2,\ldots,s_k,t_1,t_2,\ldots,t_m)$.

An {\em alphabet} is a finite nonempty set of ``symbols'' and a {\em string} (or a {\em word}) over alphabet $\Sigma$ is a finite series of symbols taken from $\Sigma$. In particular, the {\em empty string} is denoted $\lambda$. Let $|x|$ denote the {\em length} of string $x$. The set $\Sigma^*$ is composed of all strings over $\Sigma$ and $\Sigma^{+}$ denotes $\Sigma^*-\{\lambda\}$. A {\em language} over $\Sigma$ is a subset of $\Sigma^*$.  Given two languages $A$ and $B$, their {\em disjoint union} $A\oplus B$ is the set $\{0x\mid x\in A\}\cup \{1x\mid x\in B\}$.
Given each number $n\in\nat^{+}$ and $\Sigma=\{0,1\}$, $bin(n)$ represents a string $w$ in $1\Sigma^*$ ($=\{1x\mid x\in\Sigma^*\}$) that represents $n$ in binary. Additionally, we set $bin(0)=\lambda$. For example, we obtain $bin(1)=1$, $bin(2)=10$,  $bin(5)=101$, and $bin(7)=111$. Note that $|bin(n)|=\ceilings{\log_2(n+1)}$ for every number $n\in\nat^{+}$. By the contrary, for any string $w$ in $\{\lambda\}\cup 1\Sigma^*$, $rep(w)$ denotes a positive integer satisfying $bin(n)=w$. For a number $n\in\nat^{+}$, $bin(n)^{(-)}$ denotes the binary string obtained from $bin(n)$ by removing its first bit ``$1$.'' We also set $bin(0)^{(-)}=\lambda$.
A function $f:\Sigma_1^*\to\Sigma_2^*$ (resp., $f:\Sigma_1^*\to\real^{\geq0}$) for two alphabets $\Sigma_1$ and $\Sigma_2$ is {\em polynomially bounded} if there exists a polynomial $p$ satisfying $|f(x)|\leq p(|x|)$ (resp., $f(x)\leq p(|x|)$) for all inputs $x\in\Sigma_1^*$.

A{\em  directed graph} is a pair $(V,E)$ for which $V$ is a finite set of vertices and $E$ is a binary relation on $V$ and each element $(a,b)$ in $E$ is an edge. An {\em undirected graph} is similarly defined but $E$ is required to be a symmetric relation, and we treat $(a,b)$ and $(b,a)$ in $E$ equivalently.
Given any graph $G=(V,E)$ with vertex set $V$ and edge set $E$, a {\em path} of $G$ is a sequence $(v_1,v_2,\ldots,v_k)$ of vertices in $V$ satisfying that $(v_i,v_{i+1})$ is an edge for every index $i\in[k-1]$. Such a path is called {\em simple} exactly when there are no repeated vertices in it; that is, a simple path has no loop. The {\em length} of a path is the number of edges in it. A {\em tree} is an undirected connected graph with no cycle whereas a {\em forest} is an acyclic undirected graph. A {\em weighted graph} is a graph $G=(V,E)$ for which each edge (or vertex) has an associated weight given by a weight function $w:E\to\real$ (or $w:V\to\real$).

A {\em Boolean formula} is made up of (Boolean) variables and three logical connectives: $\wedge$ (AND), $\vee$ (OR), and $\neg$ (NOT) {\em in infix notation}.
Given a (Boolean) truth assignment $\sigma$, which maps $\{T,F\}$ to variables, a Boolean formula $\phi$ is said to be {\em satisfied} by $\sigma$ if $\phi$ is evaluated to be true after $\sigma$ assigns truth values to the variables in the formula.

%%%%%%%
\paragraph{Representation of Graphs, Matrices, Circuits, and Boolean Formulas:}
When we consider weak computations, it is often critical to choose what types of representation of input instances, such as graphs, matrices, circuits, and Boolean formulas.
For example, as noted in \cite{JLM97}, if we describe trees and forests using {\em bracketed expressions} as part of inputs, then the connectivity problem between two designated nodes in a given forest becomes solvable even on $\nc{1}$-circuits.
With respect to logarithmic-space computation, however, the representation of graphs via incidence matrices, adjacency matrices, or sets of ordered pairs are all ``equivalent'' \cite{JLL76}.  Unless otherwise specified, we assume that every graph is expressed by a {\em listing of its edge relation}, such as $\{(v_1,v_3),(v_2,v_5),(v_5,v_4)\}$; namely, all ordered pairs of vertices that define edges of a given graph. For isolated vertices, we list them as the names of those vertices, such as $\{(v_3),(v_6)\}$, instead of $\{(v_3,v_3),(v_6,v_6)\}$, which indicate self-loops in directed graphs. Boolean circuits are viewed as directed acyclic graphs. Boolean formulas are expressed in infix notation.\footnote{Boolean formulas in infix notation are defined inductively as follows: (i) $0$ and $1$ are Boolean formulas and (ii) if $\alpha$ and $\beta$ are Boolean formulas, then $(\neg\alpha)$, $(\alpha\vee\beta)$, and $(\alpha\wedge\beta)$ are Boolean formulas.}

%%%%%
\subsection{Basic Models of Computation}\label{sec:basic_model}

As a mechanical model of computation, we shall  use the following basic form of {\em (multi-tape) deterministic Turing machine}. For the formal definition of Turing machine, refer to, e.g., \cite{DK00,HMU01}. Our machine is equipped with a read-only input tape, multiple work tapes, and possibly an output tape.
An input $x$ of length $n$ is given on the input tape, surrounded by two endmarkers: $\cent$ (left endmarker) and $\dollar$ (right endmarker) and all input tape cells are consecutively indexed by integers between $0$ and $n+1$, where $\cent$ is at cell $0$ and $\dollar$ at cell $n+1$.
At any moment, a tape head working on the input tape either stays still on the same tape cell or moves to the left or the right.
The {\em running time} (or {\em runtime}) of a Turing machine is the total number of steps (or moves) taken by the machine starting with the input, whereas its {\em (tape) space} is the maximum number of distinct tape cells visited by a tape head during the
machine's  computation.

The behaviors of tapes and their tape heads are quite important in this paper; thus, we wish to pay our special attention to the following terminology. A tape is said to be {\em read-once} if it is a read-only tape and its tape head does not scan the same cell more than once; namely, it either stays at the same cell without reading any information (known as a {\em $\lambda$-move} or an {\em $\varepsilon$-move}) or moves instantly to the right cell.
In contrast, a {\em write-only} tape indicates that, whenever its tape head writes a nonempty symbol in a tape cell, the head should move immediately to its right cell. In this paper, ``output tapes'' are always assumed to be write-only tapes. Turing machines with write-only output tapes are considered to compute {\em (multi-valued partial) functions}, by viewing strings left on the output tapes (when the machines halt) as ``outputs.''

To describe low-complexity classes, we also use a notion of ``random access'' input tapes. In this mode, a machine is further equipped with an index tape and tries to write on this index tape a string of the form $bin(k)$ for a certain number $k\in[0,n+1]_{\integer}$. Whenever the machine enters a specific inner state (an input-query state), the input-tape head jumps in a single step to the cell indexed by $k$ and reads a symbol written in this particular cell. If $k$ is not in the range $[0,n+1]_{\integer}$, then $M$ simply reads a blank symbol as an ``out of range'' symbol. Whenever we need to clarify a use of this special model, we refer to it as {\em random-access Turing machines}.

To express ``nondeterminism'' in our framework, we introduce a special tape called a {\em read-once auxiliary input tape} and equip Turing machines with such auxiliary tapes.
An {\em auxiliary Turing machine} is the above-mentioned deterministic Turing machine equipped with an extra read-once auxiliary input tape on which a sequence of (nonblank) symbols (called an {\em auxiliary input}) is provided as an extra input (other than an ordinary input given on the input tape).
In the rest of this paper, we shall understand that ``auxiliary tapes'' means read-only auxiliary input tapes unless otherwise stated.
Such an auxiliary input given on the auxiliary input tape is surrounded by the two endmarkers. This machine can therefore read off two symbols (except for work-tape symbols) at once, one of which is from the input tape and the other from the auxiliary tape at each step in order to make a deterministic move. As our convention, when a tape head on the auxiliary tape reaches $\dollar$, the head must remain at this endmarker in the rest of a computation.

More formally, a $(k+2)$-tape auxiliary Turing machine $M$ is a tuple $(Q,\Sigma,\{\cent,\dollar\},\Gamma,\Theta,\Phi,q_0,q_{acc},q_{rej})$, where $Q$ is a finite set of inner states, $\Sigma$ is an input alphabet, $\Gamma$ is a work alphabet, $\Theta$ is an auxiliary input alphabet, $\Phi$ is an output alphabet, $q_0$ is the initial state in $Q$, $q_{acc}$ (resp., $q_{rej}$) is an accepting (resp., a rejecting) state in $Q$, and $\delta$ is a transition function from $(Q-\{q_{acc},q_{rej}\})\times (\Sigma\cup\{\cent,\dollar,\lambda\})\times \Gamma^k \times (\Theta\cup\{\dollar\})$ to $Q\times \Gamma^k \times (\Phi\cup\{\lambda\}) \times D\times D_1 \times \cdots \times D_{k}$, where $D$ and each $D_i$ ($i\in[k]$) are sets of head directions, $\{-1,0,+1\}$, of an input tape and the $i$th work tape. Notice that, since tape heads on an auxiliary tape and an output tape move in one direction, we do not need to include their head directions.

We say that an auxiliary Turing machine {\em uses log space} if there exist two constants $a,b>0$ for which, on every input $x$ and every auxiliary input $y$, $M$ uses the total of at most $a\log{|x|}+b$ cells of all work tapes (where an auxiliary input tape is not a work tape). Such a machine is succinctly called a {\em log-space auxiliary Turing machine}. Similarly, we define the notion of {\em polynomial-time auxiliary Turing machine}.

We assume that the reader is familiar with the foundation of computational complexity theory, in particular,  the definitions and properties of those   fundamental classes. See, e.g., \cite{DK00} for their fundamental properties.
The complexity class $\p$ (deterministic polynomial time) is composed of all decision problems (or languages) solved by deterministic Turing machines in polynomial time, whereas $\dl$ (deterministic logarithmic space) contains decision problems solved by log-space deterministic Turing machines.
The notation $\fp$ (resp., $\fl$) refers to a functional version of $\p$ (resp., $\dl$), provided that all functions in $\fl$ output only strings of size polynomial in the lengths of inputs. Thus, $\fl\subseteq \fp$ holds.

For later convenience, we denote by $\auxp$ (resp., $\auxl$) the collection of all sets $A \subseteq \Sigma^*\times\Sigma^*$ over alphabet $\Sigma$ for which there exist a polynomial $p$ and a polynomial-time (resp., log-space)  auxiliary Turing machine $M$ such that,  for every $x$ and $y$, (i) $(x,y)\in A$ implies $|y|\leq p(|x|)$ and (ii) whenever $|y|\leq p(|x|)$, $M$ accepts $(x,y)$ iff $(x,y)\in A$, where $y$ is given on $M$'s auxiliary tape. Their  functional versions with polynomially-bounded outputs (i.e., the size of output strings is bounded from above by a suitable polynomial in the input size) are denoted by $\auxfp$ (resp., $\auxfl$). These classes $\auxp$ and $\auxl$ are respectively associated with nondeterministic classes $\np$ and $\nl$ in the following fashion. Given a set $A\subseteq\Sigma^*\times\Sigma^*$ and any polynomial $p$, let $A_p = \{(x,y)\in A\mid |y|\leq p(|x|)\}$ and $A_p^{\exists} =\{x\in\Sigma^*\mid \exists y\;[(x,y)\in A_p]\}$. When $p$ is clear from the context, we tend to drop subscript ``$p$'' and write $A^{\exists}$ instead of $A^{\exists}_{p}$.
The nondeterministic class $\np$ (resp., $\nl$) is composed of all languages of the form $A^{\exists}_{p}$ for all $A\subseteq \Sigma^*\times\Sigma^*$ and all polynomials $p$ satisfying $A_p\in\auxp$ (resp., $\auxl$). In other words, $A_{p}\in\auxp$ (resp., $\auxl$) if and only if $A_p^{\exists}\in \np$ (resp., $\nl$).

In addition, the notation $\dlogtime$ is used to express the collection of all languages recognized by random-access Turing machines in $O(\log{n})$ time.
A function $f:\Sigma_1^*\to\Sigma_2^*$ is {\em DLOGTIME-computable} if the output size of $f$ is polynomially bounded and the language $A_f=\{(x,i,b)\mid \text{ the $i$th bit of $f(x)$ equals $b$}\}$ belongs to $\dlogtime$.

In the subsequent sections, we shall concentrate mostly on functions and languages (which can be viewed as Boolean functions) whose domains are limited to certain subsets $I$ of $\Sigma^*$ (for alphabets $\Sigma$), and thus any given input to those functions and languages are always assumed, as a ``promise,'' to be taken from  those domains $I$. Our functions and languages are therefore promise problems. To simplify our discussion in the later sections, however, we tend to teat those  promise problems $F$ as if they have no promise and we explicitly write, e.g.,  $F\in \fl$ and $F\in\auxl$ unless there is no confusion.

To describe circuit-based complexity classes, we use a standard notion of {\em Boolean circuits} (or just {\em circuits}), which is a labeled acyclic directed graph whose nodes of indegree $0$ are called inputs and the other nodes are called gates. In our setting, a circuit is made up only of two basic gates $AND$ and $OR$ with inputs, which are labeled by {\em literals} (that is, either Boolean variables or their negations). A {\em fan-in} of a gate is the number of incoming edges. A fan-in is said to be {\em bounded} (resp., {\em unbounded}) if it is smaller than or equal to $2$ (resp., it has no upper bound).
The {\em size} of a circuit is the number of its nodes and the {\em depth} is the number of the longest path from an input to an output.
A {\em family of circuits} is a set $\{C_n\mid n\in\nat\}$, where each $C_i$ is a Boolean circuit with $n$ distinct variables.

There have been a number of uniformity notions proposed in the past literature,  e.g., \cite{BIS90,DHR97,Ruz81}. The different choice of uniformity endows circuit families with (possibly) different computational power.
To explain such uniformity, we define the {\em direct connection language} of a circuit family $\{C_n\}_{n\in\nat}$ as a set of all tuples $\pair{t,a,b,y}$, where $a$ and $b$ are numbers of nodes in $C_n$, $b$ is a child of $a$, $t$ is the type (e.g., literals, $AND$,   $OR$, $NOT$, etc.) of $a$, and $y$ is any string of length $n$.
The {\em standard encoding} of $C_n$ is a string, each symbol of which is of the form $(a,t,b_{L},b_{R})$, where $a,B_{L},B_{R}$ are gate numbers, $b_{L}$ (resp., $B_{R}$) is the left (resp., right) child of $a$, and $t$ is the type of $a$.

A family $\{C_n\}_{n\in\nat}$ of Boolean circuits is {\em log-space uniform} (or {\em L-uniform}) if there exists a log-space deterministic Turing machine computing a function that maps $1^n$ to the standard encoding of $C_n$.
We say that a family $\{C_n\}_{n\in\nat}$ of Boolean circuits is {\em DLOGTIME-uniform}\footnote{As shown in \cite[Theorem 9.1]{BIS90}, this definition is equivalent to the one used in \cite{Bus87,BIS90} using formula languages.}
if the directed connection language of $\{C_n\}_{n\in\nat}$ can be recognized by a log-time random-access Turing machines. Other uniformity notions include {\em $U_{E^*}$-uniformity} and
{\em P-uniformity} \cite{Ruz81}.
For each $k\in\nat$, $\nc{k}$ (resp., $\ac{k}$) denotes the class of decision problems (or  languages) solvable by $\dlogtime$-uniform families of bounded (resp., unbounded) fan-in Boolean circuits of polynomial size and $O(\log^{k}n)$ depth. To refer to $\nc{1}$ of different uniformity, when clarification is necessary, we tend to describe it as ``$\dl$-uniform $\nc{1}$''
or ``$\p$-uniform $\nc{1}$.''
To describe their functional versions, we intentionally use the notations $\fac{k}$ and $\fnc{k}$, respectively. It is known that $\mathrm{ALOGTIME}$ (alternating logarithmic time) coincides with $\dlogtime$-uniform $\nc{1}$ \cite{Bus87}, which also equals $\nc{1}$-uniform $\nc{1}$ \cite{BIS90}.

Another characterization of $\nc{1}$ is given in \cite{BIS90} as follows. The {\em formula language} of a Boolean formula family $\{F_n\}_{n\in\nat}$ is composed of all tuples $\pair{c,i,y}$ such that $|y|=n$ and the $i$th character of the $n$th formula $F_n$ is $c$. A language $A$ is in $\nc{1}$ iff there exists a family $\{F_n\}_{n\in\nat}$ of Boolean formulas with depth $O(\log{n})$ such that (i) for every $x$, $F_{|x|}(x)$ is true exactly when $x\in A$ and (ii) there exists a log-time deterministic Turing machine recognizes the formal language of $\{F_n\}_{n\in\nat}$.

Known inclusion relationships among the aforementioned complexity classes are shown as: $\nc{0}\subsetneqq \ac{0}\subsetneqq \tc{0} \subseteq \nc{1}\subseteq \dl \subseteq \nl = \co\nl \subseteq \ac{1} \subseteq \nc{2} \subseteq\p\subseteq\np$. For more details, refer to, e.g., \cite{DK00}.

It is important to note that, on an output tape of a machine, a natural number is represented in binary, where the least significant bit is always placed at the right end of the output bits.
In the rest of paper, a generic but informal term of ``algorithm'' will be often used to refer to either a deterministic Turing machine or a uniform family of circuits.

\begin{lemma}\label{operation-complexity}
For $n$-bit numbers $x,y,x_i\in\nat$ with $i\in[n]$, the operations $x+y$ and $\max\{0,x-y\}$ are in $\ac{0}$, and $\floors{x/y}$, $\sum_{i=1}^{n}x_i$, and $\prod_{i=1}^{n}x_i$ are in $\tc{0}$ \cite{Hes01,HAB02}.
\end{lemma}

%%%%%
\subsection{Refined Optimization Problems}\label{sec:comb-OPs}

An {\em optimization  problem} is simply a search problem, in which we are asked to look for a best possible feasible solution of the problem for each given admissible input.
In the past literature, NP optimization problems have been a centerfold of the intensive study and low-complexity optimization problems have been mostly neglected except for \cite{Tan07}. To deal with those problems, we intend to refine the existing framework of NP optimization problems in terms of log-space and uniform-circuit computations.

In what follows, we shall formally introduce 14 different classes of refined combinatorial optimization problems, including 4 well-known classes $\npo$, $\po_{\npo}$, $\apxp_{\npo}$, and $\ptas_{\npo}$, in order to justify the correctness of our definitions.

\paragraph{NPO and NLO.}
As a starting point of our study, we formally introduce  {\em NP optimization problems} or {\em NPO problems} in the style of  \cite{ACG+03}.
Since our purpose is to investigate low-complexity optimization problems,
it is better for us to formulate a notion of $\npo$ problems using auxiliary Turing machines instead of nondeterministic Turing machines. An $\npo$ problem $P$ is formally a quadruple $(I,SOL,m,goal)$ whose entries satisfy the following properties.

\s
\renewcommand{\labelitemi}{$\circ$}
\begin{itemize}\vs{-2}
  \setlength{\topsep}{-2mm}%
  \setlength{\itemsep}{0mm}% original = 1mm
  \setlength{\parskip}{0cm}%
\item $I$ is a finite set of {\em admissible instances}. There must be a deterministic Turing machine that recognizes $I$ in polynomial time; that is, $I$ belongs to $\p$.

\item $SOL$ is a function mapping $I$ to a collection of certain finite sets, where $SOL(x)$ is a set of {\em feasible solutions} of input instance $x$. There must be a polynomial $q$ such that (i) for every $x\in I$ and every $y\in SOL(x)$, it holds that $|y|\leq q(|x|)$ and (ii) the set $I\circ SOL = \{(x,y) \mid x\in I, y\in SOL(x)\}$ is in $\auxp$; namely, $I\circ SOL$ is recognized in time polynomial in $|x|$ by a certain auxiliary Turing machine stating with $x$ on an input tape and $y$ on an auxiliary tape. By the definition of $SOL$, the set $\{x\in I\mid SOL(x)\neq\setempty\}$ matches $(I\circ SOL)_{q}^{\exists}$, and thus it belongs to $\np$.

\item $goal$ is either {\sc max} or {\sc min}. When $goal=\text{\sc max}$, $P$ is called a {\em maximization problem}; when $goal=\text{\sc min}$, it is a {\em minimization problem}.

\item $m$ is a {\em measure function} (or an {\em objective function}) from $I\circ SOL$ to $\nat^{+}$ whose value $m(x,y)$ is computed  in time polynomial in $|x|$ by a certain auxiliary Turing machine starting with $x$ written  on an input tape and $y$ on an auxiliary tape. Technically speaking, $m$ is a promise problem; however, by abusing notations, we often express $m$ as a member of $\auxfp$ (i.e., $m\in\auxfp$). For any instance $x\in I$,  $m^*(x)$ denotes the optimal value $goal\{m(x,y)\mid y\in SOL(x)\}$. Moreover, $SOL^*(x)$ expresses the  set $\{y\in SOL(x)\mid m(x,y) = m^*(x)\}$ of optimal solutions of $x$.
\end{itemize}

Notice that, in polynomial time, an auxiliary
Turing machine can copy any string $y$ given on an auxiliary tape into its work tape and then manipulate it freely. This makes the read-once requirement of an auxiliary tape redundant. Therefore, the above definition logically matches
the existing notion of $\npo$ problems in, e.g., \cite{ACG+03}.
Let the notation $\npo$ also express the class of all $\npo$ problems.

A measure function $m$ is called {\em polynomially bounded} if there exists a polynomial $p$ such that $m(x,y)\leq p(|x|,|y|)$ holds for all pairs $(x,y)\in I\circ SOL$. An optimization problem is also said to be {\em polynomially bounded} if its measure function is polynomially bounded. For convenience,  a succinct notation $\pbo$ indicates the collection of all optimization problems that are polynomially bounded.

To analyze the behaviors of low-complexity optimization problems,  Tantau \cite{Tan07} formulated a notion of {\em NL optimization problems} (or {\em NLO problems}, in short), which are obtained simply by replacing the term ``polynomial time'' in the above definition of $\npo$ problems with ``logarithmic space.''
For those $\nlo$ problems, the use of auxiliary Turing machine is essential and it may not be replaced by any Turing machine having no read-once auxiliary input  tapes.

Here, we draw our attention to the read-once requirement posed on an auxiliary input tape. This requirement is quite severe for Turing machines. To see this fact, let us consider the following maximization problem {\sc Max Weight-2SAT}. In the {\em maximum weighted 2-satisfiability problem} ({\sc Max Weight-2SAT}), we seek a truth assignment $\sigma$ satisfying  a given 2CNF formula on a set $X$ of variables and a variable weight function $w:X\to\nat^{+}$ such that the sum $\sum_{x\in X}\sigma(x)w(x)+1$ must be maximized.
Although its associated decision problem 2SAT, in which we are asked to decide whether a given 2CNF formula is satisfiable, is NL-complete (from a result of \cite{JLL76}), it is not clear whether {\sc Max Weight-2SAT} belongs to $\nlo$.

To express the class of all $\nlo$ problems, we use the notation of $\nlo$. It follows that $\nlo\subseteq \npo$. Moreover, $\minnl$ (resp., $\maxnl$) denotes the class of all minimization (resp., maximization) problems in $\nlo$; thus, $\nlo$ equals the union $\minnl\cup \maxnl$.

%%%%%%%%
\paragraph{PO, LO, NC$^{i}$O, and AC$^{i}$O.}
We say that an $\npo$ problem $P = (I,SOL,m,goal)$ is {\em P-solvable} if there exists a polynomial-time deterministic Turing machine $M$ such that, for every instance $x\in I$, if $SOL(x)\neq\setempty$, then $M$ returns an optimal solution $y$ in $SOL(x)$  and, otherwise, $M$ returns ``no solution'' (or a designated symbol $\bot$). Moreover, the values $m(x,M(x))$ ($=m^*(x)$) must be computed in polynomial time from inputs $x$.
As a result, the set $\{x\in I\mid SOL(x)\neq\setempty\}$ must be in $\p$.
Given a class $\DD$ of optimization problems, the notation $\po_{\DD}$ expresses the class of all optimization problems in $\DD$ that are $\p$-solvable. Similarly, we can define the notations of $\lo_{\DD}$, $\nco{i}_{\DD}$, and $\aco{i}_{\DD}$ by replacing the term ``$\p$-solvable'' with ``$\dl$-solvable,'' ``$\nc{i}$-solvable,'' and ``$\ac{i}$-solvable,'' respectively, for each index $i\in\nat$. Conventionally, $\po_{\npo}$ is written as $\po$ and $\lo_{\nlo}$ is noted briefly as $\lo$ in \cite{Tan07}. Notice that $\aco{0}_{\DD} \subseteq \nco{1}_{\DD} \subseteq \lo_{\DD} \subseteq \po_{\DD}$ for any reasonable class $\DD$.

It is important to note that, as in the case of $\lo_{\npo}$, for example, when a problem $P$ is $\dl$-solvable, its log-space algorithm, say, $M$ that solves $P$ does not need to check whether an input $x$ given to $M$ is actually admissible instance (i.e., $x\in I$), because such a task may be in general impossible for log-space machines. Hence, $P$ is technically a promise problem and we normally allow $M$ to behave arbitrarily on inputs outside of $I$ or $I\circ SOL$.

%%%%%
\paragraph{APXP, APXL, APXNC$^{i}$, and APXAC$^{i}$.}
Next, we shall  define approximation classes using a notion of $\gamma$-approximation. Given an optimization problem $P=(I,SOL,m,goal)$, the {\em performance ratio} of solution $y$ with respect to instance $x$ is defined as
\[
R(x,y) = \max\left\{ \left|\frac{m(x,y)}{m^*(x)}\right|, \left|\frac{m^*(x)}{m(x,y)}\right| \right\},
\]
provided that neither $m(x,y)$ nor $m^*(x)$ is zero. Notice that  $R(x,y)=1$ iff $y\in SOL^*(x)$. Let $\gamma>1$ be a constant indicating an upper bound of the performance ratio.
With this constant $\gamma$, we say that $P$ is
{\em polynomial-time $\gamma$-approximable} if there exists a  polynomial-time deterministic Turing machine $M$  such that, for any instance $x\in I$, if $SOL(x)\neq\setempty$, then $M(x)\in SOL(x)$ and $R(x,M(x))\leq \gamma$; otherwise, $M(x)$ outputs ``no solution'' (or a symbol $\bot$); in addition,
the values\footnote{The polynomial-time computability of the value $m(x,M(x))$ is trivial; however, the computability requirement for this value is quite important for the log-space computability and the NC$^{1}$ computability.} $m(x,M(x))$ must be computed in polynomial time from inputs $x$.
Such a machine is referred to as a {\em $\gamma$-approximate algorithm}. The $\gamma$-approximability clearly implies that the set  $\{x\in I\mid SOL(x)\neq\setempty\}$ belongs to $\p$.
The notation $\apxp_{\DD}$ denotes a class consisting of problems $P$ in   class $\DD$ of optimization problems such that, for a certain fixed constant  $\gamma>1$, $P$ is polynomial-time $\gamma$-approximable. Notice that $\apxp_{\npo}$ is conventionally expressed as $\apx$ (see, \eg \cite{ACG+03}).

Likewise, we define three extra notions of ``log-space $\gamma$-approximation'' \cite{Tan07}, ``$\nc{i}$ $\gamma$-approximation,''  and ``$\ac{i}$ $\gamma$-approximation'' by replacing ``polynomial-time Turing machine'' in the above definition with ``logarithmic-space (auxiliary) Turing machine,'' ``uniform family of $\nc{i}$-circuits,'' and ``uniform family of $\ac{i}$-circuits,'' respectively, for every index $i\in\nat$.
We then introduce the notations of $\apxl_{\DD}$, $\apxnc{i}_{\DD}$, and $\apxac{i}_{\DD}$ using ``log-space $\gamma$-approximation,'' ``$\nc{i}$ $\gamma$-approximation,'' and ``$\ac{i}$ $\gamma$-approximation,''  respectively. It follows that $\apxac{0}_{\DD} \subseteq \apxnc{1}_{\DD} \subseteq \apxl_{\DD} \subseteq \apxp_{\DD}$ for any reasonable optimization/approximation class $\DD$.

%%%%%
\paragraph{PTAS, LSAS, NC$^{i}$AS, and AC$^{i}$AS.}
A deterministic Turing machine $M$ is called a {\em polynomial-time approximation scheme} (or a PTAS) if, for any ``fixed constant'' $r\in\rational^{>1}$, there exists a polynomial $p_r(n)$ such that, for every admissible instance $x\in I$, if $SOL(x)\neq\setempty$, then $M$ takes $(x,r)$ as its input and outputs an $r$-approximate solution of $x$ in time at most $p_r(|x|)$; otherwise, $M(x)$ outputs ``no solution'' (or a symbol $\bot$).
Examples of such polynomial $p_r(n)$ are   $\ceilings{\frac{r}{r-1}}n^3$ and $n^{\ceilings{1/(r-1)}}$.
Any approximation scheme is also a $\gamma$-approximate algorithm for any chosen constant $\gamma>1$. The approximation class $\mathrm{PTAS}_{\npo}$ denotes a collection of all NPO problems that admit PTAS's. In a similar manner, we can define a notion of {\em logarithmic-space approximation scheme} (or LSAS) and the associated approximation class $\mathrm{LSAS}_{\nlo}$ by replacing ``polynomial time'' and ``polynomial'' with ``logarithmic space'' and ``logarithmic function,'' respectively.

The definitions of $\nc{i}\mathrm{AS}_{\nlo}$ and $\ac{i}\mathrm{AS}_{\nlo}$ are given essentially in the same way with a slight technical complication on uniformity condition. $\nc{i}\mathrm{AS}_{\nlo}$ (resp., $\ac{i}\mathrm{AS}_{\nlo}$) can be introduced using circuits of size $p_r(n)$ and depth $\ell_r(n)$ with bounded (resp., unbounded) fan-in gates, where $p_r(n)$ is a polynomial and $\ell_r(n)$ is a logarithmic function as long as $r$ is treated as a fixed constant. Here, the uniformity requires $\mathrm{DTIME}(\ell'_r(n))$ for another logarithmic function $\ell'_r$ with $r$ being treated as a constant.  
\ms

We have so far given 14 classes of optimization problems, which we shall discuss in details in the subsequent sections. Given an arbitrary nonempty class $\DD$ of optimization problems, it holds that  $\nco{1}_{\DD}\subseteq \lo_{\DD} \subseteq \po_{\DD}$ and  $\apxnc{1}_{\DD}\subseteq \apxl_{\DD}\subseteq \apxp_{\DD}$. It also follows that $\nco{1}_{\DD} \subseteq \ncas{1}_{\DD} \subseteq \apxnc{1}_{\DD}$, $\lo_{\DD} \subseteq \lsas_{\DD}  \subseteq \apxl_{\DD}$, and $\po_{\DD} \subseteq \ptas_{\DD} \subseteq \apxp_{\DD}$.
When $\DD=\nlo$, in particular, three classes $\apxp_{\nlo}$, $\ptas_{\nlo}$, and $\po_{\nlo}$ coincide with $\nlo$. Since the proof of this fact is short, we include it here.

\begin{lemma}\label{PO=APXP=NLO}
$\apxp_{\nlo} = \ptas_{\nlo} = \po_{\nlo} = \nlo$.
\end{lemma}

\begin{proof}
Note that $\po_{\nlo}\subseteq \ptas_{\nlo}\subseteq \apxp_{\nlo}$. First, we claim that $\apxp_{\nlo}\subseteq \nlo$. By the definition of $\apxp_{\nlo}$, all problems in $\apxp_{\nlo}$ must be $\nlo$ problems, and hence they are in $\nlo$.

Next, we show that $\nlo\subseteq \po_{\nlo}$.
Let $P=(I,SOL,m,goal)$ be any problem in $\nlo$.
Here, we consider only the case of $goal=\text{\sc max}$ because the case of {\sc min} is analogous. We want to show that $P$ belongs to $\po_{\nlo}$. Let $x$ be any instance in $I$. Consider the following algorithm on input $x$.
Here, we define $D=\{(x,y)\in I\circ SOL \mid \exists z\in SOL(x)\,[ z \geq y \wedge m(x,z)\geq m(x,y)]\}$, where the notation $\geq$ used for strings $x$ and $y$ is the lexicographic ordering. Note that $D\in\nl\subseteq\p$. Now, we can use a binary search technique using $D$ to find a maximal solution $y\in SOL^*(x)$ in polynomial time.
Therefore, we conclude that $\nlo \subseteq \po_{\nlo}\subseteq \ptas_{\nlo} \subseteq \apxp_{\nlo} \subseteq \nlo$. This implies the lemma.
\end{proof}

%%%%%

Taking a slightly different approach toward a study on $\npo$ problems, Krentel \cite{Kre88} introduced a class $\optp$ of optimization functions. Let $\maxp$ (resp., $\minp$) denote the class of all functions from $\Sigma_1^*$ to $\Sigma_2^*$, each of which satisfies the following property: there exists a polynomial-time nondeterministic Turing machine $M$ such that, for every input $x\in\Sigma_1^*$, $f(x)$ denotes the maximal (resp., minimal) string (in the lexicographic order) generated by $M$ on $x$ \cite{KST89}, where $\Sigma_1$ and $\Sigma_2^*$ are alphabets.
The class $\optp$ is simply defined as $\maxp\cup \minp$. We further define $\optl$ in a similar way but using log-space nondeterministic Turing machines.
Notice that \`{A}lvarez and Jenner \cite{AJ93} originally defined $\optl$ as the set of {\em only} maximization problems and that we need to pay a special attention to their results whenever we apply them in our setting.

%%%%%
\subsection{Approximation-Preserving Reductions}\label{sec:AP-reducibility}

To compare the computational complexity of two optimization problems, we wish to use three types of reductions between those two problems. We follow well-studied reductions, known as {\em approximation-preserving (AP) reductions} and {\em exact (EX) reductions}.
Given two optimization problems $P=(I_1,SOL_1,m_1,goal)$ and $Q=(I_2,SOL_2,m_2,goal)$, $P$ is {\em  polynomial-time AP-reducible} (or more conveniently, {\em APP-reducible}) to $Q$, denoted $P\APreduces^{\p} Q$, if there are two functions $f$ and $g$ and a constant $c\geq1$ such that the following {\em APP-condition} is satisfied:
\renewcommand{\labelitemi}{$\circ$}
\begin{itemize}\vs{-1}
  \setlength{\topsep}{-2mm}%
  \setlength{\itemsep}{0mm}% original = 1mm
  \setlength{\parskip}{0cm}%
\item for any instance $x\in I_1$ and any $r\in\rational^{>1}$, it holds that $f(x,r)\in I_2$,

\item for any $x\in I_1$ and any $r\in\rational^{>1}$, if $SOL_1(x)\neq\setempty$ then $SOL_2(f(x,r))\neq\setempty$,

\item for any $x\in I_1$, any $r\in\rational^{>1}$, and any $y\in SOL_2(f(x,r))$, it holds that $g(x,y,r)\in SOL_1(x)$,

\item $f(x,r)$ is computed by a deterministic Turing machine and $g(x,y,r)$ is computed by an auxiliary Turing machine, both of which run in time polynomial in $(|x|,|y|)$ for any $(x,y)\in I\circ SOL$ and any number  $r\in\rational^{>1}$, and

\item for any $x\in I_1$, any $r\in\rational^{>1}$, and any $y\in SOL_2(f(x,r))$, $R_2(f(x,r),y) \leq r$ implies $R_1(x,g(x,y,r))\leq 1+c(r-1)$, where $R_1$ and $R_2$ respectively express the performance ratios for $P_1$ and $P_2$.
\end{itemize}
Notice that the above APP-condition makes us concentrate only on instances of $\{x\in I\mid SOL(x)\neq\setempty\}$ and that, for other instances $x$, we might possibly set the value $f(x,r)$ arbitrarily (as long as $x\in I_1$ iff $f(x,r)\in I_2$).
When this APP-condition holds, we also say that $P$ {\em APP-reduces}
to $Q$. The triplet $(f,g,c)$ is called a {\em polynomial-time AP-reduction} (or an {\em APP-reduction}) from $P$ to $Q$. For more details, refer to, e.g., \cite{ACG+03}.

To discuss optimization problems within $\po_{\npo}$,
we further need to introduce another type of reduction $(f,g)$, in which $g$ ``exactly'' transforms in polynomial time an optimal solution for $Q$ to another optimal solution for $P$ so that ``$Q\in \po_{\npo}$'' directly implies ``$P\in \po_{\npo}$.'' We write $P\EXreduces^{\p} Q$ when the following {\em EX-condition} holds:
\renewcommand{\labelitemi}{$\circ$}
\begin{itemize}\vs{-1}
  \setlength{\topsep}{-2mm}%
  \setlength{\itemsep}{0mm}% original = 1mm
  \setlength{\parskip}{0cm}%
\item for any instance $x\in I_1$, it holds that $f(x)\in I_2$,

\item for any $x\in I_1$, if $SOL_1(x)\neq\setempty$ then $SOL_2(f(x))\neq\setempty$,

\item for any $x\in I_1$ and any $y\in SOL_2(f(x))$, it holds that $g(x,y)\in SOL_1(x)$,

\item $f(x)$ is computed by deterministic Turing machine and $g(x,y)$ is computed by an auxiliary Turing machine, both of which run in time polynomial in $(|x|,|y|)$, and

\item for any $x\in I_1$ and any $y\in SOL_2(f(x))$, $R_2(f(x),y) =1$ implies $R_1(x,g(x,y)) =1$, where $R_1$ and $R_2$ respectively express the performance ratios for $P_1$ and $P_2$.
\end{itemize}
The above pair $(f,g)$ is called a {\em polynomial-time EX-reduction} (or an {\em EXP-reduction}) from $P$ to $Q$.

It is quite useful to introduce a notion that combines both  $\APreduces^{\p}$ and $\EXreduces^{\p}$. Let us define the notion of {\em polynomial-time strong AP-reduction} (strong APP-reduction or sAPP-reduction), denoted $\sAPreduces^{\p}$, obtained from $\APreduces^{\p}$ by allowing $r$ (used in the above definition of APP-reduction) to be chosen from $\rational^{\geq1}$ (instead of $\rational^{>1}$).

Next, we weaken the behaviors of polynomial-time (strong) APP-reductions by modifying the ``polynomial-time'' requirement imposed on the aforementioned definition of (strong) APP-condition.  When we replace ``polynomial-time'' by ``logarithmic-space,''  ``uniform family of $\nc{1}$-circuits,'' and  ``uniform family of $\ac{0}$-circuits,''  we respectively obtain the corresponding notions of {\em (strong) APL-reduction} ($\APreduces^{\dl}$, $\sAPreduces^{\dl}$),  {\em (strong) APNC$^1$-reduction} ($\APreduces^{\nc{1}}$, $\sAPreduces^{\nc{1}}$), and {\em (strong) APAC$^0$-reduction} ($\APreduces^{\ac{0}}$, $\sAPreduces^{\ac{0}}$). Notice that the notion of error-preserving reduction (or E-reduction), which was used in \cite{Tan07}, essentially matches sAPL-reduction. Likewise, we define {\em EXL-reduction} ($\EXreduces^{\dl}$), {\em EXNC$^{1}$-reduction} ($\EXreduces^{\nc{1}}$), and {\em EXAC$^{0}$-reduction} ($\EXreduces^{\ac{0}}$) from EXP-reduction.

The following two lemmas are immediate from the definition of sAP-reductions and we omit their proofs.

\begin{lemma}\label{reduction-transitive}
For any two reduction type $e_1,e_2\in\{\p,\dl,\nc{1},\ac{0}\}$, if $e_1\subseteq e_2$ (seen as complexity classes), then $P_1\sAPreduces^{e_1}P_2$ implies $P_1\sAPreduces^{e_2}P_2$. The same statement holds for $\APreduces^{e}$ and $\EXreduces^{e}$.
\end{lemma}

\begin{lemma}
For any reduction type $e\in\{\p,\dl,\nc{1},\ac{0}\}$, $P_1\sAPreduces^{e} P_2$ implies both $P_1\APreduces^{e} P_2$ and $P_1\EXreduces^{c} P_2$.
\end{lemma}

In the next lemma, we shall present a useful property, called a {\em downward closure property}, for $\APreduces^{\dl}$- and $\EXreduces^{\dl}$-reductions. A similar property holds also for $\APreduces^{\nc{1}}$- and $\EXreduces^{\nc{1}}$-reductions.

\begin{lemma}\label{downward-property}
[downward closure property]
Let $P$ and $Q$ be any two optimization problems in $\nlo$.
\begin{enumerate}\vs{-2}
  \setlength{\topsep}{-2mm}%
  \setlength{\itemsep}{1mm}%
  \setlength{\parskip}{0cm}%

\item Let $\DD\in\{\nlo,\apxl,\lsas,\apxnc{1},\ncas{1}\}$. If $P\APreduces^{\nc{1}}Q$ and $Q\in \DD_{\nlo}$, then $P\in\DD_{\nlo}$, where $\nlo_{\nlo}$ is understood as $\nlo$.

\item Let $\DD\in\{\lo,\nco{1}\}$. If $P\EXreduces^{\nc{1}}Q$ and $Q\in \DD_{\nlo}$, then $P\in\DD_{\nlo}$.
\end{enumerate}
\end{lemma}

An immediate consequence of Lemma \ref{downward-property} is the following corollary. In comparison, by setting  $\DD\in\{\nlo,\apxl,\lsas,\lo\}$, for any $P,Q\in \nlo$, if $P\sAPreduces^{\dl}Q$ and $Q\in \DD_{\nlo}$, then $P\in\DD_{\nlo}$ \cite{Tan07}.

\begin{corollary}
Let $\DD\in\{\nlo,\apxl,\lsas,\lo\}$. For any $P,Q\in \nlo$, if $P\sAPreduces^{\nc{1}}Q$ and $Q\in \DD_{\nlo}$, then $P\in\DD_{\nlo}$, where $\nlo_{\nlo}$ is $\nlo$.
\end{corollary}

Here, we shall briefly give the proof of Lemma \ref{downward-property}.

\begin{proofof}{Lemma \ref{downward-property}}
Take any two optimization problems $P=(I_1,SOL_1,m_1,goal_1)$ and $Q=(I_2,SOL_2,m_2,goal_2)$ in $\nlo$. In what follows, we shall prove only the case of $\DD=\apxnc{1}_{\nlo}$, because the other cases can be similarly
treated.

(1)  Assume that $P\sAPreduces^{\nc{1}} Q$ via an APNC$^{1}$-reduction $(f,g,c)$ and that $Q$ is in $\apxnc{1}_{\nlo}$. Given any constant $r'>1$, let $C_{r'}$ be an  NC$^{1}$ $r'$-approximate algorithm solving $Q$.
To show that $P\in\apxnc{1}_{\nlo}$, it suffices to construct, for each constant $r>1$, an appropriate NC$^{1}$ circuit, say, $D_r$ that finds $r$-approximate solutions for $P$.

Given a constant $r>1$, let us define $r'=1+(r-1)/c >1$ and consider $C_{r'}$.  Since $C_{r'}$ is an NC$^{1}$ $r'$-approximate algorithm, it follows that the performance ratio $R_2$ for $C_{r'}$ satisfies $R_2(z,C_{r'}(z))\leq r'$ for any $z\in I_2$.
Next, we define the desired algorithm $N_{r}$ as follows: on input $x\in I_1$, compute simultaneously $z= f(x,r)$ and $y = C_{r'}(z)$ and then output $g(x,y,r')$. Since $R_2(z,C_{r'}(z))\leq r'$, it follows by the definition of $\APreduces^{\nc{1}}$ that $R_1(x,N_{r}(x)) = R_1(x,g(x,C_{r'}(z),r'))\leq 1+c(r'-1)=r$. Hence, $N_r$ is an $r$-approximate algorithm for $P$.

We still need to show that $N_{r}$ can be realized by an NC$^{1}$-circuit.
For this purpose, we prepare an NC$^{1}$ circuit $C_f$ that, on input $(x,e)\in \Sigma^*\times 1\{0,1\}^*$,  outputs the $rep(e)$-th bit of $f(x,r)$. Notice that  $|f(x,r)|$ is polynomially bounded.
Moreover, let $C_g$ denote an NC$^{1}$ circuit computing $g$.
We construct an NC$^{1}$-circuit $M'$ that, on input $(x,e)\in \Sigma^*\times 1\{0,1\}^*$, computes the $rep(e)$-th bit of $C_{r'}(f(x,r'))$. During this procedure, whenever $C_{r'}$ tries to access the $j$th bit of $f(x,r')$, we run $C_f$ on $(x,bin(j))$.
The desired algorithm $N_{r}$ is executed as follows. We first run $C_g$ using the first and third input tapes for $x$ and $r'$ and leaving the second tape blank. Whenever $C_g$ tries to access the $i$th bit $y_i$ of $y=C_{r'}(f(x,r'))$, we run $M'$ on $(x,bin(i))$. It is not difficult to show that this procedure can be implemented on an appropriate NC$^{1}$-circuit.

(2) Assume that $P\EXreduces^{\nc{1}} Q$ via an $\EXreduces^{\nc{1}}$-reduction $(f,g)$ with  $Q\in\nco{1}_{\nlo}$.  Since $Q\in\nco{1}_{\nlo}$,
there exists an NC$^{1}$ circuit $M$ for which  $M(x)\in SOL_{2}(x)$ and $R_2(x,M(x))=1$ for any $x\in I_2$. To show that $P$ is in $\nco{1}_{\nlo}$, let us consider the following algorithm $N$. On input $x$, compute $y=M(f(x))$ and output $w=g(x,y)$. For a similar reason to (1), $N$ can be implemented by a certain NC$^{1}$ circuit. Since $R_1(x,N(x))=R_1(x,g(x,y))$, by the definition of an $\EXreduces^{\nc{1}}$-reduction, we obtain $R_2(f(x),M(f(x))) = 1$. Therefore, $N$ exactly solves $P$.
\end{proofof}

Our AP-, EX-, and sAP-reductions can help us identify the most difficult problems in a given optimization/approximation class.
Such problems are generally called ``complete problems,'' which have played a crucial role in understanding the structural features of
optimization and approximation classes.

Formally, let $\leq$ be any reduction discussed in this section, and let  $\DD$ be any class of optimization problems. An optimization problem $P$ is called {\em $\leq$-hard} for $\DD$ if, for every problem $Q$ in $\DD$,  $Q\leq P$ holds. Moreover, $P$ is said to be {\em $\leq$-complete} for $\DD$ if $P$ is in $\DD$ and it is $\leq$-hard for $\DD$.
This completeness will be a central subject in Sections \ref{sec:completeness}--\ref{sec:PBP}.

%%%%%%%%%%%%%%%%%%%%%%%%
%%%%%%%%%%%%%%%%%%%%%%%%
\section{General Complete Problems}\label{sec:completeness}

Complete problems represent a certain structure of a given optimization or approximation class and they provide useful insights into specific features of the class. To develop a coherent theory of NLO problems, it is essential to study such complete problems. In the subsequent subsections, we shall present numerous complete problems for various optimization and approximation classes.

%%%%%%%
\subsection{Why APNC$^{1}$- and EXNC$^{1}$-Reductions?}\label{sec:why-NC1}

To discuss complete problems for refined optimization and approximation classes under certain reductions, it is crucial to choose reasonable types of reductions. By Lemma \ref{reduction-transitive}, for example, any $\APreduces^{\nc{1}}$-complete problem for an optimization/approximation class $\DD$ is also $\APreduces^{\dl}$-complete, but the converse may not be true in general. In what follows, we briefly argue that the  $\APreduces^{\dl}$- and $\EXreduces^{\dl}$-reductions are so powerful that all problems in  $\apxl_{\nlo}$ and $\lo_{\nlo}$ respectively become reducible to similar problems residing even in $\apxac{0}_{\nlo}\cap\pbo$ and $\aco{0}_{\nlo}\cap\pbo$.

\begin{proposition}\label{characterization}
\begin{enumerate}%\vs{-1}
  \setlength{\topsep}{-2mm}%
  \setlength{\itemsep}{0mm}% original = 1mm
  \setlength{\parskip}{0cm}%

\item $\apxl_{\nlo} = \{P\in\nlo\mid \exists Q\in \apxac{0}_{\nlo} \cap\pbo \,[ P\sAPreduces^{\dl} Q]\}$.

\item $\lo_{\nlo} = \{P\in\nlo\mid \exists Q\in \aco{0}_{\nlo} \cap\pbo \,[ P\EXreduces^{\dl} Q]\}$.
\end{enumerate}
\end{proposition}

\begin{proof}
(1) This claim is split into two opposite containments.

($\supseteq$) Let $P\in \nlo$ and $Q\in \apxac{0}_{\nlo}$, and assume that $P\sAPreduces^{\dl}Q$. Notice that $Q$ also belongs to $\apxl_{\nlo}$.
Lemma \ref{downward-property}(1) therefore implies that $P\in\apxl_{\nlo}$.

($\subseteq$) Since  $\apxl_{\nlo} = \apxl_{\maxnl}\cup \apxl_{\minnl}$, we first consider the case of $\apxl_{\maxnl}$. Take any maximization problem $P=(I_1,SOL_2,m_1,\text{\sc max})$ in  $\apxl_{\maxnl}$.  We want to define a new maximization problem $Q = (I_2,SOL_2,m_2,\text{\sc max})$ and show that $Q\in\apxac{0}_{\nlo}$ and $P\sAPreduces^{\dl}Q$.

Since $P\in\apxl_{\nlo}$, there exists a constant $e\in\rational^{>1}$ and a log-space deterministic Turing machine $M$ that produces $e$-approximate solutions of $P$; namely, the performance ratio $R_1$ of $M$'s outcome for $P$ satisfies $R_1(x,M(x))\leq e$ for every $x\in (I_1\circ SOL_1)^{\exists}$.
First, we set $I_2$ to be composed of all instances of the form $(x,M(x))$ for  $x\in I_1$. Notice that, whenever $SOL_1(x)=\setempty$, $M$ outputs the designated symbol $\bot$. Since $M$ uses only log space, $I_2\in\dl$ follows. Next, we define $SOL_2(x,y) = SOL_1(x)$ and $m_2((x,y),z)= m_1(x,z)$ for any $x\in I_1$ and $y,z\in SOL_1(x)$. By those definitions, $Q$ is a problem in $\nlo$.

Let us consider an $\ac{0}$-circuit that outputs $y$ on admissible instance   $(x,y)$ in $(I_2\circ SOL_2)^{\exists}$. For any $(x,y)\in (I_2\circ SOL_2)^{\exists}$, it follows that $R_2((x,y),C(x,y)) = \frac{m_2^*(x,y)}{m_1((x,y),C(x,y))} =  \frac{m_1^*(x)}{m_1(x,y)} = R_1(x,M(x))\leq e$, where $R_2$ means the performance ratio for $Q$.
Hence, $C(x,y)$ is an $e$-approximate solution of $Q$. Thus, $Q$ belongs to $\apxac{0}_{\nlo}$.

Next, we want to show that $P\sAPreduces^{\dl} Q$ via $(f,g,1)$. Take any number $r\in\rational^{\geq1}$ and define $f(x,r)=(x,M(x))$ and $g((x,y),z,r)=z$ for $x\in I_1$ and $y,z\in SOL_1(x)$. It follows that $R_2(f(x,r),z) = \frac{m_1^*(x)}{m_1(x,z)} = \frac{m_1^*(x)}{m_1(x,g((x,y),z,r))} =R_1(x,g((x,y),z,r))$. Since  $f$ is in $\fl$ and $g$ is in $\fac{0}$, $P$ indeed  sAPL-reduces to $Q$. Because $P$ is arbitrary, we conclude that every maximization problem in $\apxl_{\nlo}$ is $\sAPreduces^{\dl}$-reducible to $Q$.

In a similar fashion, we can show that every minimization problem $P'$ in $\apxl_{\minnl}$ can be reduced to a certain minimization problem $Q'$ in $\apxac{0}_{\nlo}$.

(2) This claim can be proven in a similar way to (1).

($\supseteq$) Take two optimization problems  $P\in\nlo$ and $Q\in\aco{0}_{\nlo}$. Assume that $P\EXreduces^{\dl}Q$ via $(f,g)$. Lemma \ref{downward-property}(2) then ensures that $P$ belongs to $\lo_{\nlo}$.

($\subseteq$) We begin with the case of $\lo_{\maxnl}$.
Let $P=(I_1,SOL_1,m_1,\text{\sc max})$ be an arbitrary problem in $\lo_{\maxnl}$ and take a deterministic Turing machine $M$ that solves $P$ using log space.
We intend to construct another problem $Q = (I_2,SOL_2,m_2,\text{\sc max})$ so that  $P$ is $\EXreduces^{\dl}$-reducible to $Q$. For this desired problem $Q$, we set $I_2= I_1\circ SOL_1$ and $SOL_2(x,y)=\{y\}$ for any $(x,y)\in I_2$. The measure function $m_2$ is defined as $m_2((x,y),z)=2$ if $y=z$, and $1$ otherwise. Obviously, $m_2$ is polynomially bounded and is in $\fac{0}$.
For a particular input $(x,M(x))$, since $m_2((x,M(x)),M(x)) = m_2^*(x,M(x))$, we obtain $M(x)\in SOL_2^*(x,M(x))$. Thus, it follows that $Q\in\aco{0}_{\nlo}$. Finally, we define a reduction $(f,g)$ as
$f(x)=(x,M(x))$ and $g(x,z)=M(x)$ for any $z\in SOL_2(f(x))$. Note that $R_2(f(x),z) =1$ implies $z=M(x)$, and thus the performance ratio $R_1(x,g(x,z))$ for $P$ satisfies $R_1(x,g(x,M(x))) = R_1(x,M(x)) =1$. Therefore, $(f,g)$ $\EXreduces^{\dl}$-reduces $P$ to $Q$.
\end{proof}

Proposition \ref{characterization} suggests that  the notions of $\sAPreduces^{\dl}$- and $\EXreduces^{\dl}$-completeness do not capture the essential difficulty of the optimization complexity class $\apxl_{\nlo}$ and $\lo_{\nlo}$. Therefore, in what follows, we intend to use weaker types of reductions. In particular, we limit our interest within  $\sAPreduces^{\nc{1}}$-reductions and $\EXreduces^{\nc{1}}$-reductions.

%%%%%
%%%%%

As a quick example of $\EXreduces^{\nc{1}}$-complete problems, let us consider the {\em minimum weighted $s$-$t$ cut problem} ({\sc Min Weight-st-Cut}), which is to find an $s$-$t$ cut of a given weighted directed graph so  that the {\em (weighted) capacity}  of the cut (i.e., the total weight of edges from $S_0$ to $S_1$)  is minimized, where an {\em $s$-$t$ cut} for two distinct vertices $s,t\in V$ is a partition $(S_0,S_1)$ of the vertices for which $s\in S_0$ and $t\in S_1$.
We represent this cut $(S_0,S_1)$ by an assignment $\sigma$ from $V$ to $\{0,1\}$ satisfying the following condition: for every $v\in V$ and every $i\in\{0,1\}$, $\sigma(v)=i$ iff $v\in S_i$.

\ms
{\sc Minimum Weighted s-t Cut Problem} ({\sc Min Weight-st-Cut}):
\renewcommand{\labelitemi}{$\circ$}
\begin{itemize}\vs{-2}
  \setlength{\topsep}{-2mm}%
  \setlength{\itemsep}{1mm}%
  \setlength{\parskip}{0cm}%
\item {\sc instance:} a directed graph $G=(V,E)$, two distinguished vertices $s,t\in V$, where $s$ is a {\em source} and $t$ is a {\em sink} (or a {\em target}), and an edge weight function $c:E\to\nat^{+}$.

\item {\sc Solution:} an $s$-$t$ cut $(S_0,S_1)$, specified by an assignment $\sigma:V\to\{0,1\}$ as described above.

\item {\sc Measure:} the {\em (weighted) capacity} of the $s$-$t$ cut (i.e., $\sum_{(v,w)\in E\wedge v\in S_0\wedge w\in S_1}c(v,w)$).
\end{itemize}

Note that the capacity of any $s$-$t$ cut is at most $\max_{e\in E}\{c(e)\}|E|$. It is possible to prove that {\sc Min Weight-st-Cut} is $\EXreduces^{\nc{1}}$-complete for $\po_{\npo}$.

\begin{proposition}\label{min-st-cut-is-po}
{\sc Min Weight-st-Cut} is $\EXreduces^{\nc{1}}$-complete for $\po_{\npo}$.
\end{proposition}

The proof of Proposition \ref{min-st-cut-is-po} can be obtained by an appropriate modification of the $\p$-completeness proof of Goldschlager \etalc~\cite{GSS82} for the ``decision version'' of the {\em maximum $s$-$t$ flow problem}.
The proof of Proposition \ref{min-st-cut-is-po} is placed in Appendix for readability. The proposition will be used in Section \ref{sec:complexity-OP}.

%%%%%
\subsection{Complete Problems Concerning Path Weight}\label{sec:general-complete}

We have seen in Section \ref{sec:why-NC1} the importance of $\sAPreduces^{\nc{1}}$- and $\EXreduces^{\nc{1}}$-reductions for discussing the computational complexity of our refined optimization problems.
In this and the next subsections under those special reductions, we shall present a few complete problems for various optimization and approximation classes.

There are two categories of NLO problems to distinguish in our course of studying  the complexity of NLO problems. The first category contains NLO problems $(I,SOL,m,goal)$ for which the set $(I\circ SOL)^{\exists}$ ($=\{x\in I\mid SOL(x)\neq\setempty\}$) belongs to $\nl$ but may not fall into $\dl$ unless $\dl = \nl$. The second category, in contrast, requires the set $(I\circ SOL)^{\exists}$ to be in $\dl$. Many of the optimization problems of the first category are unlikely to fall into $\apxl_{\nlo}$ or $\lo_{\nlo}$.

First, we shall look into an optimization analogue of
the well-known {\em directed $s$-$t$ connectivity problem} (also known as the {\em graph accessibility problem} and the {\em graph reachability problem} in the past literature), denoted by $\mathrm{DSTCON}$, in which, for any  directed graph $G=(V,E)$ and two vertices $s,t\in V$, we are asked to determine whether there is a path from $s$ to $t$ in $G$.
Earlier, Jones \cite{Jon75} showed that $\mathrm{DSTCON}$ is $\nl$-complete under $\leq_{m}^{\dl}$ (log-space many-one) reductions. These reductions can be replaced by appropriate $\leq_{m}^{\nc{1}}$-reductions, and thus  $\mathrm{DSTCON}$ becomes $\leq_{m}^{\nc{1}}$-complete for $\nl$.
Let us consider a series of problems associated with minimum path weights of  graphs. First, recall the minimum path weight problem ({\sc Min Path-Weight}) introduced in Section \ref{sec:intro}.

\ms
{\sc Minimum Path Weight Problem} ({\sc Min Path-Weight}):
\renewcommand{\labelitemi}{$\circ$}
\begin{itemize}\vs{-2}
  \setlength{\topsep}{-2mm}%
  \setlength{\itemsep}{1mm}%
  \setlength{\parskip}{0cm}%

\item {\sc instance:} a directed graph $G=(V,E)$, two distinguished vertices $s,t\in V$, and a (vertex) weight function $w:V\to\nat$.

\item {\sc Solution:} a path $\SSS=(v_1,v_2,\ldots,v_k)$ from $s$ to $t$ (i.e., $s=v_1$ and $t=v_k$).

\item {\sc Measure:} ``biased'' path weight $w(S) = \max\{1,rep(bin(w(v_1))bin(w(v_2))\cdots bin(w(v_k)))\}$.
\end{itemize}

In the above definition, we generally do not demand that $s$ is a {\em source} (i.e., a node of indegree $0$) and $t$ is a {\em sink} (i.e., a node of outdegree $0$)  although such a restriction does not change the completeness of the problem.

Here, we need to remark that the choice of our measure function for {\sc Min Path-Weight} is quite artificial.
As a quick example, if $\SSS=(v_1,v_2,v_3,v_4)$ with $w(v_1)=3$, $w(v_2)=0$,  $w(v_3)=2$, and $w(v_4)=4$, then $w(\SSS) = rep(1110100)$ since $bin(0)=\lambda$ (the empty string).
It is important to note that we use the {\em biased path weight} instead of a {\em standard path weight} defined as $\sum_{i\in[k]}w(v_i)$. This comes from the fact that,  because log-space computation cannot store super-logarithmically many bits, it cannot sum up all super-logarithmically large weights of vertices. However, if we set all vertices of a given input graph have weights of exactly  $1$, then {\sc Min Path-Weight} is essentially identical to a problem of finding the ``shortest'' $s$-$t$ path in the graph.

In comparison, we also define a polynomially-bounded form of {\sc Min Path-Weight} simply by demanding that $1\leq w(v)\leq|V|$ for all $v\in V$ and by changing $w(\SSS)$ to the {\em total path weight} $w'(\SSS)=\sum_{i=1}^{k}w(v_i)$. Notice that $w'(\SSS)\leq k|V|$. For our later reference in Section \ref{sec:PBP}, we call this modified problem the {\em minimum bounded path weight problem} ({\sc Min BPath-Weight}) to emphasize the polynomially-boundedness of the problem.

Hereafter, we shall prove that {\sc Min Path-Weight} is $\sAPreduces^{\ac{0}}$-complete for $\minnl$.

\begin{theorem}\label{Min-Path-complete}
$\text{\sc Min Path-Weight}$ is $\sAPreduces^{\ac{0}}$-complete for $\minnl$.
\end{theorem}

For the $\sAPreduces^{\ac{0}}$-hardness part of Theorem \ref{Min-Path-complete}, we want to introduce a useful notion of {\em configuration graph} of a log-space auxiliary Turing machine $M$ on a given input $x$ together with any possible auxiliary input $y$,  which describes an entire computation tree of $M$ working on $x$ and $y$.
This is a weighted directed graph, which will be used in later proofs, establishing the hardness of target optimization problems; however, in those proofs, we may need to modify  the original configuration graph given below.
Since each vertex of a configuration graph is labeled by a ``partial configuration, '' we first define such partial configurations of $M$ on $x$. To simplify the following description, we consider the case where $M$ has only one work tape.

Recall from Section \ref{sec:basic_model} an auxiliary Turing machine $M = (Q,\Sigma,\{\cent,\dollar\},\Gamma,\Theta,\Phi,q_0,q_{acc},q_{rej})$ with its transition function $\delta$ mapping $(Q-\{q_{acc},q_{rej}\}) \times (\Sigma\cup\{\cent,\dollar,\lambda\}) \times \Gamma \times (\Theta\cup\{\dollar\})$ to $Q\times \Gamma\times (\Phi\cup\{\lambda\}) \times D\times D_1$.
The current tape situation is encoded into $uhw$, which indicates that  the tape content is $uv$ and the tape head is scanning the leftmost symbol of $w$, where $h$ is a special symbol representing the tape heard. A {\em partial configuration}  of $M$ on input $x$ is a tuple $v = \pair{q,w,u,\tau,\xi,k}$, which intuitively indicates a snap shot of $M$'s computation at time $k$ ($k\in\nat$) when $q$ is an inner state, $w$ is an encoding of $M$'s input tape, $u$ is an encoding of $M$'s work tape, $\tau$ is a scanning auxiliary input symbol, and $\xi$ is an output symbol or $\lambda$ to write.

We connect each partial configuration $v=(q,w,u,\tau,\xi,k)$ to  others $(p,w',u',\tau',\xi',k+1)$ for all $\tau'\in\in\Gamma\cup\{\dollar\}$ by applying a transition ``$\delta(q,\sigma'_1,\sigma'_2,\tau) = (p,\sigma'_2,\xi',d_1,d_2)$,'' where $w'$ (resp., $u'$) is an encoding of the input (resp., work) tape obtained from $w$ (resp., $u$) by this transition.
When a machine makes a $\lambda$-move on the output tape, we use the same symbol ``$\lambda$'' in place of  $\xi$ and $\xi'$. Here, we encode such partial configurations into binary strings of the same length by padding extra garbage bits (if necessary).

The {\em weight} of this vertex $v$  is defined as $\xi$ (expressed in binary).
We can view a {\em computation path} $y$ of $M$ on $x$ together with a series of nondeterministic choices of $M$,  as a sequence of partial configurations.
For two vertices $u$ and $v$, $(u,v)$ is a direct edge if, seen as partial configurations, $v$ is obtained from $u$ by a single application of $\delta$ and a choice of auxiliary input symbol.
Since  each vertex is represented by $O(n)$ symbols, the total number of  vertices is at most a polynomial in $n$. We denote by $G^{M}_{x}$ the obtained configuration graph of $M$ on $x$ since $M$ halts in polynomial time.
Note that the size of $G^{M}_{x}$ is bounded from above by a polynomial in the size of input instance $x$ of $M$.

It is important to note that, from a given encoding of a computation path $y$, we can easily extract an associated auxiliary input, because each partial configuration in $y$ contains a piece of information on the auxiliary input and $M$'s head on the auxiliary tape moves in only one direction.

\begin{proofof}{Theorem \ref{Min-Path-complete}}
For notational convenience, in the following argument, $\text{\sc Min Path-Weight}$ is expressed as $(I_0,SOL_0,m_0,\text{\sc min})$. Firstly, we want to claim that $I_0\in\dl$. This follows from the facts that  $\mathrm{DSTCON}\in\nl$ and that $(I_0\circ SOL_0)^{\exists}$ (more accurately, $(I_0\circ SOL_0)^{\exists}_{q}$ for a suitable polynomial $q$) is essentially ``equivalent'' to $\mathrm{DSTCON}$, except for the presence of a weight function $w$.
Next, we claim that {\sc Min Path-Weight} belongs to $\minnl$.
This claim comes from the following facts. On input $x=(G,s,t,w)$, let $\SSS = (v_1,v_2,\ldots,v_k)$ denote an arbitrary path from $s$ to $t$ in $G$. Since $m_0(x,\SSS)$ equals $w(\SSS)$ by definition, the value $m_0(x,\SSS)$ can be computed by an appropriate auxiliary Turing machine that writes down $bin(w(v_i))$ sequentially on a write-only output tape using $O(\log{n})$ space-bounded work tapes.
Similarly, given $x=(G,s,t,w)$ and an arbitrary sequence $\SSS$ of vertices, we can decide whether $\SSS\in SOL_0(x)$ by checking whether $\SSS$ is a path from $s$ to $t$ using a certain log-space auxiliary Turing machine.

Secondly, we shall claim that {\sc Min Path-weight} is $\sAPreduces^{\ac{0}}$-hard for $\minnl$; namely, every minimization problem in $\nlo$ is $\sAPreduces^{\ac{0}}$-reducible to {\sc Min Path-Weight}.
To prove this claim, let $P=(I,SOL,m,\text{\sc min})$ be any minimization problem in $\nlo$. Note that $I\in\dl$, $I\circ SOL\in\auxl$, and $m\in\auxfl$.
Since $m\in\auxfl$, we take an appropriate log-space auxiliary Turing machine $M$
(with three tapes) computing $m$, where any solution candidate to $P$ is provided on an auxiliary read-once tape.
Notice that there is a unique initial partial configuration. To ensure that $M$ has a {\em unique} accepting partial configuration, it suffices to force $M$ to clear out all tapes just before entering a unique accepting state.

Let us define an $\sAPreduces^{\ac{0}}$-reduction $(f,g,1)$ from $P$ to {\sc Min Path-Weight} as follows. Let $r\geq1$ and define  $f(x,r)$ to be a configuration graph $G^{M}_{x}$ of $M$ on input $x$.
If $x\in I$, then $f(x,r)\in I_0$. Let $s$ denote the initial partial configuration of $M$ on $x$ and let $t$ be the unique accepting partial configuration of $M_2$ on $x$. As a solution to {\sc Min Path-Weight}, let $y$ be any path in the graph $f(x,r)$ starting with $s$.
Each vertex in $y$ contains the information on content $\tau_i$ of the tape cell at which the auxiliary-tape head scans at time $i$. Hence, from $y$, we can recover the content of the auxiliary tape as follows. Given $y=(y_0,y_1,\ldots,y_m)$ with $y_i=(q_i,w_i,u_i,\tau_i,\xi_i,k_i)$, we retrieve $\tau_i$ for all indices $i\in[0,m]_{\integer}$ and output $\tau_0\tau_1\tau_2\cdots \tau_m$. This procedure requires only an AC$^{0}$ circuit.
Let $g(x,y,r)$ denote the entire content of the auxiliary  tape that is reconstructed from $y$ as described above. Clearly, $g$ is in $\fac{0}$ and, for any $y\in SOL(x)$, we obtain $g(x,y,r)\in SOL_0(f(x,r))$.  It is not difficult to show that $m(f(x,r),y) = m_0(x,g(x,y,r))$. Hence, $R_2(f(x,r),y)$ equals $R_1(x,g(x,y,r))$.

To complete the proof, we still need to verify that $f$ belongs to  $\fac{0}$.
For this, consider the following procedure. Recall that a graph is represented by a list of edges (i.e., vertex pairs). Starting with any input $x=(G,s,t,w)$ and $r\geq1$, generate all pairs $(u,v)$ of partial configurations and mark $(u,v)$ whenever it is an edge of $G^{M}_{x}$. This procedure needs to wire only a finite number of bits between $u$ and $v$. Hence, $f$ can be computed by an $\ac{0}$ circuit.

Therefore, {\sc Min Path-Weight} is $\sAPreduces^{\nc{1}}$-complete for $\minnl$.
\end{proofof}

%%%

In contrast to {\sc Min Path-Weight}, it is possible to define a maximization problem, {\sc Max Path-Weight}, simply by taking the maximally-weighted $s$-$t$ path for the minimally-weighted one in the definition of {\sc Min Path-Weight}. A similar argument in the proof of Theorem \ref{Min-Path-complete} establishes the $\sAPreduces^{\nc{1}}$-completeness of {\sc max Path-Weight} for $\maxnl$.

\begin{corollary}\label{Max-Path-Weight-complete}
{\sc Max Path-Weight} is $\sAPreduces^{\nc{1}}$-complete for $\maxnl$.
\end{corollary}

Is {\sc Min Path-weight} also $\sAPreduces^{\nc{1}}$-complete for $\maxnl$ and thus for $\nlo$ ($=\maxnl\cup\minnl$)?
Unlike NPO problems, the log-space limitation of work tapes of Turing machines complicates the circumstances around NLO problems.
At present, we do not know that {\sc Min Path-Weight} is  $\sAPreduces^{\nc{1}}$-complete for $\nlo$. This issue will be discussed later in Section \ref{sec:PBP}.
Under a certain assumption on $\auxfl$, nevertheless, it is possible to achieve the $\sAPreduces^{\nc{1}}$-completeness of {\sc Min Path-Weight} for $\nlo$.

We say that $\auxfl$ is {\em closed under division} if, for any two functions $f,g\in\auxfl$ outputting natural numbers in binary, the function $h$ defined by $h(x,y)=\ceilings{f(x,y)/g(x,y)}$ for all inputs $x$ and all auxiliary inputs $y$ is in $\auxfl$, provided that $g(x,y)>0$ for all inputs $(x,y)$.

\begin{proposition}\label{Min-Path-in-NLO}
Assume that $\auxfl$ is closed under division. {\sc Min Path-Weight} is $\sAPreduces^{\nc{1}}$-complete for $\nlo$.
\end{proposition}

We have already proven that {\sc Min Path-Weight} is $\sAPreduces^{\nc{1}}$-complete for $\minnl$ in Theorem \ref{Min-Path-complete}. In Lemma \ref{reduction-MaxNL-MinNL}, we shall demonstrate that every problem $P_1$ in $\maxnl$ is sAP$\ac{0}$-reducible to an appropriately chosen  problem $P_2$ in $\minnl$ if $\auxfl$ is closed under division. Since $\nlo= \maxnl\cup \minnl$, this implies that {\sc Min Path-Weight} is also $\sAPreduces^{\nc{1}}$-hard for $\maxnl$, completing the proof of Proposition \ref{Min-Path-in-NLO}.

We shall prove the remaining lemma, Lemma \ref{reduction-MaxNL-MinNL}.

\begin{lemma}\label{reduction-MaxNL-MinNL}
Assume that $\auxfl$ is closed under division. Every problem $P_1$ in $\maxnl$ is $\sAPreduces^{\ac{0}}$-reducible to an appropriate  problem $P_2$ in $\minnl$.
\end{lemma}

\begin{proof}
Let $P_1=(I_1,SOL_1,m_1,\text{\sc max})$ be any optimization problem in $\maxnl$.   Take an appropriate polynomial $p$ satisfying $2^{p(|x|)} \geq m_1^*(x)$ for every instance $x\in I_1$. For brevity, we set $b(x) = 2^{p(|x|)}$ for all $x\in I_1$. We shall construct the desired minimization problem $P_2=(I_2,SOL_2,m_2,\text{\sc min})$ in $\minnl$.
Let $I_2=I_1$ and $SOL_2=SOL_1$. Moreover, for every pair  $(x,y)\in I_2\circ SOL_2$,  define $m_2(x,y) = \ceilings{\frac{b(x)^2}{m_1(x,y)}}$.
It is important to note that, by our definition of measure function, $m_1$ always returns {\em positive values}. From this definition, it follows that  $\frac{b(x)^2}{m_2(x,y)}\leq m_1(x,y) \leq \frac{b(x)^2}{m_2(x,y)+1}$ for any $(x,y)\in I_2\circ SOL_2$.

Clearly, $I_2\in\dl$ and $I_2\circ SOL_2\in\auxfl$. From our assumption on the closure property of $\auxfl$ under division, $m_2$ falls into $\auxfl$.
Therefore, $P_2$ belongs to $\nlo$.

Let us define an sAPAC$^{0}$-reduction $(f,g,c)$ from $P_1$ to $P_2$ as follows.  Let  $f(x,r)=x$ and $g(x,y,r)=y$ for $r\in\rational^{\geq1}$, $x\in I_1$, and $y\in SOL_2(f(x,r))$. Obviously, $f,g\in\fac{0}$ follows.
If $R_2(f(x,r),y)\leq r$ for $r\geq1$; namely, $m_2^*(x)/r \leq m_2(x,y)\leq m_2^*(x)$, then
the performance ratio $R_1(x,g(x,y,r))$ for $P_1$ is upper-bounded as
\[
R_1(x,g(x,y,r)) = \frac{m_1^*(x)}{m_1(x,y)} \leq \frac{b(x)^2}{m_2^*(x)+1} \div  \frac{b(x)^2}{m_2(x,y)} = \frac{m_2(x,y)}{m_2^*(x)}.
\]
The last term is further upper-bounded by $\frac{r m_2^*(x)}{m_2^*(x)+1}\leq 1+c(r-1)$, where $c=1$, since $m_2(x,y)\leq r m_2^*(x)$. Overall, we obtain $R_1(x,g(x,y,r))\leq 1+c(r-1)$. We then conclude that $(f,g,c)$ is indeed an $\sAPreduces^{\ac{0}}$-reduction from $P_1$ to $P_2$.
\end{proof}

Henceforth, we shall discuss several variants of {\sc Min Path-Weight}. A simple variant is an {\em undirected-graph version} of {\sc Min Path-Weight}, denoted by  {\sc Min UPath-Weight}. It is possible to demonstrate that {\sc Min UPath-Weight} is log-space ${n^{O(1)}}$-approximable because, by the result of Reingold  \cite{Rei08}, using only log space, we not only determine the existence of a certain feasible solution for {\sc Min UPath-Weight} but also find at least one feasible solution if any. The special case where the weights of all vertices are exactly $1$ is the problem of finding the {\em shortest $s$-$t$ path}. This problem was discussed in \cite{Tan07}; nonetheless, it is unknown that {\sc Min UPath-Weight} belongs to $\apxl_{\nlo}$.

As another variant of {\sc Min Path-Weight}, we consider {\em forests}.  Cook and McKenzie \cite{CM87} showed that the $s$-$t$ connectivity problem for  forests is complete for $\dl$ under L-uniform NC$^{1}$ many-one reductions. Similarly, when all admissible input graphs of {\sc Min UPath-Weight} are restricted to be forests, we call the corresponding problem {\sc Min Forest-Path-Weight}. As shown in the following proposition, {\sc Min Forest-Path-Weight} turns out to be one of the most difficult problems in $\lo_{\nlo}$.
It is of importance that, unlike $\nlo$, the class $\lo_{\nlo}$ does possess complete problems.

\begin{proposition}\label{forest-path-weight}
$\text{\sc Min Forest-Path-Weight}$ is $\sAPreduces^{\ac{0}}$-complete for $\lo_{\nlo}$.
\end{proposition}

To simplify the proof of Proposition \ref{forest-path-weight}, we first give a useful lemma that helps us pay central attention to optimization problems of particular form.
Here, we say that an optimization problem $P=(I,SOL,m,goal)$ {\em admits unique solutions} if $|SOL(x)|\leq1$ holds for all $x\in I$.

\begin{lemma}\label{LO-simple-form}
For any maximization problem $Q\in \lo_{\nlo}$ (resp., $\lo_{\nlo}\cap\pbo$), there are another maximization problem $P=(I,SOL,m,\text{\sc max})$ in $\lo_{\nlo}$ (resp., $\lo_{\nlo}\cap\pbo$) and a log-space deterministic Turing machine $M_P$ such that, for any $x\in I$, (i) $Q\sAPreduces^{\ac{0}}P$, (ii) $m(x,z)=m(x,M_P(x))$ for all $z\in SOL(x)$, and (iii) $P$ admits unique solutions. The same statement holds for minimization problems.
\end{lemma}

\begin{proof}
Let $Q=(I_1,SOL_1,m_1,\text{\sc max})$ be any maximization problem in $\lo_{\nlo}$. Let $M_Q$ be a log-space deterministic Turing machine producing optimal solutions of $Q$. We then define $P=(I_2,SOL_2,m_2,\text{\sc max})$ as follows. First, we set $I_2=I_1$ and $SOL_2(x) =\{M_P(x)\}$ if $M_P(x))\neq\bot$, and $SOL_2(x)=\setempty$ otherwise.
From this definition follows $|SOL_2(x)|\leq1$ for all $x\in I_2$. Moreover, we define $m_2$ by setting $m_2(x,z)=m_1(x,M_Q(x))$ for any $(x,z)x\in I_2\circ SOL_2$. Here, we set $M_P$ to be the same as $M_Q$. Obviously, $m_2(x,M_P(x))=m_2^*(x)$ holds if $SOL_1(x)\neq\setempty$ since $M_Q(x)\in SOL_1^*(x)$.

For the desired $\sAPreduces^{\ac{0}}$-reduction $(f,g,c)$, we define $c=1$, $f(x,r)=x$, and $g(x,y,r)=y$. Clearly, $f,g\in\fac{0}$. Consider the performance ratio $R_1$ and $R_2$ for $Q$ and $P$, respectively, and assume that $R_2(f(x,r),y)\leq r$ for any $y\in SOL_2(f(x,r))$ and $r\in\rational^{\geq1}$. This assumption yields $y=M_P(x)$. Hence, $R_1(x,g(x,y,r)) = R_1(x,y) = R_1(x,M_Q(x)) =1\leq r$. Therefore, $(f,g,c)$ reduces $Q$ to $P$.
\end{proof}

Let us begin the proof of Proposition \ref{forest-path-weight}.

\begin{proofof}{Proposition \ref{forest-path-weight}}
{\sc Min Forest-Path-Weight} is assumed to have the form $(I_0,SOL_0,m_0,\text{\sc min})$. The membership relation $\text{\sc Min Forest-Path-Weight}\in \lo_{\nlo}$ essentially comes from a simple fact that, by the forest property of a given graph $G$, two nodes
$s$ and $t$ are connected in $G$ if and only if a {\em unique} path exists  between them. We can search such a unique path by starting from $s$ and following recursively adjacent edges to next nodes until either no more edges remain unsearched or $t$ is found. At the same time, we progressively write down the weight, in binary, of each node along this found path.
The recursive part of this procedure works as follows. Let $u$ be the currently visiting node. We then pick each neighbor, say, $v$ and check if there is a path between $v$ and $t$ in a graph obtained from $G$ by deleting the edge $(u,v)$. This procedure needs no more than log space.

Let $P=(I,SOL,m,\text{\sc min})$ be any minimization problem in $\lo_{\nlo}$. We assume that $P$ satisfies Conditions (ii)--(iii) of Lemma \ref{LO-simple-form}.  Our goal is to show that $P$ is $\sAPreduces^{\ac{0}}$-reducible to {\sc Min Forest-Path-Weight} via a suitably constructed $\sAPreduces^{\ac{0}}$-reduction $(f,g,c)$.
Choose a log-space deterministic Turing machine $M_1$ that produces optimal solutions of $P$. For convenience, we set $b(x)$ to be $m(x,M_1(x))$ for any instance $x\in (I\circ SOL)^{\exists}$.
Since $m\in\auxfl$, there is a log-space auxiliary Turing machine $M_m$ computing $m$. By combining $M_1$ and $M_m$ properly, we can design another log-space deterministic Turing machine, say, $M_2$ that computes $b$ with no auxiliary tape.

To make all final partial configurations unique, we want to force $M_2$ to erase all symbols on all tapes just before entering a halting state. To avoid the same partial configurations to be reached along a single computation, we additionally equip an {\em internal clock} to $M_2$.

Let us consider partial configurations of $M_2$. Note that $M_2$ is deterministic and the internal clock marks all partial configurations of $M_2$ on $x$ with different time stamps.
Note also that each symbol of the string $M_1(x)$ appears as a symbol read from the auxiliary input tape encoded into certain partial configurations of $M_1$. Hence, if we have a valid series $y$ of partial configurations associated with an accepting computation path of $M_2$ on $x$, then we can recover the string $M_1(x)$ correctly. Notationally,  $\eta(y)$ denotes this unique string obtained from a valid series $y$ of partial configurations.

Take a configuration graph $G^{M_2}_{x} = (V,E)$ from $M_2$. Note that there is at most one correct computation path of $M_2$. We set $s$ to be the initial partial configuration of $M_2$ on $x$ and set $t$ be a unique accepting partial configuration of $M_2$ on $x$.
The resulted graph forms an {\em acyclic undirected graph}, namely a forest,  because, otherwise, there are two accepting computation paths on the same input $x$. Given any partial configuration $v\in V$, we define $w(v)$ to be one bit written down newly on the output tape in this partial configuration $v$.
For the desired reduction, we define $c=1$, $f(x,r)=\pair{G,s,t,w}$, and $g(x,y,r)=\eta(y)$ for any $y\in SOL_0(f(x,r))$.
We obtain $f,g\in\fnc{1}$. It follows that $m_0(f(x,r),y)=m(x,\eta(y))$ for any $x\in I$ and $y\in SOL_0(f(x,r))$. In particular, $g(x,y,r)$ is a minimal solution of $x$ if and only if $y$ is a minimal solution of $f(x,r)$. Thus, $(f,g,c)$ $\sAPreduces^{\ac{0}}$-reduces $P$ to {\sc Min Forest-Path-Weight}.

Next, we consider any maximization problem $P$ in $\lo_{\nlo}$. Since $P$ also satisfies Condition (ii)--(iii) of Lemma \ref{LO-simple-form}, the above argument also works for this $P$ and thus establishes the $\sAPreduces^{\ac{0}}$-reducibility of $P$ to {\sc Min Forest-Path-Weight}.
\end{proofof}

As other variants of {\sc Min Path-Weight}, Nickelsen and Tantau \cite{NT05}  studied {\em series-parallel graphs} and {\em tournaments}.

%%%%%%%%%%%%%%%%%%%%%%%
\subsection{Complete Problems Concerning Finite Automata}\label{sec:approximation-class}

We shall leave graph problems behind and look into problems associated with   finite automata. \`{A}lvarez and Jenner \cite{AJ93} and later Tantau \cite{Tan07} discussed an intimate relationship between accepting computations of nondeterministic finite automata and log-space search procedures for optimal solutions.
Those problems are also closely related to {\em maximal word problems (or functions)} for fixed underlying machines.
Allender, Bruschi, and Pighizzini \cite{ABP93}, for instance, discussed the maximal word problems of various types of auxiliary pushdown automata. Within our framework of NLO problems, Tantau \cite{Tan07} presented a maximization problem finding the maximal input strings accepted by nondeterministic finite automata and demonstrated that this problem is $\sAPreduces^{\dl}$-complete for $\maxnl$.
Here, we shall show that a restricted version of this problem is $\sAPreduces^{\nc{1}}$-complete for $\apxl_{\maxnl}$.

A {\em one-way one-head nondeterministic finite automaton with $\lambda$-moves} (or a {\em $\lambda$-1nfa}, in short) $M$ is a tuple $(Q,\{0,1\},\{\cent,\dollar\},\delta,q_0,F)$ working with the input alphabet $\{0,1\}$ and a transition function $\delta: (Q-F)\times\{0,1,\lambda\}\to\PP(Q)$, where $q_0\in Q$, and $F\subseteq Q$. Initially, an input $x\in\{0,1\}^*$ is written on an input tape, surrounded by two endmarkers $\cent$ (left) and $\dollar$ (tight).  If $M$ makes a $\lambda$-move simply by applying $p\in\delta(q,\lambda)$, then $M$'s read-only tape head stays still; otherwise, the tape head moves to the next right cell.
A {\em configuration} of $M$ is a pair $(q,\sigma)$ of current inner state $q$ and scanning symbol $\sigma$. An {\em accepting computation path} $p_{M,x}$ of $M$ on input $x$ is a series of configurations starting with an initial configuration $(q_0,\cent)$ and ending with a final configuration $(q_f,\dollar)$ with $q_f\in F$ and, for any consecutive two elements $(q_i,\sigma_i)$ and $(q_{i+1},\sigma_{i+1})$ in $p_{M,x}$, $q_{i+1}$ is obtained in a single step from $(q_i,\sigma_i)$ by applying a transition of the form $q_{i+1}\in \delta(q_i,\sigma_i)$ with $\sigma_i\in\{0,1,\lambda\}$, where $\sigma_1\sigma_2\cdots\sigma_k$ is a partition of the input string $\cent x\dollar$.
If $M$ enters a certain final state in $F$ along a certain accepting computation path, then $M$ is said to {\em accept} $x$; otherwise, $M$ {\em rejects} $x$.
Associated with such $\lambda$-1nfa's, we consider the following optimization problem.
For succinctness, we hereafter express a transition ``$p\in \delta(q,\sigma)$'' as a triplet $(q,\sigma,p)$.

\ms
{\sc Maximum Fixed-Length $\lambda$-Nondeterministic Finite Automata Problem} ({\sc Max FL-$\lambda$-NFA}):
\renewcommand{\labelitemi}{$\circ$}
\begin{itemize}\vs{-2}
  \setlength{\topsep}{-2mm}%
  \setlength{\itemsep}{1mm}%
  \setlength{\parskip}{0cm}%

\item {\sc instance:} a $\lambda$-1nfa $M=(Q,\{0,1\},\{\cent,\dollar\},\delta,q_0,F)$ and a string $0^n$ for a length parameter $n$, provided that $0^n\in L(M)$.

\item {\sc Solution:} an accepting computation path of $M$ of length at most $|Q|$ on a certain input $y$ of length exactly $n$.

\item {\sc Measure:} an integer $rep(1y)$.
\end{itemize}

In the above definition, if we remove the requirement ``$0^n\in L(M)$'' and we allow $y$ to have any length up to $n$, then we obtain {\sc Max $\lambda$-NFA}, which is $\sAPreduces^{\dl}$-complete for $\maxnl$  \cite{Tan07}.

\begin{proposition}\label{Mix-2Path-Weight}
$\text{\sc Max FL-$\lambda$-NFA}$ is $\sAPreduces^{\nc{1}}$-complete for  $\apxl_{\maxnl}$.
\end{proposition}

Before proving this proposition, we show a
useful supporting lemma. The lemma helps us concentrate only on optimization problems in $\apxl_{\nlo}$ that have a certain simple structure.

\begin{lemma}\label{NLO-to-APXL}
For any maximization problem $Q$ in $\apxl_{\nlo}$ (resp., $\apxl_{\nlo}\cap\pbo$), there exist another maximization problem $P=(I,SOL,m,\text{\sc max})$ in $\apxl_{\nlo}$ (resp., $\apxl_{\nlo}\cap\pbo$)  and a log-space deterministic Turing machine $M_P$ such that,  for all $x\in I$, (i) $Q\sAPreduces^{\ac{0}} P$, (ii) $\max_{y\in SOL(x)}\{m(x,y)\}\leq 2 \min_{y\in SOL(x)}\{m(x,y)\}$ holds for all $x\in I$, (iii) there exists a function $b\in\fl$ such that  $m(x,z)\geq 2^{\floors{\log{b(x)}}}$  for all $z\in SOL(x)$, and (iv) for any $x\in I$ with $SOL(x)\neq\setempty$,  $m(x,M_P(x))=2^{\floors{\log{b(x)}}}$.
A similar statement holds for minimization problems; however, we need to replace (iii) by (iii') $m(x,z)\leq 2^{\floors{\log{b(x)}}}$.
\end{lemma}

\begin{proof}
Let $Q=(I_1,SOL_1,m_1,\text{\sc max})$ be any maximization problem in $\apxl_{\nlo}$. In what follows, we shall modify $Q$ to obtain the desired problem  $P=(I_2,SOL_2,m_2,\text{\sc max})$.

Take a polynomial $p$ such that, for any $(x,y)\in I_1\circ SOL_1$, $|y|\leq p(|x|)$ holds.
Since $Q\in\apxl_{\nlo}$, take a log-space deterministic Turing machine $M_P$ producing $\beta$-approximate solutions of $P$ for a certain constant $\beta>1$. We obtain $m_1(x,M_Q(x))\leq m_1^*(x)\leq \beta m_1(x,M_Q(x))$ for all $x\in I_1$ with $SOL_1(x)\neq\setempty$. For such an $x$, we further set $b_0(x) = 2^{\floors{\log{m_1(x,M_Q(x))}}+1}$. Note that $b_0\in\fl$
since $m_1\in\auxfl$.
Since $b_0(x)\leq b(x)\leq 2b_0(x)$, it follows that $b_0(x)\leq m_1^*(x)\leq 2\beta b_0(x)$.
Let us consider a configuration graph $G^{M_Q}_{x}$ of $M_Q$ on $x$.  Note that we can compute a string $M_Q(x)$ from $(x,G^{M_Q}_{x})$ using only log space.
For later use, we set $\alpha= 2^{\floors{\log{2\beta}}+1}+1$, which implies  $b_0(x)\leq m_1^*(x)\leq \alpha b_0(x)$.

Here, we define the desired problem $P = (I_2,SOL_2,m_2,\text{\sc max})$. For convenience, set $\Delta_{x} = (\alpha-2)b_0(x)$. Let $I_2=\{x\natural G^{M_Q}_{x}\mid x\in I_1\}$. Given $\tilde{x}=x\natural G^{M_Q}_{x} \in I_2$, $SOL_2(\tilde{x})$ contains the following strings: (i) $y\natural M_Q(x)$ for all $y\in SOL_1(x)$ satisfying $m_1(x,y)\geq b_0(x)$ and (ii) $By\natural M_Q(x)$ for every $y\in SOL_1(x)$ satisfying $m_1(x,y)<b_0(x)$, where $B$ is a special symbol. Obviously, $I_2\circ SOL_2$ is a member of $\auxl$.

For any $\tilde{y}\in SOL_2(\tilde{x})$, if $\tilde{y}= y\natural M_Q(x)$, then we set $m_2(\tilde{x},\tilde{y})= m_1(x,y) +\Delta_{x}$; if $\tilde{y}=By\natural M_P(x)$, then we set $m_2(\tilde{x},\tilde{y}) = m_1(x,y)\cdot 2^{t} + \Delta_{x}$, where $t=|bin(b_0(x))|-|bin(m_1(x,y))|$. It follows that $m_2(\tilde{x},\tilde{y})\geq b_0(x)+\Delta_{x}$ for all $(\tilde{x},\tilde{y})\in I_2\circ SOL_2$. Note that $m_2(\tilde{x},M_Q(x)\natural M_Q(x))$ equals $b_0(x)+\Delta_{x}$, which is $2^{\floors{\log2\beta}+\floors{\log{b_0(x)}}+2}$. Choose a function $b$ so that $\floors{\log{b(x)}} = \floors{\log{2\beta}}+\floors{\log{b_0(x)}}+2$
for all $x\in (I_1\circ SOL_1)^{\exists}$. Clearly, $m_2\in\auxfl$ holds.

The desired $M_P(\tilde{x})$ outputs $M_Q(x)\natural M_Q(x)$ for any $\tilde{x} \in I_2$ if $SOL_2(\tilde{x})\neq\setempty$, and it outputs $\bot$ otherwise. Let $\tilde{x}\in (I_2\circ SOL_2)^{\exists}$.
It follows that $m_2^*(\tilde{x}) = \max_{\tilde{y}\in SOL_2(x)}\{m_2(\tilde{x},\tilde{y})\} = \max_{y\in SOL_1(x)}\{m_1(x,y)+\Delta_{x}\}=m_1^*(x)+\Delta_{x}$ since $m_1^*(x)\geq b_0(x)$. Moreover, we obtain  $m_1(x,g(x,\tilde{y},r)) \geq b_0(x)$. It follows that
\begin{equation}\label{eqn:max-min-ratio}
\frac{\max_{y\in SOL_2(\tilde{x})}\{m_2(\tilde{x},y)\}}{\min_{y\in SOL_2(\tilde{x})}\{m_2(\tilde{x},y)\}} =
\frac{m_1^*(x) + \Delta_{x}}{b_0(x) + \Delta_{x}} \leq
\frac{\alpha b_0(x) + (\alpha-2)b_0(x)}{b_0(x) + (\alpha-2)b_0(x)} = 2.
\end{equation}

We remark that, by the construction of $P$ from $Q$, if $Q$ is polynomially bounded, then so is $P$. We wish to define an $\sAPreduces^{\nc{1}}$-reduction $(f,g,c)$ from $Q$ to $P$. First, we set $c= \alpha-1$. We then define $f(x,r)=x\natural G^{M_Q}_{x}$ ($=\tilde{x}$),  and $g(x,\tilde{y},r) = y$ if $\tilde{y}$ is of the form $y\natural z$, $g(x,\tilde{y},r)= z$ if $\tilde{y}$ is of the form $By\natural z$, and $g(x,\tilde{y},r)=x$ otherwise.
If $x\in I_1$, then $f(x,r)\in I_2$ since $f(x,r)=x\natural G^{M_Q}_{x}$. If $y\in SOL_1(x)$, then $y\natural M_P(x)\in SOL_2(f(x,r))$ because $M_Q(x)$ can be constructed from $(x,G^{M_Q}_{x})$ using log space.
Next, we assume that the performance ratio $R_2$ for $P$ satisfies $R_2(f(x,r),\tilde{y})\leq r$ for $x\in I_1$ and $\tilde{y}\in SOL_2(f(x,r))$. Note that $m_1(x,g(x,\tilde{y},r)) = m_1(x,y)$ if $\tilde{y}=y\natural M_Q(x)$. If $\tilde{y}=By\natural M_Q(x)$, then $m_1(x,g(x,\tilde{y},r)) = m_1(x,M_Q(x))$. We obtain
$m_1(x,g(x,\tilde{y},r))\geq m_1(x,M_Q(x))$. Write $u= g(x,\tilde{y},r)$ for simplicity. As for the performance ratio $R_1$ for $Q$, it follows that
\begin{equation}\label{eqn:R_1-alpha}
R_1(x,g(x,\tilde{y},r)) -1 = \frac{m_1^*(x)-m_1(x,u)}{m_1(x,u)} \leq
 (\alpha-1)\cdot \frac{m_2^*(\tilde{x})-m_2(\tilde{x},\tilde{y})}{(\alpha-1)m_1(x,u)}.
\end{equation}
The last term is further calculated as
\begin{equation}\label{eqn:ration-with-c}
c\cdot \frac{m_2^*(\tilde{x})-m_2(\tilde{x},\tilde{y})}{m_1(x,u)+(\alpha-2)m_1(x,u)}
\leq c\cdot \frac{m_2^*(\tilde{x})-m_2(\tilde{x},\tilde{y})}{m_1(x,u)+\Delta_{x}}
= c\cdot \frac{m_2^*(\tilde{x})-m_2(\tilde{x},\tilde{y})}{m_2(\tilde{x},\tilde{y})}.
\end{equation}
The last term equals $c(R_2(f(x,r),\tilde{y})-1)$. Since $R_2(f(x,r),\tilde{y})\leq r$, it follows that $R_1(x,g(x,\tilde{y},r)) \leq 1+ c(r-1)$. Therefore, $(f,g,c)$ reduces $Q$ to $P$.
\end{proof}

Let us begin the proof of Proposition \ref{Mix-2Path-Weight}.

\s
\begin{proofof}{Proposition \ref{Mix-2Path-Weight}}
First, we shall  argue that {\sc Max FL-$\lambda$-NFA} belongs to   $\apxl_{\maxnl}$. For simplicity, we set $\text{\sc Max FL-$\lambda$-NFA}$ as  $(I_0,SOL_0,m_0,\text{\sc min})$. Let $x=(M,0^n)$ be any instance in $I_0$ with $M=(Q,\{0,1\},\{\cent,\dollar\},\delta,q_0,F)$, where $M$ is demanded to accept $0^n$. Since we feed only inputs $y$ of length $n$, the value $rep(1y)$ varies from $2^n$ to $2^{n+1}-1$. It thus follows that, for any $u\in SOL_0(x)$,
$m_0^*(x)/2 \leq m_0(x,u) \leq m_0^*(x)$. Consider the following algorithm $N$: take $x$ as input and simulate $M$ on input $0^n$ (also by checking the size of $0^n$). This algorithm requires only log space.
We then obtain $m_0^*(x)/2\leq m_0(x,N(x))\leq m_0^*(x)$ for any $x\in I_0$. These bounds imply that $\text{\sc Max FL-$\lambda$-NFA}$ is a member of  $\apxl_{\maxnl}$.

Next, we shall show the $\sAPreduces^{\nc{1}}$-hardness of {\sc Max FL-$\lambda$-NFA}. Let $P=(I,SOL,m,\text{\sc max})$ be any problem in $\apxl_{\maxnl}$. Without loss of generality, we assume that $P$ satisfies Conditions (ii)--(iv) of Lemma \ref{NLO-to-APXL}, and thus $P$ admits a $2$-approximate algorithm.
Our goal is to show that $P$ is $\sAPreduces^{\nc{1}}$-reducible to {\sc Max FL-$\lambda$-NFA} via a suitable reduction $(f,g,c)$. Let $M_m$ be a log-space auxiliary Turing machine computing $m$ and let $M_P$ denote a log-space $2$-approximate algorithm for $P$. There is a function $b\in\fl$ such that $b(x)\leq m^*(x)\leq 2b(x)$  and $m(x,y)\geq b(x)$  for all $x\in (I\circ SOL)^{\exists}$, where $b(x)$ is of the form $2^{\floors{\log{C(x)}}}$ for appropriate function $C\in\fl$.
Moreover, we assume that, for each $x\in (I\circ SOL)^{\exists}$, there is a solution $y\in SOL(x)$ such that $m(x,y)=b(x)$.

As in the proof of Proposition \ref{Min-Path-complete}, we consider partial configurations of $M_m$.
We define $Q$ to be the set of all possible partial configurations of $M_m$.
Fix $x\in (I\circ SOL)^{\exists}$ arbitrarily and let $n$ be the size of binary string $bin(b(x))$.

We want to define a $\lambda$-1nfa $N$, which ``mimics'' a computation of $M_m$. $N$'s inner states are partial configurations of $M_m$. An input to $N$ is $bin(m(x,y))^{(-)}$ for a certain auxiliary input $y\in SOL(x)$. A move of $N$ is described as follows. Given a string $u$ and a number $k\in\nat^{+}$, $u_k$ denotes the $k$th symbol of $u$.
On such an input, $N$ nondeterministically guesses a string $y\in\{0,1\}^n$.
By reading an input symbol $\tau'\in\{0,1\}$ from $bin(m(x,y))$ one by one from left to right, $N$ changes an inner state $v=(q,x,j,\xi,u,k,\tau)$ to another inner state  $v'=(p,x,j+d_1,\xi',w,k+d_2,\tau')$ in a single step by applying $M_m$'s transition ``$(p,\tau',w_{k},d_1,d_2)\in \delta(q,x_{j},\xi,u_{k})$,'' where $\xi'$ is the next bit of $\xi$ in $y$. More precisely, we define $N$'s transition as  $\pair{p,x,j+d_1,\xi',w,k+d_2,\tau'}\in \delta_{N}(\pair{q,x,j,\xi,u,k,\tau},\tau')$ iff $(p,\tau',w_k,d_1,d_2)\in\delta(q,x_j,\xi,u_k)$, where $w$ is obtained from $u$ by changing $u_k$ to $w_k$.

We modify $N$ so that it simultaneously checks whether its input is of the form $0^n$. If so, $N$ enters a designated accepting state. Hence, $0^n\in L(N)$.

Here, let us define the desired reduction $(f,g,c)$. First, we set $f(x,r)=\pair{N}$. Given an accepting computation path $e$ in $SOL_0(f(x,r))$, we define $g(x,e,r)$ to be an auxiliary input $y_{(e)}$ fed into $M_m$, which can be obtained from $e$. It follows that $y_{(e)}\in SOL(x)$ and that $m(x,y_{(e)})=m_0(f(x,r),e)$.
Let $r\in\rational^{\geq1}$. The performance ratio $R_P$ for $P$ satisfies that $R_P(x,g(x,e,r)) = \frac{m^*(x)}{m(x,g(x,e,r))} = \frac{m_0^*(f(x,r))}{m_0(f(x,r),e)} = R_0(f(x,r),e)$.

Therefore, the reduction $(f,g,c)$ ensures $P\sAPreduces^{\nc{1}}\text{\sc Max FL-$\lambda$-NFA}$.
\end{proofof}

%%%

As a simple variant of {\sc Max $\lambda$-NFA}, we shall consider  {\em one-way one-head deterministic finite automata with $\lambda$-moves} (or $\lambda$-1dfa's, in short). A $\lambda$-1dfa $M$ is a tuple $(Q,\{0,1\},\{\cent,\dollar\},\delta,q_0,F)$ working with input alphabet $\{0,1\}$ and $\delta:(Q_{\lambda}\times\{\lambda\})\cup(Q_{+}\times \{0,1,\natural\})\to Q$, where $Q=Q_{\lambda}\cup Q_{+}$, $Q_{\lambda}\cap Q_{+}=\setempty$, $q_0\in Q$, and $F\subseteq Q$. This $M$ must satisfy the following condition: if $M$ is in state $q\in Q_{\lambda}$, then $M$'s read-only tape head stays still; otherwise, the tape head moves to the next right cell.
Here, each transition ``$\delta(q,\sigma)=p$'' is succinctly expressed as  $(q,\sigma,p)$.

\ms
{\sc Maximum Input-Restricted $\lambda$-Deterministic Finite Automata Problem} ({\sc Max IR-$\lambda$-DFA}):
\renewcommand{\labelitemi}{$\circ$}
\begin{itemize}\vs{-1}
  \setlength{\topsep}{-2mm}%
  \setlength{\itemsep}{1mm}%
  \setlength{\parskip}{0cm}%

\item {\sc instance:} a $\lambda$-1dfa $M =(Q,\{0,1\},\{\cent,\dollar\},\delta,q_0,F)$ and a list $Y=(y_1,y_2,\ldots,y_k)$ of strings over $\{0,1\}$, where $\delta$ is given as a list of (partial) transitions of the form $(q,\sigma,p)$.

\item {\sc Solution:} an accepting computation path  of $M$ of length at most $|Q|$ on a certain input $y\in Y$, which is surrounded by $\cent$ and $\dollar$.

\item {\sc Measure:} an integer $rep(1y)$.
\end{itemize}

\begin{proposition}\label{MAX-IR-lambda-DFA}
{\sc Max IR-$\lambda$-DFA} is $\sAPreduces^{\nc{1}}$-complete for $\lo_{\nlo}$.
\end{proposition}

If we  take $(M,1^n)$ as an instance and demand $y$ to have length at most $n$, then we obtain another problem, {\sc Max $\lambda$-DFA}. It is not clear that {\sc Max $\lambda$-DFA} is $\sAPreduces^{\nc{1}}$-complete for either $\lo_{\nlo}$ or $\maxnl$.

\begin{proofof}{Proposition \ref{MAX-IR-lambda-DFA}}
For convenience, let {\sc Max IR-$\lambda$-DFA} have the form $(I_0,SOL_0,m_0,\text{\sc max})$.
Let us claim that  {\sc Max IR-$\lambda$-DFA} is in $\nlo$.
Note that it is easy to check using log space whether a given instance $M =(Q,\{0,1\},\{\cent,\dollar\},\delta,q_0,F)$ is indeed a $\lambda$-1dfa; thus, $I_0\in\dl$ follows.
To see $I_0\circ SOL_0\in\auxl$, on input $(M,Y)$ together with a sequence $p = (p_1,p_2,\ldots,p_m)$ of configurations of $M$ as an auxiliary input, we can check using only log space whether $p$ is indeed an accepting computation path of $M$ (by checking that $p_1=(q_0,\cent)$, $p_m=(q_f,\dollar)$, $(p_i,\sigma,p_{i+1})$ is a transition for each $i\in[m]$ with $m\leq |Q|$ for a certain $\sigma\in\{0,1,\lambda\}$, and a series of such symbols matches one of $y$'s in $Y$. As for $m_0\in\auxfl$, it is possible to retrieve an input $y$ from $p$ and output $1y$ using log space if $y$ is in $Y$.

Next, we shall show that {\sc Max IR-$\lambda$-DFA} belongs to $\lo_{\nlo}$.
Recursively, we pick each $y$ in $Y$ in the lexicographic order and simulate a given $\lambda$-1dfa $M$ on this input $y$ to check if $M$ accepts $y$ within $|Q|$ steps. This process determines the maximal accepted input $y$ in $Y$.  Finally, we generate an accepting computation path of $M$ on this $y$. This whole procedure requires log space.
Hence, {\sc Max IR-$\lambda$-DFA} can be solved using only log space.

Hereafter, we shall  show that {\sc Max IR-$\lambda$-DFA} is $\sAPreduces^{\nc{1}}$-hard for $\lo_{\nlo}$.
Let us consider  any maximization problem $P=(I,SOL,m,\text{\sc max})$ in $\lo_{\nlo}$. We assume that $P$ satisfies Conditions (ii)--(iii) of Lemma \ref{LO-simple-form}. Since optimal solutions and their objective values are both computed using log space, for the purpose of defining $g$, we can build a log-space deterministic Turing machine $M$ that, on input $x$, records each symbol $y_i$ of a solution $y=y_1y_2\cdots y_k\in SOL^*(x)$ in one cell of one work tape (after erasing the previous symbol $y_{i-1}$ if any) and produces $bin(m(x,y))^{(-)}$ on an output tape (by removing the first bit ``$1$'' from $bin(m(x,y))$).

We shall  construct a pair $(f,g)$ of functions that $\sAPreduces^{\nc{1}}$-reduces $P$ to {\sc Max IR-$\lambda$-DFA}.
Let $x$ be any instance in $I$. We construct a $\lambda$-1dfa $N_x = (Q,\{0,1\},\{\cent,\dollar\},\delta,q_0,F)$ as follows.
We view each configuration of $M$ (including the content of its output symbol  and an input symbol) as an inner state of $N_x$. Let $Q$ be a set of all such configurations. Note that $|Q|\geq |x|$. We roughly treat an output tape of $M$ as an input tape of $N_x$. More precisely, when $M$ writes a symbol $\sigma\in\{0,1\}$ on its output tape, $N_x$ reads $\sigma$ on the input tape. When $M$ does not write any non-blank output symbol, $N_x$ makes its associated  $\lambda$-move.
Finally, we set $f(x,r) =\pair{N_x}$.  For $N_x$, let $p = (p_0,p_1,\ldots,p_m)$ be an accepting computation path of $N_x$ of length $\leq|Q|$.
We also define $g(x,p,r)$ to be an input $y_p$ to $N_x$ that is recovered from $p$ as stated above.
Note that $m_0(f(x,r),p) = rep(1bin(m(x,y_p)^{(-)})) = m(x,y_p)$.
Concerning the performance ratio $R_P$ and $R_0$ for $P$ and {\sc Max IP-$\lambda$-DFA}, respectively, it follows that $R_{P}(x,g(x,p,r)) = \frac{m^*(x)}{m(x,y_p)} = \frac{m_0^*(f(x,r))}{m_0(f(x,r),p)} = R_0(f(x),p)$. Therefore, $(f,g,1)$ is an $\sAPreduces^{\nc{1}}$-reduction from $P$ to {\sc Max IR-$\lambda$-DFA}, as requested. The case where $P$ is a minimization problem can be similarly treated.
\end{proofof}

%%%%%%%%%%%%%%%%%%%%%%%%%%%%%%%
%%%%%%%%%%%%%%%%%%%%%%%%%%%%%%%
\section{Polynomially-Bounded Complete Problems}\label{sec:PBP}

Let us recall from Section \ref{sec:comb-OPs} that an optimization problem is said to be  {\em polynomially bounded} exactly when its measure function is polynomially bounded. Recall also the notation $\pbo$, which expresses the set of all polynomially-bounded optimization problems.
For many low-complexity optimization/approximation classes below $\po_{\npo}$, polynomially-bounded optimization problems play a quite special role. With respect to log-space computation, it appears more natural to deal with polynomially-bounded optimization problems than polynomially-unbounded ones because, through Section \ref{sec:completeness}, we have been unable to present any complete problem in $\nlo$ and $\apxl_{\nlo}$ but, as we shall see shortly, we can exhibit complete problems in $\nlo\cap\pbo$ and
$\apxl_{\nlo}\cap\pbo$.

In the subsequent subsections, we shall present polynomially-bounded optimization problems, which turn out to be complete for various NL optimization and approximation classes.

%%%%%
\subsection{Maximization Versus Minimization}\label{sec:key-lemma}

Assume that we wish to show the completeness of a certain optimization problem $P$ for a target optimization/approximation class $\DD$. Since $\DD$ may be composed of maximization problems as well as minimization problems, it is necessary to construct desirable reductions to $P$ from all maximization problems in $\DD$ and also from all minimization problems in $\DD$.
Regarding $\npo$ problems, it is well-known that every minimization problem $Q$ in $\apxp_{\npo}$ has its maximization counterpart $Q'$ in $\apxp_{\npo}$ whose complexity is at least as hard as $Q$ (see, e.g., \cite[Theorem 8.7]{ACG+03} for the proof).

A similar statement holds for polynomially-bounded problems in $\apxl_{\nlo}$. This is because a log-space auxiliary Turing machine that computes a polynomially-bounded measure function can freely manipulate the outcome of the function using its space-bounded work tapes before writing it down onto an output tape.

\begin{lemma}\label{min-reduces-max}
\renewcommand{\labelitemi}{$\circ$}
\begin{enumerate}\vs{-1}
  \setlength{\topsep}{-2mm}%
  \setlength{\itemsep}{0mm}% original = 1mm
  \setlength{\parskip}{0cm}%

\item For any minimization (resp., maximization) problem $P$ in $\nlo\cap\pbo$, there  exists a maximization (resp., minimization) problem $Q$ in $\nlo\cap\pbo$ such that $P$ is $\sAPreduces^{\ac{0}}$-reducible to $Q$.

\item Let $\DD\in\{\apxl,\apxnc{1}\}$. For every minimization (resp., maximization) problem $P$ in $\DD_{\nlo}\cap\pbo$, there exists a maximization (resp., minimization) problem $Q$ in $\DD_{\nlo}\cap\pbo$ such that $P$ is $\sAPreduces^{\ac{0}}$-reducible to $Q$.

\item Let $\DD\in\{\lo,\nco{1}\}$. For any minimization (resp., maximization) problem $P$ in $\DD_{\nlo}\cap\pbo$, there  exists a maximization (resp., minimization) problem $Q$ in $\DD_{\nlo}\cap\pbo$ such that $P$ is $\EXreduces^{\ac{0}}$-reducible to $Q$.
\end{enumerate}
\end{lemma}

\begin{proof}
(1) This proof is similar in essence to that of Lemma \ref{reduction-MaxNL-MinNL}.
Given an arbitrary problem $P_1=(I_1,SOL_1,m_1,\text{\sc min})$ in $\nlo\cap\pbo$, we aim at constructing a maximization problem $P_2=(I_2,SOL_2,m_2,\text{\sc max})$ in $\nlo\cap\pbo$ to which $P_1$ is $\sAPreduces^{\ac{0}}$-reducible.
Since $P_1\in\pbo$, there is a polynomial $p$ satisfying $m_1(x,y)\leq p(|x|)$ for every $(x,y)\in I_1\circ SOL_1$. Let $I_2=I_1$, $SOL_2=SOL_1$, and $m_2(x,y) = \ceilings{\frac{p(|x|)^2}{m_1(x,y)}}$ for every $(x,y)\in I_2\circ SOL_2$. Since $m_1$ is polynomially-bounded, we can generate the entire value $m_1(x,y)$ on one of log-space work tapes and manipulate it freely as if a normal input. (If $m_2$ is not polynomially bounded, then there is no guarantee that a log-space auxiliary Turing machines  can compute $m_2$.) Since ``division'' can be implemented on TC$^{0}$ circuit \cite{Hes01}
(and thus, by a log-space machine),  $m_2(x,y)$ can be generated on a log-space work tape. Here, we define $f(x,r)=x$ and $g(x,y,r)=y$ so that $f$ and $g$ belong to $\fac{0}$. For the performance ratio $R_2$ for $P_2$, assume that $R_2(f(x,r),y)\leq r$ for any $r\geq1$; that is, $m_2^*(x)\leq m_2(x,y)\leq r m_2^*(x)$. The ratio $R_1(x,g(x,y,r)) = \frac{m_1(x,y)}{m_1^*(x)}$ is thus at most $1+c(r-1)$, where $c=1$. Hence, $(f,g,c)$ reduces $P_1$ to $P_2$.

(2) We shall show only the case of $\DD=\apxnc{1}$.
Let $P_1=(I_1,SOL_1,m_1,\text{\sc min})$ be any polynomially-bounded minimization problem in $\apxnc{1}_{\nlo}$. We want to construct a polynomially-bounded maximization problem $P_2=(I_2,SOL_2,m_2,\text{\sc max})$, which is in $\apxnc{1}_{\nlo}$ and is $\sAPreduces^{\ac{0}}$-reduced from  $P_1$. Take a constant $\gamma> 1$ and an $\nc{1}$ $\gamma$-approximate algorithm $C$ for $P_1$.
For convenience, we set $b(x) = m_1(x,C(x))$ for each instance $x\in (I_1\circ SOL_1)^{\exists}$. By the definition of $\apxnc{1}_{\nlo}$, $b$ must belong to $\fnc{1}$. Note that $m_1^*(x)\leq b(x)\leq \gamma m_1^*(x)$ for all $x\in (I_1\circ SOL_1)^{\exists}$.

We choose a constant $\Delta$ so that $\Delta >\gamma$ holds. Note that $\Delta-1>0$ since $\gamma\geq1$. Next, we define $I_2=I_1$, $SOL_2(x) = \{y\in SOL_1(x)\mid m_1(x,y)\leq b(x)\}$, and $m_2(x,y) = \Delta b(x) - \gamma m_1(x,y)$ for any $(x,y)\in I_2\circ SOL_2$.
Notice that  $m_2$ is computed by an appropriate log-space auxiliary Turing machine, and thus $I_2\circ SOL_2$ is also computed by a certain log-space auxiliary Turing machine. This implies that $P_2$ is an $\nlo$ problem.

We shall  claim that $P_2\in\apxnc{1}_{\nlo}$. Note that $m^*_2(x) = \Delta b(x) -\gamma m^*_1(x)$ for every $x\in (I_2\circ SOL_2)^{\exists}$. Moreover, $m_2(x,C(x)) = (\Delta-\gamma)b(x)$ holds. Since $m^*_1(x)\leq b(x)$, we obtain $(\Delta-\gamma)b(x)\leq m^*_2(x)$, which implies $m_2(x,C(x))\leq m^*_2(x)$. Since $b\in\fnc{1}$, the value $m_2(x,C(x))$ can be computed from $x$ by a certain $\nc{1}$-circuit.  Since $b(x)\leq \gamma m^*_1(x)$, it follows that $m^*_2(x)\leq (\Delta-1)b(x)$; thus, we obtain $m^*_2(x) \leq \frac{\Delta-1}{\Delta-\gamma}m_2(x,C(x))$. In summary, it holds that $\frac{\Delta-\gamma}{\Delta-1}m^*_2(x)\leq m_2(x,C(x))\leq m^*_2(x)$. Hence, $P_2$ belongs to $\apxnc{1}_{\nlo}$.

As for the desired $\sAPreduces^{\ac{0}}$-reduction, we define $f(x,r)=x$ and $g(x,y,r)=y$, where $r$ is any number in $\rational^{\geq1}$. It is obvious that $f$ and $g$ are in $\fac{0}$. For any fixed solution $y\in SOL_2(x)$, we assume that  $R_2(f(x,r),y)\leq r$; that is, $m_2(x,y)\leq m_2^*(x)\leq r m_2(x,y)$. We want to claim that $m_1(x,y)\leq [1+(\Delta-1)(r-1)]m_1^*(x)$.
First, we note that
\begin{equation*}
m_1(x,y) =(1/\gamma)[\Delta b(x) - m_2(x,y)] \leq (1/r\gamma)[(r-1)\Delta b(x) + \gamma m_1^*(x)]
\end{equation*}
because $m_2^*(x)\leq r m_2(x,y)$ and $m_2^*(x) = \Delta b(x) - \gamma m_1^*(x)$. Since $b(x)\leq \gamma m_1^*(x)$, we obtain $m_1(x,y)\leq (1/r)[(r-1)\Delta+1]m_1^*(x)$. Hence, it follows that $m_1(x,y)\leq [1+\theta(r-1)]m_1^*(x)$, where $\theta= (\Delta-1)/r$, which is at most $\Delta-1$ since $r\geq1$.  This shows that $P_1$ is $\sAPreduces^{\ac{0}}$-reducible to $P_2$.

(3) Consider the case of $\DD=\lo$. The construction of the desired reduction is similar to (2). Assuming that $P_1$ is in $\lo_{\nlo}\cap\pbo$, we want to show that $P_2$ is also in $\lo_{\nlo}\cap\pbo$. This is obtained simply by mapping minimal solutions for $P_1$ to maximal solutions for $P_2$.
\end{proof}

%%%%%%%%%%%%%%%%%%
\subsection{Completeness of Graph Problems}\label{sec:graph-problems}

As our starting point, we recall from Section \ref{sec:general-complete} a bounded variant of {\sc Min Path-Weight}, called {\sc Min BPath-Weight}, which uses the {\em total path weight} (i.e., $w(\SSS)= \sum_{i=1}^{k}w(v_i)$ for $\SSS=(v_1,v_2,\ldots.v_k)$ with $k\leq|V|$) as a measure function with an extra condition that $1\leq w(v)\leq |V|$ for all $v\in V$. Earlier, Tantau \cite[Theorem 5.1]{Tan07} discussed the case when all vertices of a given graph have weight exactly $1$. To verify that {\sc Min BPath-Weight} is $\sAPreduces^{\nc{1}}$-complete for $\minnl\cap\pbo$, we need to modify the proof of Proposition \ref{Min-Path-complete} in the following way.
For this purpose, we first modify a given measure function $m$ so that its log-space auxiliary Turing machine $M_m$ produces $m(x,y)$ on one of its work tapes and then copies each bit (including ``$0$'') from the lower bit to the higher bit at each step. We further modify $M_m$ so that its internal clock helps it halt in exactly $p(n)$ steps, for a suitable polynomial $p$. We then
define $w(v)$ to be $b_1 2^{e-1} +1$ if $v$ contains a string $b_1b_2\cdots b_e$ written  on an output tape of $M_m$ in a target partial final configuration.
Otherwise, define $w(v)=1$. It follows that $\sum_{i=1}^{q(n)}w(v_i) = q(n)+ m(x,y)$.
By reducing minimization problems to maximization problems by Lemma \ref{min-reduces-max}(1), we can prove that every minimization problem in $\nlo\cap\pbo$ is also $\sAPreduces^{\nc{1}}$-reducible to {\sc Min BPath-Weight}. Therefore, we obtain the following completeness result.

\begin{lemma}\label{Bpath-complete}
{\sc Min BPath-Weight} is $\sAPreduces^{\nc{1}}$-complete for $\nlo\cap\pbo$.
\end{lemma}

%%%%%%%%%%%%%%%%%%%

In Section \ref{sec:general-complete}, we have mostly dealt with optimization problems that are associated with the path weights of graphs. Another natural type of optimization problems is
a problem of searching a path of a directed graph starting at a given source   toward an appropriately chosen vertex whose weight is well-defined and must be maximal. Tantau  \cite[Theorem 3.2]{Tan07} earlier demonstrated that this maximization problem is $\sAPreduces^{\dl}$-complete for $\nlo\cap\pbo$.
Moreover, its slightly modified version was shown to be $\sAPreduces^{\dl}$-complete for $\apxl_{\nlo}\cap\pbo$ \cite[Theorem 3.7]{Tan07}.
In what follows, we shall discuss a similar optimization problem using ``undirected'' graphs with ``total'' weight functions, particularly, in the case where vertex weights are all bounded.

Let us define formally this problem as follows.

\ms
{\sc Maximum Undirected Bounded Vertex Weight Problem} ({\sc Max UB-Vertex}):
\renewcommand{\labelitemi}{$\circ$}
\begin{itemize}\vs{-2}
  \setlength{\topsep}{-2mm}%
  \setlength{\itemsep}{1mm}%
  \setlength{\parskip}{0cm}%

\item {\sc instance:} an undirected graph $G=(V,E)$, a source $s\in V$, and a (vertex) weight function  $w:V\rightarrow \nat^{+}$ satisfying $w(v)\leq |V|$ for every $v\in V$.

\item {\sc Solution:} a path of $G$ starting at $s$ and ending at a certain vertex $t$ in $V$.

\item {\sc Measure:} the weight $w(t)$ of $t$.
\end{itemize}

A {\em directed-graph version} of {\sc Max UB-Vertex}, called {\sc Max B-Vertex},
has a log-space $n^{O(1)}$-approximate algorithm \cite{Tan07} but is not known to fall into $\apxl_{\nlo}$.

\begin{proposition}\label{max-bvertex-complete}
$\text{\sc Max UB-Vertex}$ is $\sAPreduces^{\nc{1}}$-complete for $\lo_{\nlo}\cap\pbo$.
\end{proposition}

Even if the vertex $t$ in the above definition of {\sc Max UB-Vertex} is restricted to a vertex of degree exactly $1$ as the following proof shows, the obtained problem is still $\sAPreduces^{\nc{1}}$-complete for $\lo_{\nlo}\cap\pbo$.
\ms

\begin{proof}
We begin with setting {\sc max UB-Vertex} as $(I_0,SOL_0,m_0,\text{\sc max})$. Let $G=(V,E)$ be any graph given as an instance in $I_0$. notice that $SOL_0(x)\neq\setempty$ for all $x\in I_1$.
Since the weight of every vertex in $G$ is at most the input size,  $\text{\sc Max UB-Vertex}$ is polynomially bounded. Thus, it is not difficult to show that optimal solutions for $\text{\sc Max UB-Vertex}$ can be found using log space by making a series of nonadaptive queries to oracle $A=\{\pair{G,s,w,1^k}\mid \exists t\in V\,[w(t)\geq k \wedge \,\text{$s$ and $t$ are connected}\,]\}$ by incrementing $k$ from $1$ to $|V|$.
To see that $A$ is in $\dl$, we sequentially pick a different $t\in V$, check if $s$ and $t$ are connected using Reingold's log-space algorithm for DSTCON, and finally  check if $w(t)\geq k$. Therefore, {\sc Max UB-Vertex} is $\dl$-solvable; that is, $\text{\sc Max UB-Vertex}$ belongs to $\lo_{\nlo}$.

Concerning the $\EXreduces^{\nc{1}}$-hardness of {\sc Max UB-Vertex}, let us consider a polynomially-bounded maximization problem $P=(I,SOL,m,\text{\sc max})$ in $\lo_{\nlo}$. The case of minimization problems follows from Lemma \ref{min-reduces-max}(3).
By Lemma \ref{LO-simple-form}, it is possible to  assume that $P$ satisfies Conditions (ii)--(iii) of the lemma. Take a log-space auxiliary Turing machine $M_m$ that computes $m$. Moreover, we take a log-space deterministic Turing machine $M_P$ that produces maximal solutions for $P$. Define $b(x) = m(x,M_P(x))$ for each instance $x\in (I\circ SOL)^{\exists}$.
Take a polynomial $p$ such that $b(x)\leq p(|x|)$ for all $x\in (I\circ SOL)^{\exists}$. For technicality, we demand that $p(n)>1$ for all $n$.
Notice that $b$ is computed using log space by the following simple machine $M'$   equipped with
a special work tape, called a {\em solution tape}, in which we use only one tape cell. Starting with input $x$, $M'$ simulates $M_P$ on $x$. Whenever $M_P$ writes a symbol $\sigma$ on its output tape, $M'$ writes $\sigma$ on the solution tape,  simulates $M_m$ on $x$ while an auxiliary-tape head is scanning $\sigma$, and erases $\sigma$ from the solution tape. This deletion of the symbol $\sigma$ is necessary because the output $M_P(x)$ may be super-logarithmically long.
Finally, $M'$ outputs $b(x)$.

Here, we construct a configuration graph $G$ in a way similar
to the proof of Proposition \ref{forest-path-weight} using $M'$; however, we use a quite different weight function.
Since all weights are polynomially bounded,  we can embed an entire content of an output tape of $M'$ into a partial configuration. To be more precise, we can force $M'$ to use one of its work tapes to compute $b(x)$ and, in the end, copy  the content of this tape into its write-only output tape.

For a vertex $v$ representing a certain partial halting configuration, its weight $w(v)$ is set to be the value written on the output tape in this partial configuration unless the value is not zero. For any other vertex associated with partial non-halting configurations, we simply set $w(v)=1$.

We further define $f(x,r)$ to be the above-mentioned configuration graph $G$ together with the weight function $w$. Given a computation path $y$ of $M'$ on $x$, since each partial configuration in $y$ contains the information on an output symbol produced by $M_P$, it is possible to recover from $y$ an entire output string $M_P(x)$. We thus set $g(x,y,r)$ to be $M_P(x)$ reconstructed from $y$. It is easy to check that  $f,g\in\fnc{1}$.
Note that $m_0(f(x,r),y) = b(x)$ if $y$ is in $SOL_0(f(x,r))$. This implies that the performance ratio $R$ of $g(x,y,r)$ always satisfies $R(x,g(x,y,r))=1$ for any $y\in SOL_0(f(x,r))$. As a result, $(f,g,1)$ reduces $P$ to {\sc max UB-Vertex}, as requested.
\end{proof}

%%%%%

To obtain an $\sAPreduces^{\nc{1}}$-complete problem for $\apxnc{1}_{\nlo}\cap\pbo$, we place a restriction on the behavior of a (vertex) weight function of $\text{\sc Max UB-Vertex}$.
The {\em maximum undirected 2 vertex weight problem} ({\sc Max U2-Vertex}) is a variant of {\sc Max UB-Vertex} with an additional requirement: for any instance $\pair{G,s,w}$ with $G=(V,E)$, it holds that
$\max_{v\in V}\{w(v)\}\leq 2\min_{v\in V}\{w(v)\}$.

\begin{proposition}
{\sc Max U$2$-Vertex} is $\sAPreduces^{\nc{1}}$-complete for $\apxnc{1}_{\nlo}\cap \pbo$.
\end{proposition}

\begin{proof}
To simplify the following proof, we write $\text{\sc Max U2-Vertex}$ as $(I_0,SOL_0,m_0,\text{\sc max})$. Similar to {\sc Max UB-Vertex}, {\sc Max U2-Vertex} can be shown to be in $\nlo\cap\pbo$. In particular, to verify the extra condition that (*) $\max_{v\in V}\{w(v)\}\leq 2\min_{v\in V}\{w(v)\}$ for an instance $(G,s,w)$, we pick each pair $(v_1,v_2)$ of $G$'s vertices and check that either $w(v_1)\leq 2w(v_2)$ or  $w(v_2)\leq 2w(v_1)$ holds. This procedure requires only log space.

The following algorithm $C$ confirms that {\sc Max U$2$-Vertex} belongs to  $\apxnc{1}_{\nlo}$. On input $x =(G,s,w) \in I_0$, output a length-$1$ path of $G$ from vertex $s$ to the vertex that appears first in an instance (where $G$ is given in binary as a list of edge relations).
This process can be implemented by $\nc{1}$-circuits. It thus follows from Condition (*) that $m_0(x,C(x)) \leq m_0^*(x) \leq 2 m_0(x,C(x))$. Therefore, $C$ is an 2-approximate algorithm for {\sc Max U2-Veterx}. This shows that {\sc Max U2-Vertex} falls into $\apxnc{1}_{\nlo}$.

Next, we want to show that $\text{\sc Max U2-Vertex}$ is $\sAPreduces^{\nc{1}}$-hard for $\apxnc{1}_{\nlo}\cap\pbo$. Let $P=(I,SOL,m,\text{\sc max})$ be any maximization problem in $\apxnc{1}_{\nlo}\cap \pbo$.
The case of minimization is handled by Lemma \ref{min-reduces-max}(2).
Let $M_m$ be an auxiliary Turing machine computing $m$ using log space.
Since $m$ is polynomially bounded, $M_m$ first produces $m(x,y)$ on one of its work tapes and then copies it onto an output tape just before halting. Let $C_P$ be an $\nc{1}$-circuit producing $\alpha$-approximate solutions of $P$ for a certain constant $\alpha>1$. For convenience, let $b(x) = m(x,C_P(x))$ for each $x\in (I\circ SOL)^{\exists}$ and $b(x)=\bot$ for all the others $x$. It follows that $b\in\fnc{1}$ and that $b(x)\leq m^*(x)\leq \alpha b(x)$ for all $x\in (I\circ SOL)^{\exists}$.
Because of the definition of partial configurations, from any accepting computation path of $M_m(x,y)$, we can easily recover an auxiliary input $y$.

Let us prove $P\sAPreduces^{\nc{1}}\text{\sc Max U2-Vertex}$ via a certain $\sAPreduces^{\ac{0}}$-reduction $(f,g,c)$. Let $r\geq1$ and $x\in (I\circ SOL)^{\exists}$. Firstly, we define $f(x,r)$ to be $(G,s,w)$, where $G$ is a configuration graph
of $M_m$ using normal input $x$, $s$ is the initial partial configuration of $M_m$, and $w$ is defined later.  Given a computation path $u=(v_1,v_2,\ldots,v_k)$ of $M_m$ using $x$, let $y_u$ denote a unique auxiliary input used for $M_m$ constructed from $u$. Note that $y_u$ is in $SOL(x)$.
Here, we further define $g(x,u,r)=y_u$ and set $c=\alpha-1$. In addition, we set $\Delta_x=(\alpha-2)b(x)$.
For an accepting configuration $v$, let $w(v) = w' + \Delta_x$, where $w'$ is the number written on $M_m$'s output tape.  For other configurations $v$, let $w(v)= b(x)+\Delta_x$. Note that $b(x)+\Delta_x \leq w(v)\leq m^*(x)+\Delta_x$.
Similarly to Eqn.(\ref{eqn:max-min-ratio}) in the proof of Lemma \ref{NLO-to-APXL}, it holds that $\max_{v\in V}\{w(v)\}/\min_{v\in V}\{w(v)\}\leq 2$. Since $b\in\fnc{1}$, $f(x,r)$ is computed by an $\nc{1}$-circuit. For any $u=(v_1,v_2,\ldots,v_k)\in SOL_0(f(x,r))$, it follows that $m_0(f(x,r),u) = w(v_k) = m(x,y_u)+\Delta_x = m(x,g(x,u,r)) +\Delta_x$. Moreover, we obtain $m_0^*(f(x,r)) = m^*(x)+\Delta_x$.
Given a number $r\in\rational^{\geq1}$, consider the performance ratio $R$ of $g(x,u,r)$ with respect to $x$. Since $m(x,g(x,u,r))=b(x)$ for any $u\in SOL(x)$, calculations similar to  Eqns.(\ref{eqn:R_1-alpha})--(\ref{eqn:ration-with-c})  lead to $R(x,g(x,u,r)) -1 \leq c (R_0(f(x,r),u)-1)$.
We therefore conclude that $(f,g,c)$ reduces $P$ to {\sc Max U2-Vertex}.
\end{proof}

%%%%%%
\subsection{Completeness of Algebraic and Combinatorial Problems}\label{sec:algebraic-problem}

Apart from graph problems, we shall study algebraic and combinatorial problems. We begin with an algebraic problem, which turns out to be complete for $\nlo\cap\pbo$. As {\sc Max UB-Vertex} and {\sc Max U2-Vertex} have orderly structures induced by ``edge relations,'' the algebraic problem that we shall consider below has a similar structure induced by ``operations.''

Given a finite set $X$, we consider a binary operation $\circ:X\times X\to X$. A binary operation $\circ$ is {\em associative} if $(x\circ y)\circ z = x\circ (y\circ z)$ holds for all $x,y,z\in X$. For a subset $S$ of $X$, we say that a set $G(S)$ is {\em generated  by $\circ$ from $S$} if $G(S)$ is the smallest set that contains $X$ and is closed under $\circ$.
The decision problem, called {\sc AGen}, of determining whether $t$ is in $G(S)$  for a given instance $(X,S,\circ,t)$ with $t\in X$ is $\leq_{m}^{\dl}$-complete for $\nl$ \cite{JLL76}.
Let us consider its optimization counterpart, which we call {\sc Min AGen}.

\ms
{\sc Minimum Associative Generation Problem} ({\sc Min AGen}):
\renewcommand{\labelitemi}{$\circ$}
\begin{itemize}\vs{-2}
  \setlength{\topsep}{-2mm}%
  \setlength{\itemsep}{1mm}%
  \setlength{\parskip}{0cm}%

\item {\sc instance:} a finite set $X$, an associative binary operation $\circ:X\times X\to X$, a set $S\subseteq X$, an element $t\in X$, and a weight function $w:X\to\nat^{+}$ satisfying $w(x)\leq|X|$ for all $x\in X$.

\item {\sc Solution:} a sequence $(x_1,x_2,\ldots,x_m)$ of elements in $X$ with $1\leq m\leq|X|$ so that the element $x_1\circ x_2\circ \cdots \circ x_m$ in $G(S)$ equals $t$.

\item {\sc Measure:} the value $\sum_{i=1}^{m}w(x_i)$.
\end{itemize}

Note that $s$ belongs to $G(S)$ iff there is a sequence $(x_1,x_2,\ldots,x_m)$ of elements in $S$ with $1\leq m\leq|X|$ for which $s=x_1\circ x_2\circ \cdots \circ x_m$ holds \cite{JLL76}. In what follows, we demonstrate the $\sAPreduces^{\nc{1}}$-completeness of {\sc Min AGen} for $\nlo\cap\pbo$.

\begin{proposition}
{\sc Min AGen} is $\sAPreduces^{\nc{1}}$-complete for $\nlo\cap\pbo$.
\end{proposition}

\begin{proof}
For convenience, we set $\text{\sc Min AGen} = (I_0,SOL_0,m_0,\text{\sc min})$.
First, we show that {\sc Min AGen} is in $\nlo\cap\pbo$. Let $(X,S,\circ,w,t)$ be any instance to {\sc Min AGen}. We can build an auxiliary Turing machine $M$ as follows. Let $u= (x_1,x_2,\ldots,x_m)$ be a sequence of $X$ and is given to an auxiliary tape of $M$.
To see that $I_0\in\dl$, it suffices to check, using log space, whether (i) $S\subseteq X$, (ii) $\circ$ is associative, and (iii) $w(x)\leq|X|$ for all $x\in X$.
It is also easy to see that $I_0\circ SOL_0\in\auxfl$. To obtain  $m_0((X,S,\circ,w,t),u)$, we need to compute the value $x_1\circ x_2\circ \cdots \circ x_m$ and then output the value $\sum_{i=1}^{m}w(x_i)$. Since $m_0$ is polynomially bounded, this value can be obtained using log space. This $m_0$ therefore belongs to $\auxfl$.
This show that {\sc Min AGen} is in $\nlo$.

Next, we want to show that {\sc Min AGen} is $\sAPreduces^{\nc{1}}$-hard for   $\nlo\cap\pbo$. For this purpose, we shall show that $\text{\sc Min BPath-Weight}\sAPreduces^{\nc{1}}\text{\sc Min AGen}$ since, by Lemma \ref{Bpath-complete}, {\sc Min BPath-Weight} is $\sAPreduces^{\nc{1}}$-complete for $\nlo\cap\pbo$. Let $\text{\sc Min BPath-Weight} = (I_1,SOL_1,m_1,\text{\sc min})$.
Let $h= (G,s,t,w)$ be any instance in $I_1$ with $G=(V,E)$.
Note that $w(v)\leq|V|$ for all $v\in V$. For a sequence $\SSS=(v_1,v_2,\ldots,v_k)$ with $k\leq |V|$, since $m_1(h,\SSS) = \sum_{i=1}^{k}w(v_i)$, we obtain $m_1(h,\SSS)\leq|V|^2$.

We want to define an $\sAPreduces^{\nc{1}}$-reduction $(f,g,c)$ from {\sc Min BPath-Weight} to {\sc Min AGen} as follows.
Our construction essentially follows from the proof of \cite[Theorem 5]{JLL76}. Let $c=1$. Letting $f(h,r) = (X,S,\circ,\tilde{w},t)$, we define $X$, $S$, $\circ$, $\tilde{w}$ as follows. Let $X=V\cup(V\times V)\cup\{\natural\}$ and $S=E\cup\{s\}$. A binary operation $\circ$ admits the following rules: for all $u\in X$ and $x,y,z\in V$, (i) $u\circ \natural = \natural\circ u=\natural$, (ii) $x\circ y=\natural$ for all $x,y\in V$, (iii) $(x,y)
\circ z = \natural$, $x\circ (x,z) =z$,  and $x\circ (y,z)=\natural$ if $y\neq z$, and (iv) $(x,y)\circ (y,v) = (x,v)$ and $(x,y)\circ (u,v)=\natural$ if $y\neq u$. The desired weight function $\tilde{w}$ is defined as  $\tilde{w}(\natural)=|X|$, $\tilde{w}(x,y) = w(x)+w(y)$,  and $\tilde{w}(x)=w(x)$  for all $x,y\in V$. Note that $\tilde{w}(x)\leq \max\{|X|,2|V|\}\leq |X|$  for all $x\in X$ and $\tilde{w}(x_1)+\sum_{i=1}^{k}\tilde{w}(x_i,x_{i+1}) + \tilde{w}(x_k) = 2\sum_{i=1}^{k}w(x_i)$ for $x_1,x_2\,\ldots,x_k\in V$.

Note that, given an $s$-$t$ path $(x_1,x_2,\ldots,x_k)$ with $s=x_1$ and $t=x_k$, since $s\in S$, it is possible for us to prove recursively the membership $x_i\in G(S)$ for every  $i\in[2,k]_{\integer}$; hence, $t\in G(S)$ follows.
For any sequence $u=(x_1,(x_1,x_2),(x_2,x_3),\ldots,(x_{k-1},x_k))$ in $SOL_{0}(f(h,r))$ with $s=x_1$ and $t=x_k$,  we define $g(h,u,r) = (x_1,x_2,\ldots,x_k)$.
It follows that $m_0(f(h,r),u) = \tilde{w}(x_1) + \sum_{i=1}^{k}\tilde{w}(x_i,x_{i+1}) = 2\sum_{i=1}^{k}w(x_i) -w(t)$ and $m_1(h,g(h,u,r)) = \sum_{i=1}^{k}w(x_i)$. From those equalities, we obtain $2m_1(h,g(h,u,r)) = m_0(f(h,r),u) + w(t)$.

Here, we intend to verify that $(f,g,c)$ correctly reduces {\sc Min BPath-Weight} to {\sc Min AGen}. Take any $r\in\rational^{\geq1}$ and any $u\in SOL_0(f(h,r))$.  Assume that the performance ratio $R_0$ for {\sc Min AGen} satisfies $R_0(f(h,r),u)\leq r$. It then follows that $R_1(h,g(h,u,r)) = \frac{m_1(h,g(h,u,r))}{m_1^*(h)} = \frac{m_0(f(h,r),u)+w(t)}{m_0^*(f(h,r))+w(t)} \leq \frac{m_0(f(h,r),u)}{m_0^*(f(h,r))} \leq r$.
Therefore, $(f,g,c)$ is a correct $\sAPreduces^{\nc{1}}$-reduction.
\end{proof}

%%%%%%%%%%%

Jenner \cite{Jen95} studied a few variants of the well-known {\em knapsack problem}. In particular, she introduced a decision problem, called CUK, of determining whether, given unary string pairs  $(0^w,0^p), (0^{w_0},0^{p_0}),(0^{w_1},0^{p_1}),\ldots,(0^{w_n},0^{p_n})$, where  $w,w_i\in\nat^{+}$ and $p,p_i\in\nat$ for $i\in[n]$, there is a $\{0,1\}$-sequence $(z_0,z_1,z_2,\ldots,z_n)$ satisfying $w2^{p} = \sum_{i=0}^{n}z_i\cdot w_i2^{p_i}$.
She showed that this problem is ($\dl$-uniform) $\leq^{\nc{1}}_{m}$-complete for $\nl$. Here, we turn this decision problem into an optimization problem, which will be proven to be $\sAPreduces^{\nc{1}}$-complete for $\apxl_{\nlo}\cap\pbo$.

\ms
{\sc Maximization 2-Bounded Close-to-Unary Knapsack Problem} ({\sc Max 2BCU-Knapsack}):
\renewcommand{\labelitemi}{$\circ$}
\begin{itemize}\vs{-2}
  \setlength{\topsep}{-2mm}%
  \setlength{\itemsep}{1mm}%
  \setlength{\parskip}{0cm}%

\item {\sc instance:} a pair $(0^w,0^p)$, a sequence $(0^{w_0},0^{p_0}),(0^{w_1},0^{p_1}),\ldots,(0^{w_n},0^{p_n})$ of pairs, and a weight sequence $(c_0,c_1,\ldots,c_n)$ of positive integers, where $w,w_i\in\nat^{+}$, $p,p_i\in\nat$, $1\leq c_i\leq nwp$ for all $i\in[0,n]_{\integer}$, and $\max_{0\leq i\leq n}\{c_i\}\leq 2\min_{0\leq i\leq n}\{c_i\}$, provided that $w_0=w$ and $p_0=p$. Here, the notation $0^0$ expresses the empty string $\lambda$.

\item {\sc Solution:} a sequence $z=(z_0,z_1,\ldots,z_n)$ of Boolean values satisfying $w2^p = \sum_{i=1}^{n}z_i\cdot w_i2^{p_i}$.

\item {\sc Measure:} $\max_{0\leq i\leq n}\{ c_i z_i \}$.
\end{itemize}

A trivial solution $z=(1,0,\ldots,0)$, which indicates the choice of $(0^{w_0},0^{p_0})$, is needed to ensure that $\text{\sc Max 2BCU-Knapsack}$ is indeed in $\apxl_{\nlo}$.
Regarding each value $w_i2^{p_i}$, the following simple inequalities hold:  $|bin(w_i2^{p_i})| \leq \log|0^{w_i}|+|0^{p_i}|+1 = \log{w_i}+p_i+1$.

In what follows, we show the completeness of {\sc Max 2BCU-Knapsack} for $\apxl_{\nlo}\cap\pbo$.

\begin{lemma}
{\sc Max 2BCU-Knapsack} is $\sAPreduces^{\nc{1}}$-complete
for $\apxl_{\nlo}\cap\pbo$.
\end{lemma}

\begin{proof}
Let $\text{\sc Max 2BCU-Knapsack} = (I_0,SOL_0,m_0,\text{\sc max})$. First, we argue that {\sc Max 2BCU-Knapsack} is in $\apxl_{\nlo}\cap\pbo$. It is obvious that $m_0$ is polynomially bounded. Earlier, Jenner \cite{Jen95} demonstrated  that $\mathrm{CUK}\in\nl$. A similar argument shows that $I_0\circ SOL_0\in\auxl$ and $m_0\in\auxfl$; therefore, {\sc Max 2BCU-Knapsack} falls into $\nlo\cap\pbo$.
Let $x=(K,C)$ be any instance to {\sc Max 2BCU-Knapsack}, where $C=(c_0,c_1,\ldots,c_n)$ and $K$ is composed of $(0^{w},0^{p}),(0^{w_0},0^{p_0}),(0^{w_1},0^{p_1}),\ldots,(0^{w_n},0^{p_n})$.
Choose $z_0=(1,0,\ldots,0)\in\{0,1\}^{n+1}$.  It follows that $z_0\in SOL_0(x)$ and that
$R_0(x,z_0) = \frac{m_0^*(x)}{m_0(x,z_0)} \leq 2$ since $\min_{0\leq i\leq n}\{c_i\}\leq m_0^*(x)\leq 2\min_{0\leq i\leq n}\{c_i\}$.
This implies that {\sc Max 2BCU-Knapsack} is in $\apxl_{\nlo}$.

For the $\sAPreduces^{\nc{1}}$-hardness of {\sc Max 2BCU-Knapsack}, let us consider an arbitrary maximization problem $P=(I,SOL,m,\text{\sc max})$ in $\apxl_{\nlo}\cap\pbo$ that satisfies Conditions (ii)--(iv) of  Lemma \ref{NLO-to-APXL}.
To construct an $\sAPreduces^{\nc{1}}$-reduction $(f,g,2)$ from $P$ to {\sc Max 2BCU-Knapsack}, we first consider a log-space deterministic Turing machine $M_P$ that produces  $2$-approximate solutions of $P$.
In addition, we take a function $b\in\fl$ for which
$b(x)\leq m^*(x)\leq 2b(x)$ and $m(x,y)\geq b(x)$ for all $(x,y)\in I\circ SOL$ by Lemma \ref{NLO-to-APXL}. Let $M_m$ be a log-space auxiliary Turing machine computing $m$.

For the purpose of defining an appropriate reduction, we assume that $M_m$ on input $(x,y)$ writes each symbol of $y$ on a designated cell of a particular work tape so that,

The definition of partial configurations makes it possible to retrieve the entire string $y$ from any halting computation path $e$ of $M_m$ on input $(x,y)$. We write $h(e)$ for the auxiliary input $y$ fed into $M_m$.
Since $m$ is polynomially bounded, we force $M_m$ to calculate $m(x,y)$ on one of its work tapes and copies it onto an output tape just before halting so that, for any  $y\in SOL(x)$, $m(x,y)$ equals the sum of the numbers written on the output tape of configuration along a computation path of $M_{m}$ on input $(x,y)$.

Let us recall the notion of configuration graph from the proof of Proposition \ref{Min-Path-complete} and consider a configuration graph $G_{x}^{M_m}$ of $M_m$ on input $x\in I$. In what follows, we consider only the case where $SOL(x)\neq\setempty$.
For convenience, we set $G^{M_m}_{x} = (V,E)$ with $V=\{v_1,v_2,\ldots,v_m\}$ for a  certain integer $m\geq1$, where $v_1$ expresses the initial partial configuration of $M_m$. Note that the first move of $M_m$ is fixed and does not depend on the choice of inputs.
We further make $M_m$ terminate in exactly $q(|x|)$ steps for a suitable polynomial $q$, independent of the choice of $y$ satisfying $|y|\leq p(|x|)$.

Fix $x\in(I\circ SOL)^{\exists}$ and write $q$ for $q(|x|)$.
Take an arbitrary number $r\in\rational^{\geq1}$.
Following \cite[Theorem 1]{Jen95}, let $w=2^{t}$ and $p=2(q+1)t$, where $t = 2\ceilings{\log{q}}$. If $(v_i,v_j)\in E$, then let $w_{ij} = (2^{t}-j)2^{2t}-(2^{t}-i)$ and $p_{ij,k} = 2kt$ for  $0\leq k\leq q-1$.
Let $w_0=w$, $p_0=p$, $w_1=2^{t}$, $p_1=0$, $w_{q}=q2^{2t}$, and $p_{q} = 2(q-1)t$.
Let $K$ be composed of the following pairs: $(0^w,0^p)$, $(0^{w_0},0^{p_0})$,  $(0^{w_1},0^{p_1})$, $(0^{w_{q}},0^{p_{q}})$, and $(0^{w_{ij}},0^{p_{ij}})$ for all $(v_i,v_j)\in E$. A series $C=(c_0,c_1,c_{q},c_{ij})_{(v_i,v_j)\in E}$ is defined as follows.
For any partial configuration pair $(v_i,v_j)$, if $v_j$ is an accepting partial configuration, then let $c_{ij}$ be the number written on an output tape of $M_m$; otherwise, let $c_{ij}=b(x)$. Note that $\min_{(v_i,v_j)\in e}\{c_0,c_1,c_q,c_{ij}\}\geq b(x)$.
For any $x\in I$ and $r\in\rational^{\geq1}$, we define $f(x,r)$ to be the pair $(K,C)$ given above.

Given any accepting computation path $e=(v_1,v_2,\ldots,v_q)$ of $M_m$ on input $x$ ($|x|=n$), we define a series $z=(z_0,z_1,z_q,z_{ij})_{(v_i,v_j)\in e}$ as $z_0=0$, $z_1=z_q=1$, and $z_{ij}=1$ if $(v_i,v_j)\in e$, and $z_{ij}=0$ otherwise.
To emphasize $e$, we write $z_{(e)}$ for this series $z$.
As was shown in \cite{Jen95}, we obtain $z\in SOL_0(f(x,r))$ with $z\neq(1,0,\ldots,0)$ iff there exists an accepting computation path $e$ satisfying $z=z_{(e)}$, namely, $w2^{p} = w_12^{p_1}+w_q2^{p_q}+ \sum_{(v_i,v_j)\in e}\sum_{0\leq k\leq q-1}w_{ij}2^{p_{ij,k}}$ and thus $m_0(f(x,r),z_{(e)}) = \max_{(v_i,v_j)\in e}\{c_0,c_1,c_q,c_{ij}\} = m(x,h(e))$ since $m(x,y)\geq b(x)$ for all $y\in SOL(x)$. From this fact, we define $g(x,z_{(e)},r) = h(e)$ for every  solution $z_{(e)}\in SOL_0(f(x,r))$.
Since $z\in SOL_0(f(x,r))$, there exists a suitable accepting computation path $e$ in $G^{M_m}_{x}$ satisfying $z=z_{(e)}$.
Thus,  we obtain
\[
R(x,g(x,z_{(e)},r)) = \frac{m^*(x)}{m(x,g(x,z_{(e)},r))} = \frac{m_0^*(f(x,r))}{m_0(f(x,r),z_{(e)})} = R_0(f(x,r),z_{(e)}).
\]
Therefore, $(f,g,2)$ reduces $P$ to {\sc Max 2BCU-Knapsack}.
\end{proof}

Next, we shall present a complete problem in $\nco{1}_{\nlo}\cap \pbo$. Our problem is a simple extension of the Boolean formula value problem.
The {\em Boolean formula value problem} (BFVP) is to determine whether a given Boolean formula $\phi$ is satisfied by a given truth assignment $\sigma$. This problem can be compared with CVP (circuit value problem), which is known to be $\p$-complete \cite{Lad75}.
Buss \cite{Bus87} showed the membership of BFVP to $\mathrm{ALOGTIME}$ and the $\leq_{m}^{\mathrm{DLOGTIME}}$-hardness of BFVP for $\mathrm{ALOGTIME}$ is relatively easy to verify. Therefore, since $\mathrm{ALOGTIME}$ equals $\nc{1}$,
$\text{\sc BFVP}$ is $\leq^{\dlogtime}_{m}$-complete for $\nc{1}$.

Let us consider its optimization counterpart.

\ms
{\sc Maximum Boolean Formula Value Problem} ({\sc Max BFVP}):
\renewcommand{\labelitemi}{$\circ$}
\begin{itemize}\vs{-2}
  \setlength{\topsep}{-2mm}%
  \setlength{\itemsep}{1mm}%
  \setlength{\parskip}{0cm}%

\item {\sc instance:}  a set $\Phi = \{\phi_1,\phi_2,\ldots,\phi_n\}$ of Boolean formulas and a Boolean assignment $\sigma$ for all variables in the formulas in $\Phi$ with $n\geq1$.

\item {\sc Solution:} a {\em nonempty} subset $S\subseteq \Phi$ of formulas satisfied by $\sigma$.

\item {\sc Measure:} the cardinality $|S|$ of $S$.
\end{itemize}

Recall that the formula language of a Boolean formula family $\{\phi_n\}_{n\in\nat}$ is composed of all tuples $\pair{c,i,y}$ such that $|y|=n$ and the $i$th character of the $n$th formula $\phi_n$ is $c$. Following \cite{BIS90}, a language $L$ is in $\nc{1}$ iff there exists a family of Boolean formulas representing $L$ whose formula language is in $\dlogtime$.

\begin{lemma}
{\sc Max BFVP} is $\sAPreduces^{\ac{0}}$-complete for $\nco{1}_{\nlo}\cap\pbo$.
\end{lemma}

\begin{proof}
For our convenience, express $\text{\sc Max BFVP}$ as $(I_0,SOL_0,m_0,\text{\sc max})$. Given a set $\Phi= \{\phi_1,\phi_2,\ldots,\phi_n\}$ of Boolean formulas and its truth assignment $\sigma$, since the number of all Boolean formulas $\phi_i$ satisfied by $\sigma$ is upper-bounded by the input size, {\sc Max BFVP} is polynomially bounded. Next, we argue that {\sc Max BFVP} is in $\nco{1}_{\nlo}$. It is obvious that {\sc Max BFVP} belongs to $\nlo$. On input $(\Phi,\sigma)$,
we check in parallel whether $\sigma$ satisfies $\phi_i$  by making nonadaptive queries to oracle BFVP (i.e.,  $(\phi_i,\sigma)\in\mathrm{BFVP}$), and finally we count the number of satisfied formulas $\phi_i$.
The last counting process requires another $\nc{1}$-circuit.  Since $\mathrm{BFVP}\in\nc{1}$ \cite{Bus87}, we can implement the whole procedure using $\nc{1}$-circuits. Thus, {\sc Max BFVP} belongs to $\nco{1}_{\nlo}$.

To see the $\sAPreduces^{\ac{0}}$-hardness of {\sc Max BFVP} for $\nco{1}_{\nlo}\cap\pbo$, let $P=(I,SOL,m,\text{\sc max})$ be any maximization problem in $\nco{1}_{\nlo}\cap \pbo$. The minimization case can be similarly treated.
Let us take a uniform family $\{C_n\}_{n\in\nat}$ of $\nc{1}$-circuits computing maximal solutions of $P$, that is, $C_{|x|}(x)\in SOL^*(x)$ for any $x\in (I\circ SOL)^{\exists}$. Define $SOL'(x)=\{C_{|x|}(x)\mid C_{|x|}(x)\neq\bot\}$ for each $x\in I$.
Note that $|SOL'(x)|\leq 1$ holds for all $x\in I$.
For simplicity, we assume that $I\subseteq \{0,1\}^*$ and take any binary string  $x=x_1x_2\cdots x_n$ in $I$. Let the output size of $C_n$ be $k$. For each index $i\in[k]$, we define a new $\ac{0}$-circuit $D_k^{(i)}(y_1,\ldots,y_k)=y_i$. It follows that $D_k^{(i)}(C_n(z_1,\ldots,z_n))$ is an $\nc{1}$-circuit, where  $z_1,z_2,\ldots,z_n$ are Boolean variables.

Owing to \cite[Theorem 9.1]{BIS90}, we can replace a uniform family of $\nc{1}$-circuits by a family of Boolean formulas whose formula language is in $\dlogtime$.
For each $i\in[k]$, let $\phi_i$ be a Boolean formula expressing the circuit $D_k^{(i)}(C_n(z_1,\ldots,z_n))$. Note that $D_k^{(i)}(C_n(x_1\cdots x_n)) = u_i$ if $C_n(x_1x_2\cdots x_n)=u_1u_2\cdots u_k$.  Let $\sigma_x$ be the assignment that assigns value $x_i$ to variable $z_i$.

Finally, we define $f(x,r) = \pair{\pair{\phi_1,\ldots,\phi_m},\sigma_x}$ and $g(x,y,r)= u_1u_2\cdots u_k$ for each $y\in SOL_0(f(x,r))$, where $y$ is of the form $(\phi_{i_1},\phi_{i_2},\ldots,\phi_{i_t})$, and $u_i=1$ if $\phi_i\in y$ and $u_i=0$ otherwise. It follows that $|SOL_0(f(x,r))|\leq 1$ for all $x\in I_0$. If $y=(\phi_{i_1},\phi_{i_2},\ldots,\phi_{i_t})\in SOL_0(f(x,r))$ with $f(x,r) = \pair{\pair{\phi_1,\ldots,\phi_m},\sigma_x}$, then $\sigma_x$ satisfies $\{\phi_{i_1},\phi_{i_2},\ldots,\phi_{i_t}\}$. Hence, we obtain $g(x,y,r)\in SOL^*(x)$. Hence, $\frac{m^*(x)}{m(x,g(x,y,r))} = \frac{m^*(x)}{m(x,u_1u_2\cdots u_k)}=1$. It thus follows that $P$ is $\sAPreduces^{\ac{0}}$-reducible to {\sc Max BFVP}.
\end{proof}

%%%%%%%%%%%%%%%%%%%%%%%%%%%%%%
%%%%%%%%%%%%%%%%%%%%%%%%%%%%%%
\section{Weak Approximation Schemes}\label{sec:LSAS}

Approximability of optimization problems is a crucial concept in a course of our study on the computational complexity of those problems. The classes $\apxl_{\nlo}$ and $\apxnc{1}_{\nlo}$ both admit approximation algorithms; however, the effectiveness of approximation is often far from the desirable one used in practice. Approximation schemes, on the contrary, provide a much finer notion of approximability because the performance ratio of used approximation algorithms can be made arbitrarily smaller; in other words, we can find solutions that are arbitrary close to optimal solutions.

We shall discuss optimization problems that admit such approximation schemes.  Concerning $\lsas_{\nlo}$, Nickelsen and Tantau \cite{NT05} and also Tantau \cite{Tan07} proposed $\nlo$ problems that admit an $\lsas_{\nlo}$; in particular, a polynomially-bounded maximization problem, called {\sc Max-HPP},   was shown in \cite[Theorem 5.7]{Tan07} as a member of $\lsas_{\nlo}$. Here, we rephrase this problem in terms of {\em complete graphs} in comparison with {\sc Min BPath-Weight}.

%{\sc Max CPath-Weight} (maximum complete-graph path weight) takes a complete directed graph $G=(V,E)$ (including self-loops) and a weight function $w:V\times V\to[1,|V|]_{\integer}$, and finds a path $\SSS=(v_1,v_2,\ldots,v_k)$ from $s$ to $t$ of length exactly $k$ in $G$ whose path  weight  $w(\SSS) = \sum_{i=1}^{k-1}w(v_{i},v_{i+1})$ is maximized.

\ms
{\sc Maximum Complete-Graph Path Weight Problem} ({\sc Max CPath-Weight}):
\renewcommand{\labelitemi}{$\circ$}
\begin{itemize}\vs{-2}
  \setlength{\topsep}{-2mm}%
  \setlength{\itemsep}{1mm}%
  \setlength{\parskip}{0cm}%

\item {\sc instance:}   a directed complete graph $G=(V,E)$ with self-loops, a source $s\in V$, and an edge weight function $w:V\times V\to\nat^{+}$ with $w(v_1,v_2)\leq|V|$ for any $v_1,v_2\in V$.

\item {\sc Solution:} a path $\SSS=(v_1,v_2,\ldots,v_{k})$ of length $k\leq |V|$  in $G$ starting at $s$ (i.e., $s=v_1$).

\item {\sc Measure:} total path  weight  $w(\SSS) = \sum_{i=1}^{k-1}w(v_{i},v_{i+1})$.
\end{itemize}

\begin{proposition}\label{Min-Comp-Wight-in-LSAS}
\begin{enumerate}
  \setlength{\topsep}{-2mm}%
  \setlength{\itemsep}{1mm}%
  \setlength{\parskip}{0cm}%

\item \cite{Tan07} {\sc Max CPath-Weight} is in $\lsas_{\nlo}\cap\pbo$.

\item {\sc Max CPath-Weight} is $\EXreduces^{\nc{1}}$-hard for $\lo_{\nlo}\cap\pbo$.
\end{enumerate}
\end{proposition}

\begin{proof}
(2) Let $\text{\sc Max CPath-Weight} = (I_0,SOL_0,m_0,\text{\sc max})$.
Given any maximization problem $P=(I,SOL,m,\text{\sc max})$ in $\lo_{\nlo}\cap\pbo$, we want to define an $\EXreduces^{\nc{1}}$-reduction $(f,g)$ from $P$ to {\sc Max CPath-Weight} in the following fashion.
Let $M_P$ be a log-space deterministic Turing machine computing optimal solutions of $P$; in other words, $m^*(x)= m(x,M_P(x))$ for all $x\in (I\circ SOL)^{\exists}$ with $SOL(x)\neq\setempty$.
Let $x$ be any instance in $(I\circ SOL)^{\exists}$ and consider a configuration graph  with $G^{M_P}_{x}=(V_x,E_x)$ of $M_P$ on $x$.
Let $s$ denote an initial partial configuration of $M_P$ on $x$.

Define $V=V_x$ and $E=V\times V$ and set $G=(V,E)$. Since $M$ is deterministic,
there exists a unique computation path $(p_1,p_2,\ldots,p_k)$ starting at  $s=p_1$ and ending at an accepting configuration $t=p_k$.
As for an edge weight function $w$, we define $w(v,v')=|V|$ if either (i) $v$ is an accepting configuration and $v=v'$ or (ii) $v'$ is obtained from $v$ in a single step by $M_P$; otherwise, let $w(v,v')=1$.
The desired $f$ is obtained by setting $f(x) = (G,s,w)$. This $f$ can be computed by an appropriate $\ac{0}$-circuit.
It is not difficult to see that any optimal solution of {\sc Max CPath-Weight} is $u_x= (p_1,p_2,\ldots,p_k,p_k,\ldots,p_k)$ of length exactly $|V|$.
Since $(p_1,p_2,\ldots,p_k)$ is a computation path, we can retrieve from it an output string produced by $M_P$ on $x$. More generally, given a path $u$ starting at $s$ and ending at a certain accepting configuration, we reconstruct from $u$ an output  string $y_u$ of $M_P$ on $x$.
Using these strings, we set $g(x,u)=y_u$. This function $g$ can be implemented by a certain $\ac{0}$-circuit. Note that $m_0^*(f(x)) = (|V|-1)|V|$.

If $u\in SOL_0^*(f(x))$, then, since $m_0^*(f(x)) = (|V|-1)|V|$, $u$ coincides with $u_x$. Hence, $y_u$ equals $M_P(x)$. Thus, $(f,g)$ reduces $P$ to {\sc Max CPath-Weight}.
\end{proof}

%%%%%%

Next, we shall look into another approximation class $\ncas{1}_{\nlo}\cap\pbo$.
The next example is slightly artificial but it can be proven to fall into  $\ncas{1}_{\nlo}$.
Recall from Section \ref{sec:PBP} {\sc Max U2-Vertex}, which is a member of  $\apxnc{1}_{\nlo}\cap \pbo$. We modify this problem by changing its requirement of  $\max_{v\in V}\{w(v)\}\leq 2\min_{v\in V}\{w(v)\}$ for each instance $(G,s,w)$ to $\max_{v\in V}\{w(v)\}\leq (1+1/k_G)\min_{v\in V}\{w(v)\}$ with  $k_G = \log\log{|V|}$. We call the resulted problem {\sc Max UApp-Vertex} (maximum undirected approximable vertex weight problem).

\begin{proposition}\label{ncsas-example}
\begin{enumerate}
  \setlength{\topsep}{-2mm}%
  \setlength{\itemsep}{1mm}%
  \setlength{\parskip}{0cm}%

\item {\sc Max UApp-Vertex} is in $\ncas{1}_{\nlo}\cap\pbo$.

\item {\sc Max UApp-Vertex} is $\sAPreduces^{\ac{0}}$-hard for $\lo_{\nlo}\cap\pbo$.
\end{enumerate}
\end{proposition}

\begin{proof}
(1) We express $\text{\sc Max UApp-Vertex}$ as $(I_0,SOL_0,m_0,\text{\sc max})$ and we plan to define its approximation scheme, which can be implemented by certain $\ac{0}$-circuits.
Let $r\in\rational^{>1}$ and set $k=\frac{1}{r-1}$. Take any instance $(x,r)$ in $I_0$ given to {\sc Max UApp-Vertex} with $G=(V,E)$. We consider
two cases separately, depending on the value of $k$.

(i) Assume that $k\leq \log\log|V|$. It follows from the requirement for $I_0$ that $\max_{v\in V}\{w(v)\}\leq (1+1/k)\min_{v\in V}\{w(v)\}$. In this case, our desired approximation algorithm first searches the nearest vertex, say, $v_0$ connected from $s$ and directly outputs it. This vertex $v_0$ gives $R(x,v_0) = \frac{m_0^*(x)}{w(v_0)} \leq \frac{\max_{v\in V}\{w(v)\}}{\min_{v\in V}\{w(v)\}} \leq 1+1/k = r$. This procedure provides an $r$-approximate solution and can be easily implemented by an appropriate $\ac{0}$-circuit.

(ii) Assume that $k>\log\log|V|$; in other words, $|V|< 2^{2^k}$. Consider the following brute force algorithm: pick a vertex $y$ one by one, check if $s$ and $y$ are connected, and choose the one $y$ that has the largest value $w(y)$. This procedure requires space $O(|V|\log|V|)$, which equals $O(2^{2^k}\log{n})$.

(2) Similarly to the proof of Proposition \ref{max-bvertex-complete}, let
$P=(I,SOL,m,\text{\sc max})$ denote any maximization problem in
$\lo_{\nlo}\cap \pbo$. Write $M_m$ for a log-space auxiliary Turing machine computing $m$ and write $M_P$ for a log-space deterministic Turing machine finding maximal solutions of $P$. Let $b(x)=m(x,M_P(x))$ for all $x\in (I\circ SOL)^{\exists}$.  Since $b$ is in $\fl$, choose a log-space deterministic Turing machine $M'$ that computes $b$ together with generating each symbol of $M_P(x)$ for the purpose of the later reconstruction of $M_P(x)$ from any halting computation patj of $M'$ on $x$.
For the problem $P$, we assume Conditions (ii)--(iii) of Lemma \ref{LO-simple-form}. Note that $|SOL(x)|\leq 1$ holds for
all instances $x\in I$.

Let us define $f(x,r) = (G,s,w)$, where $G=(V,E)$ is a configuration graph of $M'$ on $x$ and $s$ is the initial configuration. Here, we assume that $M'$ is $c\log{n}$ space-bounded for a certain constant $c>0$; moreover, by modifying $M'$ slightly, we also assume that $M'$ has one work tape with tape alphabet of cardinality at least $2$. Because of this space bound, there are at least $2^{c\log{n}}=n^c$ vertices in $G$. Thus, $\log\log{|V|}\geq \log\log{n}$ follows.
Let $w(v)$ be $\ceilings{\log\log{n}}+1$ if $v$ is a halting partial  configuration of $M'$ on $x$; otherwise, let $w(v) = \ceilings{\log\log{n}}$.  It follows that $\ceilings{\log\log{n}} \leq m(x,y) \leq \ceilings{\log\log{n}}+1$ for all path $y$ in $G$. Hence, $\max_{y\in SOL(f(x,r))}\{m(x,y)\}\leq \ceilings{\log\log{n}}+1\leq (1+1/\log\log{n})\ceilings{\log\log{n}} \leq \min_{y\in SOL(f(x,r))}\{m(x,y)\}$. Next, we define $g(x,y,r)$ to be the string $M_P(x)$ recovered from $y$. This is possible by the assumption on the behavior of $M'$. Finally, we set $c'=1$. It is not difficult to show that $(f,g,c')$ reduces $P$ to {\sc Max UApp-Vertex}.
\end{proof}

%%%%%%%%%%%%%%%%%%%%%%%%%%%%%%
%%%%%%%%%%%%%%%%%%%%%%%%%%%%%%
\section{Relations among Refined Optimization/Approximation  Classes}\label{sec:complexity-OP}

Let us turn our attention to relationships among classes of refined optimization problems introduced in Section \ref{sec:preliminaries}.
Optimization problems are, by definition, associated directly with their {\em underlying decision problems}. It is therefore natural to ask what relationships exist between classes of optimization problems and classes of decision problems.
We intend to investigate relationships between optimization problems and decision problems.
First, let us recall from Lemma \ref{PO=APXP=NLO} the collapse  of classes $\po_{\nlo}$, $\ptas_{\nlo}$, and $\apxp_{\nlo}$ down to $\nlo$. The class $\nlo$ may further collapses to much lower-complexity classes if $\ac{1}$ collapses to $\dl$ or $\nc{1}$.
Here, we prove the following assertions concerning the class $\nlo$.

\begin{proposition}\label{by-AJ-result}
\begin{enumerate}\vs{-1}
  \setlength{\topsep}{-2mm}%
  \setlength{\itemsep}{0mm}% original = 1mm
  \setlength{\parskip}{0cm}%

\item \cite{Tan07} If $\dl = \ac{1}$, then $\lo_{\nlo} =\nlo$.

\item If $\nc{1} =\ac{1}$, then  $\nco{1}_{\nlo} =\nlo$.
\end{enumerate}
\end{proposition}

The opposite direction of this proposition is not known to hold.
The proposition comes directly from a more general result stated below. In particular, Lemma \ref{char-with-AC1}(1) extends Lemma \ref{PO=APXP=NLO}.

\begin{lemma}\label{char-with-AC1}
\begin{enumerate}\vs{-1}
  \setlength{\topsep}{-2mm}%
  \setlength{\itemsep}{0mm}% original = 1mm
  \setlength{\parskip}{0cm}%

\item $\aco{1}_{\nlo} = \nlo$.

\item If $\dl = \ac{1}$, then $\lo_{\nlo} =\aco{1}_{\nlo}$.

\item If $\nc{1} =\ac{1}$, then $\nco{1}_{\nlo} =\aco{1}_{\nlo}$.
\end{enumerate}
\end{lemma}

\begin{proof}
(1) Obviously, $\aco{1}_{\nlo}$ is included in $\nlo$. For the opposite inclusion, consider {\sc Max $\lambda$-NFA}, given in Section \ref{sec:approximation-class}, which is $\sAPreduces^{\nc{1}}$-complete for
$\maxnl$ \cite[Theorem 3.1]{Tan07}.
It is sufficient to show that (i) {\sc Max $\lambda$-NFA} is in $\aco{1}_{\nlo}$ and (ii) $\aco{1}_{\nlo}$ is closed under $\sAPreduces^{\dl}$ (and thus $\sAPreduces^{\nc{1}}$).

(i)  \`{A}lvarez and Jenner \cite{AJ95} also Allender \etalc~\cite{ABP93} showed the containment $\optl\subseteq \fac{1}$.
Take any $\lambda$-1nfa $M$ as an instance to {\sc Max $\lambda$-NFA}.
We want to find a maximal solution using $\ac{1}$ circuits.
For this purpose, using $\ac{1}$ circuits, we convert $M$ into an equivalent regular grammar $G$ in {\em Chomsky Normal Form}.
As done in the proof of \cite[Corollary 4.6]{ABP93}, we can find the lexicographically first string over $\{0,1\}$ of length at most $n$ produced by $G$. This process can be implemented on $\ac{1}$-circuits. Therefore, {\sc Max $\lambda$-NFA} can be solved by $\ac{1}$-circuits.

(ii) Assume that $P\sAPreduces^{\dl}Q$ and $Q\in\aco{1}_{\nlo}$. Take an $\sAPreduces^{\dl}$-reduction $(f,g,c)$ from $P$ to $Q$. Let $P=(I_1,SOL_1,m_1,\text{\sc max})$ and $Q=(I_2,SOL_2,m_2,\text{\sc max})$. Let $C_Q$ be an $\ac{1}$ circuit computing optimal solutions of $Q$. Consider the following algorithm $D$. On input $x\in I_1$, compute $f(x,1)$ and output $C_Q(f(x,1))$. Clearly, $D(x)\in SOL_1(x)$. Note that $\fac{1}$ is closed under functional composition. Since $\fl\subseteq \fac{1}$, $D$ is also in $\fac{1}$. Hence, $P$ is in $\aco{1}_{\nlo}$.

(2)  Assume that $\dl=\ac{1}$. This implies that $\fl=\fac{1}$. Take $P=(I,SOL,m,goal)$ in $\aco{1}_{\nlo}$. There is an $\ac{1}$ circuit $C_P$ computing optimal solutions of $P$ and another $\ac{1}$-circuit $C_m$ that computes $m(x,C_P(x))$ for all $x\in (I\circ SOL)^{\exists}$. Treating circuits as functions they compute, we obtain $C_P,C_m\in\fac{1}$. Since $\fl=\fac{1}$, $C_P$ and $C_m$ fall into $\fl$. This implies that $P$ belongs to $\lo_{\nlo}$.
Thus, $\lo_{\nlo}=\aco{1}_{\nlo}$ follows.

(3) Here, we assume that $\nc{1}=\ac{1}$. Note that this assumption implies $\fnc{1}=\fac{1}$. Let us consider an arbitrary optimization problem $P=(I,SOL,m,goal)$ in $\aco{1}$. Take $\ac{1}$-circuits that compute optimal solutions of $P$ as well as their values.
By a way similar to (2), since $\fnc{1}=\fac{1}$, $P$'s optimal solutions can be computed  using appropriate $\nc{1}$-circuits. As a result, $P$ belongs to $\nco{1}_{\nlo}$.
\end{proof}

%%%%%

As revealed in Section \ref{sec:PBP}, the polynomial-boundedness property  is  crucial for optimization problems in discussing their computational complexity.  If we fix our focal point on polynomially-bounded optimization problems, we can manage to give another characterization of $\nlo\cap\pbo$ in terms of adaptive relativization.

To explain a notion of relativization, we empower each log-space Turing machine $M$ with {\em query mechanism}, in which $M$ makes a series of ``queries'' to a given ``oracle'' $A$ (which is simply a language) and $A$ returns its answers to $M$ in a single step. The machine $M$ has an extra {\em query tape} on which $M$ writes a query string $z$ by moving its tape head from the left to the right. After entering a designated inner state, called a {\em query state}, the string $z$ is transmitted to the oracle $A$. The oracle erases all symbols on the query tape and it writes $1$ if $z\in A$ and writes $0$ otherwise. The machine's inner state is changed to an answer state, from which $M$ resumes its computation. We make the length of every query string polynomially bounded. This machine $M$ is generally known as an {\em oracle Turing machine}.
We write $\lo_{\nlo}^{A}$ to denote the class of all $\nlo$ problems whose optimal solutions and their values are computed by log-space oracle Turing machines using oracles $A$. The union of $\lo_{\nlo}^{A}$ for all sets $A\in\nl$ is denoted by $\lo_{\nlo}^{\nl}$.

\begin{lemma}\label{LO-with-oracle-NL}
$\lo_{\nlo}^{\nl}\cap\pbo = \nlo\cap\pbo$.
\end{lemma}

\begin{proof}
($\subseteq$) This inclusion is trivial because
$\lo_{\nlo}^{\nl}\subseteq \nlo$ holds by the definition of
$\lo_{\nlo}^{\nl}$.

($\supseteq$) Let $P=(I,SOL,m,goal)$ be any $\nlo$ problem. It suffices to show the existence of log-space oracle Turing machines that compute optimal solutions and their values relative to oracles in $\nl$.
Here, we consider only the case of $goal=\text{\sc max}$. We define two sets $A$ and $B$. Take a polynomial $p$ for which $y\in SOL(x)$ implies both $|y|\leq p(|x|)$ and $m(x,y)\leq p(|x|)$. Let $A$ be composed of all strings $(x,k)$ such that there exists a string $z\in SOL(x)$ with $k\leq p(|x|)$ satisfying  $m(x,z)=k$. By making a series of queries $(x,1),(x,2),\ldots,(x,p(|x|))$ one by one to oracle $A$, we can find the maximal number $k_0$ satisfying $(x,k_0)\in B$. Next, we define $B$ as the set of all strings $(x,k,u,b)$ with $b\in\{0,1\}$ and $k,|ub|\leq p(|x|)$ such that there is a string $z\in SOL(x)$ for which $m(x,z)=k$, $ub$ is an initial segment of $z$. Given the maximal number $k_0$, by making the $i$th query $(x,k_0,u_i,0)$ and $(x,k_0,u_i,1)$ for all $i\in[0,p(|x|)]_{\integer}$ sequentially, we reconstruct $z_0$ satisfying  $m(x,z_0)=k_0$. As the desired oracle, we take $A\oplus B$, which clearly belongs to $\nl$. Therefore, we find a maximal solution $z_0$ using only log space. This yields $P\in\lo_{\nlo}^{A\oplus B}\subseteq \lo_{\nlo}^{\nl}$.
\end{proof}

Here, we present three close connections between classes of decision problems and classes of optimization problems.

\begin{proposition}\label{equivalent-relation}
\begin{enumerate}\vs{-1}
  \setlength{\topsep}{-2mm}%
  \setlength{\itemsep}{0mm}% original = 1mm
  \setlength{\parskip}{0cm}%

\item \cite{Tan07} $\dl=\nl$ iff $\lo_{\nlo}\cap\pbo = \nlo\cap\pbo$.

\item $\nc{1}=\dl$ iff $\nco{1}_{\nlo}\cap\pbo = \lo_{\nlo}\cap\pbo$ iff $\nco{1}_{\nlo} = \lo_{\nlo}$.

\item $\dl=\p$ iff $\lo_{\npo}\cap\pbo = \po_{\npo}\cap\pbo$ iff $\lo_{\npo} = \po_{\npo}$.
\end{enumerate}
\end{proposition}

For Proposition \ref{equivalent-relation}(1)--(2), we use DSTCON and USTCON on unweighted directed graphs. Earlier, Jones \cite{Jon75} showed that $\mathrm{DSTCON}$ is $\nl$-complete under log-space reduction. His reduction can be improved to $\nc{1}$-reduction, and thus DSTCON is $\leq_{m}^{\nc{1}}$-complete for $\nl$. Reingold \cite{Rei08} proved that USTCON belongs to $\dl$. From its $\leq_{m}^{\nc{1}}$-hardness for $\dl$, we conclude that USTCON is $\leq_{m}^{\nc{1}}$-complete for $\dl$.
For Proposition \ref{equivalent-relation}(3), we use the {\em circuit value problem} ($CVP$) of determining whether a given Boolean circuit outputs $1$ on a certain input. Ladner \cite{Lad75} pointed out that this is $\leq^{\dl}_{m}$-complete for $\p$; more strongly, it is $\leq^{\nc{1}}_{m}$-complete for $\p$.

\begin{proofof}{Proposition \ref{equivalent-relation}}
(1) Here, we give a proof different from the one given in \cite{Tan07}. Our proof relies on Lemma \ref{LO-with-oracle-NL}. If $\dl=\nl$, then $\lo_{\nlo}^{\nl}$ equals $\lo_{\nlo}^{\dl}$, which coincides with $\lo_{\nlo}$. Since $\lo_{\nlo}^{\nl}\cap\pbo =\nlo\cap\pbo$ by Lemma \ref{LO-with-oracle-NL}, we obtain $\lo_{\nlo}\cap\pbo=\nlo\cap\pbo$. The desired consequence follows because $\lo_{\nlo}=\lo_{\nlo}^{\dl}$.

Conversely, assume that $\lo_{\nlo}\cap\pbo = \nlo\cap\pbo$. Lemma \ref{LO-with-oracle-NL} implies that $\lo_{\nlo}^{\nl}\cap\pbo = \lo_{\nlo}\cap\pbo$. Consider DSTCON, which is $\leq_{m}^{\ac{0}}$-complete for $\nl$. Define $P=(I,SOL,m,\text{\sc max})$ as follows. Let $I$ consist of all strings $(G,s,t)$ for which $G=(V,E)$, $s,t\in V$, and there is an $s$-$t$ path in $G$. Let $x=(G,s,t)$. Define $SOL(x)$ to be a collection of all $s$-$t$ path in $G$. Let $m(x,y)=2$ if $y\in SOL(x)$ and $m(x,y)=1$ otherwise. Note that $P$ is in $\nlo\cap\pbo$. By our assumption, $P$ belongs to $\lo_{\nlo}$. Hence, take a log-space Turing machine $M$ that computes optimal solutions of $P$. It follows that $M(x)\neq\bot$ iff $x\in\mathrm{DSTCON}$. Hence, DSTCON can be recognized using log space, yielding $\dl=\nl$.

(2) For this claim, we shall show separately that (i) if  $\nco{1}_{\nlo}\cap\pbo = \lo_{\nlo}\cap\pbo$, then $\nc{1}=\dl$ and (ii) if $\nc{1}=\dl$, then $\nco{1}_{\nlo} = \lo_{\nlo}$. Since $\nco{1}_{\nlo}=\lo_{\nlo}$ implies $\nco{1}_{\nlo}\cap\pbo = \lo_{\nlo}\cap\pbo$, the claim follows immediately.

(i) First, we assume that $\nco{1}_{\nlo}\cap\pbo = \lo_{\nlo}\cap\pbo$.  Since $\mathrm{USTCON}$ is $\leq_{m}^{\nc{1}}$-complete for $\dl$, it suffices for us to verify that $\mathrm{USTCON}$ falls into $\nc{1}$ using our assumption.
We intend to define a polynomially-bounded maximization problem $P=(I,SOL,m,\text{\sc max})$ in the following manner.  An instance $w$ to $P$ is a tuple $(G,s,t)$ for an undirected graph $G=(V,E)$, including an edge $(s,t)$, and two vertices $s,t\in V$. A feasible solution of $w$ is a path $y=(y_1,y_2,\cdots, y_k)$ in $G$ starting at $s=y_1$. If
$y$ is a path from $s$ to $t$ with $k\geq2$, we set its objective value $m(w,y)$ to be $2$; otherwise, we set  $m(w,y)=1$. Obviously, $m$ is polynomially bounded and belongs to $\auxfl$.
Note that $w\in \mathrm{USTCON}$ iff there exists an $s$-$t$  path $y$ for which  $m(w,y)=2$.
Next, we claim that $P$ is in $\lo_{\nlo}$. Consider a log-space deterministic Turing machine $M$ that computes 2-approximate solutions by simply outputting the special edge $(s,t)$.
Hence, $P$ falls into $\lo_{\nlo}\cap\pbo$.

Since $P\in\pbo$ and $\lo_{\nlo}\subseteq \nco{1}_{\nlo}$, our assumption implies that $P$ is in $\nco{1}_{\nlo}$.
There exists a uniform family $\{C_n\}_{n\in\nat}$  of $\nc{1}$-circuits finding maximal solutions of $P$ (if any). Using those $C_n$'s, we construct another uniform $\nc{1}$-circuit family $\{D_n\}_{n\in\nat}$ such that $D_n(w)=1$ iff $C_n(w)$ is an $s$-$t$ path by checking that every consecutive elements in $y$ are connected by a single edge and the first and the last elements are $s$ and $t$, respectively. Clearly, $\{D_n\}_{n\in\nat^{+}}$ recognizes $\mathrm{USTCON}$. This implies that $\mathrm{USTCON}\in\nc{1}$, as requested.

(ii)  Assuming $\nc{1}=\dl$, let us consider any problem $P=(I,SOL,m,goal)$ in $\lo_{\nlo}$. To simplify the following argument, we assume that all feasible solutions of $P$ are expressed {\em in binary}. Consider the case of $goal=\text{\sc min}$. Let $p$ be a polynomial that upper-bounds the size of any feasible solution of $P$.
Let $M$ be a log-space deterministic Turing machine computing minimal solutions of $P$. To build a solver for $P$, we first define an oracle $A$ as a collection of $\pair{x,1^i,b}$ for which there exists a binary string $y$ with $|y|\leq p(|x|)$ such that $M(x)$ outputs $y$ and the $i$th bit of $y$ is $b\in\{0,1\}$.  Since $A$ belongs to $\dl$, our assumption implies that $A$ falls into  $\nc{1}$.
Note that, if $M(x)=y$ with $y=y_1y_2y_3\cdots y_e$, then $y$ coincides with $A(\pair{x,1,y_1}) A(\pair{x,1^2,y_2})
A(\pair{x,1^3,y_3}) \cdots A(\pair{x,1^{e},y_e})$, where $A$ is treated as its characteristic function.

To recover the $i$th bit $y_i$ of $M(x)$ for all $i\in[e]$, we compute both $A(\pair{x,1^i,0})$ and $A(\pair{x,1^i,1})$ and decide that $y_i$ takes a value $b\in\{0,1\}$ exactly when  $A(\pair{x,1^i,b})=1$ and $A(\pair{x,1^i,\overline{b}})=0$.  This procedure can be done in parallel with accesses to the oracle $A$ suing an $\nc{1}$ circuit. Since $A\in\nc{1}$, the procedure can be implemented using uniform $\nc{1}$-circuits without
any oracle Since $\nc{1}=\dl$ implies $\fnc{1}=\fl$, the measure function $m$ is in $\auxfnc{1}$. From this, we conclude that $m(x,M(x))$ is computed from $x$ by a suitable $\nc{1}$-circuit.
Thus, $P$ belongs to $\nco{1}_{\nlo}$.

(3) Notice that, in this proof, we consider $\npo$ problems instead of $\nlo$ problems. Similarly to (2), we wish to prove that (i) $\lo_{\npo}\cap\pbo = \po_{\npo}\cap\pbo$ implies $\dl=\p$ and (ii) $\dl=\p$ implies $\lo_{\npo}=\po_{\npo}$.

(i) Assume that $\lo_{\npo}\cap\pbo = \po_{\npo}\cap\pbo$.
Consider CVP, which is $\leq_{m}^{\nc{1}}$-complete for $\p$.
Founded on CVP, we intend to build a maximization problem $P$ of the form $(I,SOL,m,\text{\sc max})$ lying  in  $\po_{\npo}\cap\pbo$.
Define $I$ as the set of all Boolean circuits and, for any circuit $C\in I$, define  $SOL(x)$ to be a set of all variable assignments $\sigma$ setting all input variables of $C$ to either $0$ or $1$ for which $\sigma$ forces $C$ to output $1$. Since $\mathrm{CVP}\in\p$, $P$ belongs to $\po_{\npo}\cap\pbo$.
By our assumption, $P$ falls into $\lo_{\npo}\cap\pbo$. This yields a log-space deterministic Turing machine solving $P$. By running $M$ and checking $M$'s outputs, we can solve $\mathrm{CVP}$ using only log space. Hence, $\mathrm{CVP}$ is in $\dl$, implying that $\dl=\p$.

(ii) Assuming $\dl=\p$, let us consider {\sc Min Weight-st-Cut}, defined in Section \ref{sec:why-NC1}, which is $\EXreduces^{\nc{1}}$-complete for $\po_{\npo}$ by Proposition \ref{min-st-cut-is-po}.
Let $M$ be any  polynomial-time deterministic Turing machine computing minimal  solutions of {\sc Min Weight-st-Cut}.
Let $p$ be a polynomial satisfying that,  for all instances $x$ to {\sc Min Weight-st-Cut}, $p(|x|)$ bounds $M$'s output size. We define a problem $A$ to be composed of all strings of the form $\pair{x,1^i,b}$ for which $M(x)$ outputs a certain string $y$ and $y$'s $i$th bit is $b$.
It is not difficult to show that $A$ is in $\p$. Our assumption of $\dl=\p$ then implies that $A$ indeed belongs to $\dl$.

The following log-space algorithm can compute optimal solutions of {\sc Min Weight-st-Cut}: on input $x$ to {\sc Min Weight-st-Cut}, reconstruct $M(x)$ bit by bit by incrementing $i$ by one from $1$ to $p(|x|)$, check whether $\pair{x,1^i,0}\in A$ or $\pair{x,1^i,1}\in A$ or neither, and determine the $i$th bit of $M(x)$ to be $b\in\{0,1\}$ if $\pair{x,1^i,b}\in A$ and $\pair{x,1^i,\overline{b}}\notin A$.  Since $A\in \dl$, this procedure requires only log space. Hence, {\sc Min Weight-st-Cut} belongs to $\lo_{\npo}$.
\end{proofof}

%%%%%%

Here, we further refine Proposition \ref{equivalent-relation}.

\begin{theorem}\label{refined-relation}
\begin{enumerate}\vs{-1}
  \setlength{\topsep}{-2mm}%
  \setlength{\itemsep}{0mm}% original = 1mm
  \setlength{\parskip}{0cm}%

\item  $\dl\neq \nl$  iff  $\lo_{\nlo}\cap\pbo \neq
\lsas_{\nlo}\cap\pbo \neq
\apxl_{\nlo}\cap\pbo \neq \nlo\cap\pbo$.

\item $\nc{1}\neq\dl$ iff the following classes are mutually distinct: $\nco{1}_{\nlo}$, $\ncas{1}_{\nlo}$, $\apxnc{1}_{\nlo}$, and $\lo_{\nlo}$.

\item  $\dl\neq \p$ iff the following classes are mutually distinct: $\lo_{\npo}$, $\lsas_{\npo}$, $\apxl_{\npo}$, and $\po_{\npo}$.
\end{enumerate}
\end{theorem}

Given any string $y\in\{0,1\}^{+}$, we set $del(y)$ to be the string obtained from $y$ by deleting all $0$s from the leftmost bit of $y$ until the first $1$. For example, $del(0010)=10$, $del(01001)=1001$, and $del(00)=\lambda$. We then define $rep_{+}(y)$ to be {\em  one plus} the natural number represented in binary as $del(y)$; that is, $rep_{+}(y)=1+rep(del(y))$.

\begin{proofof}{Theorem \ref{refined-relation}}
(1) Tantau \cite[Theorem 4.1]{Tan07} showed that, under the assumption of $\dl\neq\nl$, the following four classes are mutually distinct:  $\lo_{\nlo}\cap\pbo$, $\lsas_{\nlo}\cap\pbo$, $\apxl_{\nlo}\cap\pbo$, and $\nlo\cap\pbo$.
To see the converse, assume that those classes are different.  If $\dl=\nl$, then Proposition \ref{equivalent-relation}(1) implies $\lo_{\nlo}\cap\pbo =\nlo\cap\pbo$. This obviously contradicts our assumption because $\lo_{\nlo}\subseteq \lsas_{\nlo} \subseteq \apxl_{\nlo} \subseteq\nlo$.
Hence, the desired separation between $\dl$ and $\nl$ follows immediately.

(2) We shall show separately two directions of ``iff.''

(If--part) To show that $\nc{1}\neq\dl$, we assume that the following four classes are mutually different: $\nco{1}_{\nlo}o$, $\ncas{1}_{\nlo}$, $\apxnc{1}_{\nlo}$, and  $\lo_{\nlo}$. Toward a contradiction, we assume that $\nc{1}=\dl$. Proposition \ref{equivalent-relation}(2) implies that $\nco{1}_{\nlo} = \lo_{\nlo}$. This is a clear contradiction; therefore, we obtain $\nc{1}\neq\dl$, as requested.

(Only If--part)
Assume that $\nc{1}\neq\dl$. First, notice that, by Proposition \ref{equivalent-relation}(2), we obtain $\nco{1}_{\nlo}\cap\pbo \neq \lo_{\nlo}\cap\pbo$. Since $\nco{1}_{\nlo}\subseteq \ncas{1}_{\nlo}\subseteq \apxnc{1}_{\nlo}$, it suffices to verify four  inequalities: (a) $\nco{1}_{\nlo}\cap \pbo \neq \ncas{1}_{\nlo}\cap \pbo$, (b) $\ncas{1}_{\nlo}\cap \pbo \neq \apxnc{1}_{\nlo}\cap \pbo$, (c) $\lo_{\nlo}\cap\pbo\nsubseteq \apxnc{1}_{\nlo}\cap \pbo$, where (c) implies $\lo_{\nlo}\neq\lo_{\nlo}$ and also $\ncas{1}_{\nlo} \neq \lo_{\nlo}$.

(a) To lead to a contradiction, let us assume that $\nco{1}_{\nlo}\cap \pbo = \ncas{1}_{\nlo}\cap \pbo$. We express $\text{\sc Max UApp-Vertex}$ given in Section \ref{sec:LSAS} as $(I,SOL,m,\text{\sc max})$. Since {\sc Max UApp-Vertex} belongs to $\ncas{1}_{\nlo}\cap\pbo$ by Proposition \ref{ncsas-example}(1), our assumption places {\sc Max UApp-Vertex} in $\nco{1}_{\nlo}\cap\pbo$.

Take an $\nc{1}$ circuit $C$ computing optimal solutions of {\sc Max UApp-Vertex}. Consider the set $L$ composed of $(G,s,w,\ell)$ such that $(G,s,w)\in I$, $\ell\in\nat^{+}$, and there exists a vertex $t\in V$ for which $w(t)\geq \ell$ and $s$ and $t$ are connected in $G$. This problem $L$ falls into $\nc{1}$, because it can be computed by the following $\nc{1}$-circuit: on input $x=(G,s,w,\ell)$, we first compute $y=C(G,s,w)$, extract a vertex $t\in V$ from $y$ satisfying $w(t)=m^*(x)$, and check whether $w(t)\geq \ell$. It therefore suffices to verify that $\mathrm{USTCON}\leq_{m}^{\nc{1}} L$, because this implies that $\mathrm{USTCON}\in\nc{1}$, a contradiction.
From any instance $(G,s,t)$ to $\mathrm{USTCON}$, we define $(G,s,w,k)$ as follows. Define $w(t)=1+\log\log|V|$ and $w(v)=\log\log|V|$ for all $v\in V-\{t\}$. Let $k=w(t)$.
Clearly, $(G,s,w,k)$ can be computed from $(G,s,t)$ using an appropriate $\nc{1}$-circuit.
It is not difficult to show that $(G,s,t)\in \mathrm{USTCON}$ iff $(G,s,w,k)\in L$. Hence, $\mathrm{USTCON}\leq_{m}^{\nc{1}}\dl$ follows.

(b) Since $\mathrm{USTCON}$ in $\dl$, it follows from our assumption that  $\mathrm{USTCON}\notin\nc{1}$. From $\mathrm{USTCON}$, we construct the maximization problem $P=(I,SOL,m,\text{\sc max})$ defined as in the proof (i) of Proposition  \ref{equivalent-relation}(2). This problem $P$ actually belongs to $\apxnc{1}_{\nlo}\cap\pbo$, because there is an NC$^{1}$ algorithm computing 2-approximate solutions to $P$. For such a solution $y$ on an instance $x$, we obtain $m(x,y)/m^*(x)\leq 2$, which means $m(x,y)\leq 2$. If $P\in\ncas{1}_{\nlo}$, then, by setting $r=3/2$, we obtain $r$-approximate solutions $z$ on input $x$ using an $\nc{1}$-circuit. If $x\in\mathrm{USTCON}$, then, since $R(x,z)\leq 3/2$,  $m(x,z)=2$ follows; thus, $z$ must be a $s$-$t$ path.
If $x\notin\mathrm{USTCON}$, then $m(x,z)=1$ and $z\neq t$. This implies that $\mathrm{USTCON}\in\nc{1}$, a contradiction. Therefore, $P$ is not in $\ncas{1}_{\nlo}$.

(c) Assuming that $\lo_{\nlo}\cap \pbo \subseteq \apxnc{1}_{\nlo}\cap \pbo$, we plan to derive that $\mathrm{DSTCON}\in\nc{1}$. Consider {\sc Max UB-Vertex} and set it as $(I_0,SOL_0,m_0,\text{\sc max})$. Proposition \ref{max-bvertex-complete} shows its $\EXreduces^{\nc{1}}$-completeness for $\lo_{\nlo}\cap\pbo$. From our assumption, we derive that this problem is also in $\apxnc{1}_{\nlo}$.
With a suitable constant $r>1$, take an NC$^{1}$ circuit $C$ producing  $r$-approximate solutions to {\sc Max UB-Vertex}.
Let us  construct another NC$^{1}$ circuit solving DSTCON. Given any instance $(G,s,t)$ to DSTCON with $G=(V,E)$, we define a weight function $w$ as $w(t)=r+1$ and $w(v)=1$ for all $v\in V-\{t\}$. Here, we consider $x=(G,s,w)$. We run $C$ to find an $r$-approximate solution, say, $u\in V$.  If $(G,s,t)\in \mathrm{DSTCON}$, then, since $\frac{m_0^*(x)}{m_0(x,u)}\leq r$, it follows that $m_0(x,u)\geq \frac{m_0^*(x)}{r}=1+\frac{1}{r}>1$. Since $w(u)>1$, $u$ must be $t$. If $(G,s,t)\notin\mathrm{DSTCON}$, then we can conclude that $u$ cannot be $t$. Hence, $\mathrm{DSTCON}$ must be in $\nc{1}$, as requested.

(3) Note that $\lo_{\npo}\subseteq \po_{\npo}$ and $\lo_{\npo}\subseteq \lsas_{\npo}\subseteq \apxl_{\npo}$. In Proposition \ref{equivalent-relation}(3), we have already proven that $\dl\neq\p$ iff $\lo_{\npo}\neq\po_{\npo}$. It still remains to show that $\dl\neq\p$ implies (a) $\po_{\npo}\nsubseteq \apxl_{\npo}$, (b)  $\lo_{\npo}\neq \lsas_{\npo}$, and (c) $\lsas_{\npo}\neq \apxl_{\npo}$.
Note that $\po_{\npo}\nsubseteq \apxl_{\npo}$ yields both $\apxl_{\npo}\neq\po_{\npo}$ and $\lsas_{\npo}\neq\po_{\npo}$.

(a) For this claim, it suffices to prove the following statement: (*)  $\po_{\npo}\subseteq \apxl_{\npo}$ implies $\dl=\p$.
Let us assume that $\po_{\npo}\subseteq \apxl_{\npo}$. Our goal is to show that $\mathrm{CVP}$ belongs to $\dl$ since $\dl=\p$ follows from the $\leq_{m}^{\nc{1}}$-completeness of $\mathrm{CVP}$ for $\p$.
Consider $\text{\sc MinCVP}$, which is closely related to the decision problem {\sc CVP}, defined in Appendix.
Notice that $\text{\sc MinCVP}$ is $\EXreduces^{\nc{1}}$-complete for  $\po_{\npo}$ by Lemma \ref{MinCVP-complete}.  Let $\text{\sc MinCVP} = (I,SOL,m,\text{\sc min})$. Since $\mathrm{MinCVP}\in \apxl_{\npo}$ by our assumption, we can take a function $h\in\fl$ and a constant $\gamma\geq 1$ such that $R(x,h(x))\leq \gamma$ for all $x\in I$, where $R$ indicates the performance ratio for {\sc MinCVP}. Take an integer $k\geq1$ for which $\gamma<2^k$.

Let $z=\pair{C_x}$ be an instance of $\mathrm{CVP}$. This implies that  $\pair{C_x}\in\mathrm{CVP}$ iff $1\in SOL^*(\pair{C_x,1^{n},1})$, where $n=|x|$.
We define $C'_x$ as follows. For the outcome $z$ of $C_x$, $C'_x$ outputs $0^k$ if $z=0$; $10^{k-1}$ otherwise. The construction of $C'_x$ can be done using log space. Given any outcome $z'$ of $C'_x$, $rep(z')$ is either $0$ or $2^{k}$. Let $w=\pair{C'_x,1^n,1^k}$, an instance to {\sc MinCVP}. It holds that, since $R(w,h(w))\leq \gamma$, $C_x$ outputs $0$ (resp., $1$) iff $h(w)=0^k$ (resp., $10^{k-1}$). Hence, we can determine whether $\pair{C_x}\in \mathrm{CVP}$ simply by constructing $w=\pair{C'_x,1^n,1^k}$ and checking whether $h(w)=10^{k-1}$. This implies that $\mathrm{CVP}$ belongs to $\dl$, as requested.

(b) In this case, we define $\text{\sc MinASCVP}=(I,SOL,m,\text{\sc min})$ as follows.
Define $I=\{(C,x)\mid C: \text{circuit with $|x|$ inputs and $1$ output},\, x\in \{0,1\}^*\}$ and $SOL(C,x)=\{y1^{|x|}\mid y\geq C(x)\}$ for any $(C,x)\in I$.
Given any $(C,x)\in I$ and any $y\in SOL(C,x)$, define $m((C,x),y)= 2^{|x|+1}+C(x)$. Note that $\max_{y\in SOL(C,x)}\{m((C,x),y)\}\leq (1+\frac{1}{2^{n+1}})\min_{y\in SOL(C,x)}\{m((C,x),y)\}$. Consider the following algorithm $M$: on input $((C,x),r)$ with $r\in\rational^{>1}$, if $n+1>\log(\frac{1}{r-1})$, then $M$ outputs $1^{|x|+1}$; otherwise, $M$ computes $C(x)$ by brute force and outputs $C(x)1^{|x|}$. Since $C(x)$ can be computed deterministically within $O(n^2)$ steps, $M$ requires the total tape space of $O(\log{n})+O(n^2)$, which equals $O(\log{n})+O(\log^2(\frac{1}{r-1}))$.
I also follows from  $1+2^{-(n+1)}<r$ that $M$ produces an $r$-approximate solution to $x$.
Hence, {\sc MinASCVP} belongs to $\lsas_{\npo}$.

Next, we claim that $\text{\sc MinASCVP}\in\lo_{\npo}$ implies $\mathrm{CVP}\in\dl$. Assume that $\text{\sc MinASCVP}\in\lo_{\npo}$ and take a log-space deterministic Turing machine $N$ that solves {\sc MinASCVP}. The following log-space algorithm then solves CVP. On input $z=(C,x)$, run $N$ on $z$ and output $y$ if $N$'s output is of the form $y1^{|x|}$. It is easy to verify that this algorithm is correct. Thus, CVP is in $\dl$. This contradicts our assumption that $\dl\neq\p$ since CVP is $\p$-complete. Therefore, {\sc MinASCVP} does not belong to $\lo_{\npo}$.

(c) Toward a contradiction, we assume that $\lsas_{\npo}=\apxl_{\npo}$.
Here, we define a new problem $\text{\sc Min1CVP}=(I,SOL,m,\text{\sc min})$ by setting $I=\{(C,x)\mid C:\,\text{circuit with $|x|$ inputs and 1 output},\, x\in\{0,1\}^{+}\}$, $SOL(C,x)=\{y1^{|x|}\mid y\geq C(x)\}$,  and $m((C,x),y1^{|x|}) = rep_{(+)}(y1^{|x|})$ for any $(C,x)\in I$ and any $y1^{|x|}\in SOL(C,x)$.

To claim that {\sc Min1CVP} belongs to $\apxl_{\npo}$, consider the following  log-space deterministic Turing machine $M$: on each input $(C,x)\in I$, $M$  produces $1^{|x|+1}$ on its output tape. Since $rep_{(+)}(M(C,x))=2^{|x|+2}$ and $rep_{(+)}(y1^{|x|})\in\{2^{|x|+2},2^{|x|+1}\}$, $M(C,x)$ is a $2$-approximate solution to $(C,x)$. Thus, {\sc Min1CVP} is in $\apxl_{\npo}$.
By our assumption, {\sc Min1CVP} is also in $\lsas_{\npo}$. There is a log-space deterministic Turing machine $N$ that produces $3/2$-approximate solutions to {\sc Min1CVP}. From this $N$, we define $N'$ as follows. On input $z=(C,x)$, $N'$ first  runs $N$, and outputs $1$ if $N(C,x)$ equals $1^{|x|+1}$, and outputs $0$ otherwise. Clearly, $N'$ solves $\mathrm{CVP}$. Hence, $\mathrm{CVP}$ belongs to $\dl$, a contradiction against $\dl\neq\p$ because $\mathrm{CVP}$ is $\leq_{m}^{\dl}$-complete for $\p$.
\end{proofof}

%%%

To close this section, we wish to demonstrate further separations with {\em no unproven assumption}. Notice that $\ac{0}$ is known to be {\em properly} included inside $\nc{1}$ because the parity function, which is in $\nc{1}$,  requires at least non-uniform constant-depth circuits of super-polynomial size \cite{Ajt83,FSS84}.
The following proof relies on this fact.

\begin{theorem}\label{AC0-separation}
$\aco{0}_{\nlo}\neq \acas{0}_{\nlo}\neq \apxac{0}_{\nlo} \neq \nco{1}_{\nlo}$.
\end{theorem}

Concerning $\nlo$ and $\npo$,  most optimization and approximation classes enjoy the following {\em upward separation property}: if $\DD^{(1)}_{\nlo}\neq \DD^{(2)}_{\nlo}$, then $\DD^{(1)}_{\npo}\neq \DD^{(2)}_{\npo}$. Theorem \ref{AC0-separation} thus yields the separations $\aco{0}_{\npo}\neq \acas{0}_{\npo} \neq\apxac{0}_{\npo} \neq \nco{1}_{\npo}$.

To derive Theorem \ref{AC0-separation}, we shall use the {\em parity function} $\pi$, which is defined as $\pi(x_1,x_2,\ldots,x_n) = x_1\oplus x_2\oplus \cdots \oplus x_n$, where each variable  $x_i$ takes a binary value.
The fact that the parity function is out of $\ac{0}$ was first claimed by  Furst, Saxe, and Sipser \cite{FSS84} as well as Ajtai \cite{Ajt83}.
Using this $\pi$, we further define
$\pi^*(x_{11},\ldots,x_{1n},x_{21},\ldots,x_{2n},\ldots,x_{n1},\ldots,x_{nn})$ to be the $n$-bit string $\pi(x_{11},\ldots,x_{1n}) \pi(x_{21},\ldots,x_{2n})\cdots \pi(x_{n1},\ldots,x_{nn})$. It is not difficult to show that $\pi^*$ is in $\fnc{1}$ but not in $\fac{0}$ because $\pi$ resides in the difference $\nc{1}-\ac{0}$.

\begin{proofof}{Theorem \ref{AC0-separation}}
The subsequent proof consists of three parts: (1) $\nco{1}_{\nlo}\nsubseteq \apxac{0}_{\nlo}$,  (2) $\acas{0}_{\nlo}\neq \apxac{0}_{\nlo}$, and (3) $\aco{0}_{\nlo}\neq \acas{0}_{\nlo}$.

(1) For this separation,  we consider a minimization problem $\text{\sc Min m-Parity}=(I,SOL,m,\text{\sc min})$ defined as follows.
Let $I = \bigcup_{n\in\nat^{+}}\{0,1\}^{n^2}$ and let $SOL(x) = \{y\in\{0,1\}^{|x|} \mid  rep_{+}(y)\geq rep_{+}(\pi^*(x))\}$ for each $x\in I$ with  $|x|=n^2$. Let $m(x,y) = rep_{+}(y)$ for all $y\in SOL(x)$. Clearly, $I$ is in $\dl$, $I\circ SOL$ is in $\auxl$, and $m$ is in $\auxfl$. This indicates that {\sc Min m-Parity} is an $\nlo$ problem.
Next, we shall  argue that $\text{\sc Min m-Parity}\in\nco{1}_{\nlo}$.  Note that, for every $x\in I$, $SOL^*(x)=\{\pi^*(x)\}$ holds. Since $\pi^*$ is in $\fnc{1}$, it follows that $\text{\sc Min m-Parity}$ is $\nc{1}$-solvable.

Let us prove that $\text{\sc Min m-Parity}\not\in\apxac{0}_{\nlo}$. To lead to a contradiction, we assume otherwise; that is, $\text{\sc Min m-Parity}$ is in $\apxac{0}_{\nlo}$. There exist a constant $\gamma>1$ and a uniform family $\{C_n\}_{n\in\nat^{+}}$ of $\ac{0}$-circuits such that,  for every $x\in I$,  $C_{|x|}(x)$ computes a string $y$ in $SOL(x)$ satisfying that (*)
$rep_{+}(y)/\gamma \leq rep_{+}(\pi^*(x)) \leq rep_{+}(y)$. Take any number $n$ satisfying $2^n > \gamma$ and any string $x\in\{0,1\}^n$. Consider $x^n$ and define  $y = C_{n^2}(x^n)$.
If $\pi(x)=1$, then we obtain $rep_{+}(\pi^*(x^n))=2^{n+1}$.  Since $|y|=n$ and $y\in SOL(x)$, it must hold that $rep_{+}(y)=2^{n+1}$, implying $y=1^n$. By contrast, if $\pi(x)=0$, then we obtain $rep_{+}(\pi^*(x^n))=1$. By Condition (*), it follows that $1\leq rep_{+}(y)\leq \gamma$. Since $\gamma<2^n$, $y$ must have the form $0z$ for a certain string $z$ in $\{0,1\}^{|x|-1}$. As a consequence, $\pi(x)$ equals the first bit of $y$. This gives an $\ac{0}$-circuit that computes $\pi$.
This is a contradiction against the fact that $\pi$ is not in $\ac{0}$. Therefore, $\text{\sc Min m-Parity}$ is not in $\apxac{0}_{\nlo}$.

(2) Based on the parity function $\pi$, we define a simple minimization problem, {\sc Min 1-Parity}, which is somewhat similar to {\sc Min m-Parity} defined in the proof of (1). {\sc Min 1-Parity} takes an instance $x\in\{0,1\}^n$ and finds a solution $y\in\{0,1\}$ such that $rep(1y)\geq rep(1\pi(x))$ with the measure function $m(x,y) = rep(1y)$.

Here, we claim that {\sc Min 1-Parity} is in $\apxac{0}_{\nlo}$. Consider the following $\ac{0}$-circuit $C$: on input $x$, output $y_1=1$. Since $rep(1y_1)=3$ and $rep(1\pi(x))\in\{2,3\}$, it follows that $rep(1y_1)/2 \leq rep(1\pi(x))\leq rep(1y_1)$. Thus, $C$ is a $2$-approximate algorithm for {\sc Min 1-Parity}. We then conclude that {\sc Min 1-Parity} belongs to $\apxac{0}_{\nlo}$.

Assume that $\acas{0}_{\nlo}= \apxac{0}_{\nlo}$. This means that {\sc Min 1-Parity} is in $\acas{0}_{\nlo}$. Take a uniform family $\{C_n\}_{n\in\nat}$ of $\ac{0}$-circuits such that, for any $x\in\{0,1\}^*$, $C_{|x|}(x)$ outputs a 5/4-approximate solution $y$. For simplicity, write $y$ for $C_{|x|}(x)$.
Consider the case of $\pi(x)=0$. If $y=1$, then $rep(1y)=3$. This contradicts the inequalities: $(4/5)rep(1y)\leq rep(1\pi(x))\leq rep(1y)$. This implies that $y$ is $0$. In the case of $\pi(x)=0$, when $y=0$, it does not hold that $(4/5)rep(1y)\leq rep(1\pi(x))\leq rep(1y)$. Hence, $y$ must be $1$.
Therefore, $\{C_n\}_{n\in\nat}$ computes $\pi$ correctly. This implies that $\pi$ is actually in $\ac{0}$, a contradiction against $\pi\notin\ac{0}$. Therefore, $\acas{0}_{\nlo}\neq \apxac{0}_{\nlo|}$ follows.

(3) We define another minimization problem {\sc Min Bit-Parity}, which belongs to $\acas{0}_{\nlo}$ by setting $I=\bigcup_{n\in\nat^{+}}\{0,1\}^{n}$, $SOL(x) =\{ y\in \{0,1\}^{|x|}\mid rep(y)\geq rep(1^{|x|}\pi(x))\}$ for each $x\in I$, and $m(x,y) = \max\{1,rep(y)\}$ for $y\in SOL(x)$.

Consider a circuit $C_n$ that takes input $(x,r)$ with $n=|x|$ and, if $r\geq 1+\frac{1}{2^{n+1}-2}$, then  produces $1^{n}0$, and  if $1<r< 1+\frac{1}{2^{n+1}-2}$, then  computes $\pi(x)$ exactly and outputs $1^n\pi(x)$. Note that the value $\pi(x)$ can be exactly computed by an appropriately chosen circuit, say, $D$ of size $n^{O(1)}$ and depth $O(\log{n})$. When $1<r< 1+\frac{1}{2^{n+1}-2}$, since this is equivalent to $n<\log(\frac{2r-1}{r-1})-1$ and $r$ is treated as a constant, $D$ has depth $O(\log\log(\frac{2r-1}{r-1}))$, which is bounded by a certain constant. Thus, $C_n$ is an $\ac{0}$-circuit parameterized by $r$. Let $y=C_n(x,r)$. In the case of $r\geq 1+\frac{1}{2^{n+1}-2}$, we obtain $rep(y)=2^{n+1}-1$. When $\pi(x)=0$, since $rep(1^n\pi(x))=2^{n+1}-2$, it follows that $rep(y)/r = 2^{n+1}-2\leq rep(1^n\pi(x))\leq rep(y)$. In contrast, when $\pi(x)=1$, clearly we obtain $rep(y)/r<rep(1^n\pi(x))\leq rep(y)$. This implies that $C_n$ computes  $r$-approximate solutions.

Toward a contradiction, we assume that $\aco{0}_{\nlo}=\acas{0}_{\nlo}$. This implies that {\sc Min Bit-Parity} is in $\aco{0}_{\nlo}$. Let $\text{\sc Min Bit-Parity} =(I,SOL,m,\text{\sc min})$.
Next, we argue that $\pi\in\ac{0}$, leading to a clear contradiction. Since $\text{\sc Min Bit-Parity}\in\acas{0}$, there is a uniform family $\{C'_n\}_{n\in\nat}$ of $\ac{0}$-circuits solving {\sc Min Bit-Parity}. Since $SOL(x)=\{1^{|x|}\pi(x)\}$, it follows that $C'_n(x)=1^{|x|}\pi(x)$. From $C'_n$, we can design another $\ac{0}$-circuit that computes $\pi(x)$. Thus, we obtain the desired conclusion of $\pi\in\ac{0}$. This yields the separation between $\aco{0}_{\nlo}$ and $\acas{0}_{\nlo}$.
\end{proofof}

%%%%%%%%%%%%%%%%%%%%%%%%%%%%%
\section{Discussions and Future Research Directions}

We have refined the existing framework of combinatorial optimization problems in hope of providing a useful means of classifying various optimization problems lying inside the $\p$-solvable class $\po_{\npo}$; in particular, we have been focused on NL optimization problems or $\nlo$ problems, following early work of \'{A}lvarez and Jenner \cite{AJ93,AJ95} and Tantau \cite{Tan07}. In a course of
our exploring study on such optimization problems, unfortunately, we have left numerous fundamental issues unsettled. For an advance of our knowledge, they definitely require reasonable answers and explanations. As a quick guide to those pending issues, we want to shed clear light on some of the challenging issues arising naturally in the course of our study.

\begin{enumerate}\vs{-1}
  \setlength{\itemsep}{1mm}% original = 1mm
  \setlength{\parskip}{0cm}%

\item {[Class Separations and Inclusions]} We have introduced numerous classes of refined optimization problems in Section \ref{sec:comb-OPs} but a few classes are successfully proven to be different (cf. Theorem \ref{AC0-separation}). One of our ultimate goals is to give proofs for separations of other important optimization and approximation classes. As seen in Section \ref{sec:complexity-OP}, those separations are closely tied up to long-standing open questions regarding the classes of their associated decision problems, and thus the separations of refined classes immediately lead to definitive answers to such open questions. Beside the separation issues, there are important unsettled questions concerning the inclusion relationships among newly refined classes. For example, contrary to the well-known inclusion $\nc{1}\subseteq \dl$, we suspect that  $\apxnc{1}_{\nlo}\nsubseteq \lo_{\nlo}$, $\ncas{1}_{\nlo}\nsubseteq\lo_{\nlo}$, and $\apxnc{1}\nsubseteq\lsas_{\nlo}$.

\item {[Completeness Issues]} (1) In this paper, a number of optimization  problems have been demonstrated to be complete for certain optimization and approximation classes within $\nlo$. It is imperative to find  more natural and useful optimization problems and prove them to be complete for target classes, in particular, $\lsas_{\nlo}\cap\pbo$ and $\ncas{1}_{\nlo}\cap\pbo$.

    (2) In Theorem \ref{Min-Path-complete} and Corollary \ref{Max-Path-Weight-complete}, we have shown that both $\maxnl$ and $\minnl$ contain $\sAPreduces^{\nc{1}}$-complete problems. Unlike $\npo$, we do not know whether $\nlo$ ($=\maxnl\cup \minnl$) contains  $\sAPreduces^{\nc{1}}$-complete problems unless $\auxfl$ is closed under division (cf. Proposition \ref{Min-Path-in-NLO}. Can we remove this closure property? Similarly, are there any complete problem in $\apxl_{\nlo}$ without any assumption? Moreover, several optimization problems discussed in this paper have not proven to be complete. Those pending problems include {\sc Max $\lambda$-DFA}, {\sc Max B-Vertex}, and {\sc Min UPath-Weight}.

    (3) There is a large collection of works presenting L-complete and NL-complete decision  problems \cite{AG00,CM87,JLL76,Jon75}. In most cases, it is relatively straightforward to turn those problems into their associated NLO problems. Which of them are actually complete for refined optimization and approximation classes under suitably chosen reductions?

\item {[Polynomially-Bounded Problems]} The relationship between NL and NLO appears to be quite different from the relationship between NP and NPO. Such difference comes primarily from an architectural limitation of log-space (auxiliary)  Turing machines. In particular, polynomially-bounded optimization problems are quite special for log-space computation. We have proven that $\nco{1}_{\nlo}=\lo_{\nlo}$ iff $\nco{1}_{\nlo}\cap\pbo =\lo_{\nlo}\cap\pbo$ (Proposition \ref{equivalent-relation}(2)) but we do not know whether (i) $\lo_{\nlo}=\nlo$ iff $\lo_{\nlo}\cap\pbo=\nlo\cap\pbo$ and (ii) $\nco{1}_{\nlo}=\nlo$ iff $\nco{1}_{\nlo}\cap\pbo =\nlo\cap\pbo$? Prove or disprove these equivalences (i)--(ii).

\item {[Further Refinement of Approximability]} Beyond $\apxl_{\nlo}$, Tantau \cite{Tan07} discussed classes of NLO problems whose $n^{O(1)}$-approximate (as well as $2^{n^{O(1)}}$-approximate) solutions are computed by log-space deterministic Turing machines. We briefly call them $\mathrm{poly\mbox{-}}\apxl_{\nlo}$ and $\mathrm{exp\mbox{-}}\apxl_{\nlo}$. Those two classes fill a seemingly wide gap between $\apxl_{\nlo}$ and $\nlo$. For example, we have already mentioned in Section \ref{sec:general-complete} and \ref{sec:graph-problems} that {\sc Min UPath-Weight} and {\sc Max B-Vertex} belong to $\mathrm{poly\mbox{-}}\apxl_{\nlo}$. Find more natural complete problems falling into these special approximation classes. These classes can be further expanded into a more general class denoted by $\ell(n)\mbox{-}\apxl_{\nlo}$, where $\ell(n)$ refers to $\ell(n)$-approximate solutions. Explore its features by finding natural problems inside it.

\item {[Heuristic and Probabilistic Approaches]} This paper has aimed at cultivating a basic theory of NLO problems with a major focus on complete problems. As a consequence, this paper has neglected any practical heuristic approach toward the NLO problems. The next possible step is to find such heuristic approaches to real-life NLO problems. Study also the power of probabilistic algorithms to solve NLO problems.

\item {[Help by Advice]} In 1980s, Karp and Lipton \cite{KL82} studied a notion of advice, which is an external source used to enhance any underlying Turing machines, and they introduced two major complexity classes, $\p/\log$ and $\p/\poly$, where ``$\log$'' and ``$\poly$'' refer to advice sizes of $O(\log{n})$ and $n^{O(1)}$, respectively. In analogy with these advised classes, it may be possible to formulate $\nlo/\log$ and $\nlo/\poly$ containing advised optimization problems. Explore the properties of those advised classes and show the separations among them.

\item {[Fixed Parameter Complexity]} Recall that optimization problems in $\lsas_{\nlo}$ and $\ncas{1}_{\nlo}$ admit certain approximation schemes, which are algorithms whose  resources (i.e., tape space or circuit size) described in terms of {\em fixed parameters}. NLO problems whose log-space approximation schemes have the performance ratios of, e.g.,  $k+\log{n}$ and $2^k\log{n}$, are treated equally inside $\lsas_{\nlo}$, but those schemes behave differently in practice. Therefore, it is more practical  to specify those fixed parameters and refine those classes.

\item {[More Classes Inside $\po_{\npo}$]} (1: Classes inside NLO) If the number of solutions of a given problem is relatively small, can we use this fewness information to find an optimal solution much more efficiently? Similarly to $\mathrm{FewP}$ (few polynomial time), we may define its log-space optimization counterpart, denoted by $\mathrm{FewLO}$, which is a class of NLO problems whose  sets of feasible solutions have sizes bounded from above by certain polynomials in $|x|$.
    Study the computational complexity of optimization problems inside  $\mathrm{FewLO}$.

    (2: Classes between $\po_{\npo}$ and NLO) A focal point of the current paper is optimization problems inside NLO. It still remains to fill the gap between $\po_{\npo}$ and NLO. Uniform circuit families, such as $\nc{k}$, $\ac{k}$, $\sac{k}$, and $\mathrm{TC}^{k}$  for $k\geq1$, may be used to form optimization/approximation classes beyond NLO. For example, SAC$^{1}$-circuits may introduce optimization and approximation classes, say, $\mathrm{SAC}^{1}\mathrm{O}_{\npo}$ and $\mathrm{APXSAC}^{1}_{\npo}$. Find complete problems for those new classes lying above NLO. Concerning to the power of $\sac{1}$, it is also important to determine whether  $\nlo$ coincides with $\mathrm{SAC}^{1}\mathrm{O}_{\nlo}$, expanding Lemma \ref{char-with-AC1}(1).

\item {[Other Categories of Optimization Problems]} Lately, Yamakami \cite{Yam14} studied OptCFL, which is a CFL analogue of Krentel's OptP \cite{Kre88} and \`{A}lvarez and Jenner's OptL \cite{AJ93,AJ95}, where CFL stands for {\em context-free languages}. This fact suggests a possibility of introducing CFL-based optimization problems under the current framework of refined optimization problems. Cultivate a theory based on those weak optimization problems.
\end{enumerate}

%%%%%%%%%%%%%%%%%%%%%%%%%%%%%
%%%%%%%%%%%%%%%%%%%%%%%%%%%%%
\section*{Appendix: Proof of Proposition \ref{min-st-cut-is-po}}

Proposition \ref{min-st-cut-is-po} states the completeness of {\sc Min Weight-st-Cut} for $\po_{\npo}$. Here, we restate the proposition.

\ms
\n{\bf {\em Proposition \ref{min-st-cut-is-po} (again).}}
{\it $\text{\sc Min st{-}Cut}$ (on weighted directed graphs) is $\EXreduces^{\nc{1}}$-complete for $\po_{\npo}$.
}
\ms

Note that an unweighted version of {\sc Min Weight-st-Cut} is also known as the {\em edge connectivity problem}. Gabow \cite{Gab91} gave a sequential algorithm of finding the edge connectivity $c$  of a (directed or undirected) graph of $n$ vertices and $m$ edges in time $O(cn \log (n^2/m))$.

For readability, we have left Proposition \ref{min-st-cut-is-po} unproven in Section \ref{sec:why-NC1}. In this appendix, we shall give the proof of Proposition \ref{min-st-cut-is-po}  for completeness.
As noted in Section \ref{sec:why-NC1}, this result is drawn from the $\p$-completeness proof of Goldschlager \etalc~\cite{GSS82} for a decision version of the maximum $s$-$t$ flow problem.

To prove the proposition, we first look into the decision problem $\mathrm{CVP}$ (Circuit Value Problem) into a minimization problem in $\po_{\nlo}$. Given each   binary string $x$, the notation $rep(x)$ denotes one plus the natural number represented in binary by $x$.

\ms
{\sc Minimum Circuit Value Problem} ({\sc MinCVP}):
\renewcommand{\labelitemi}{$\circ$}
\begin{itemize}\vs{-2}
  \setlength{\topsep}{-2mm}%
  \setlength{\itemsep}{1mm}%
  \setlength{\parskip}{0cm}%
\item {\sc instance:} $\pair{C_x,1^n,1^k}$ with a circuit $C_x$ of $n$ inputs and $k$ outputs with AND, OR, and NOT gates for an input $x\in\{0,1\}^n$ (which must be specified in the description of the circuit).

\item {\sc Solution:} a binary string $y$ of length $k$ such that $rep(y) \geq rep(z)$, where $z$ is the outcome of the circuit $C_x$.

\item {\sc Measure:} the number $rep(y)$.
\end{itemize}

We further demand a circuit to satisfy the following conditions: (i) it is monotone (i.e., using only $AND$ and $OR$ gates), (ii) each input has fan-out at most $1$, (iii) each gate has fan-out at most $2$, and (iv) the top (i.e., root) gate is an $OR$ gate. With those conditions, we obtain the {\sc Minimum Monotone Circuit Value 2 Problem} ({\sc MinMCV2}) defined similarly to {\sc MinCVP}.
The decision version of {\sc MinMCV2} was shown to be $\leq_{m}^{\dl}$-complete for $\p$ \cite{GSS82}.

Here, we establish the $\EXreduces^{\nc{1}}$-completeness of {\sc MinCVP} and {\sc MinMCV2}.

\begin{lemma}\label{MinCVP-complete}
$\text{\sc MinCVP}$ and $\text{\sc MinMCV2}$  are $\EXreduces^{\nc{1}}$-complete for $\po_{\npo}$.
\end{lemma}

\begin{proof}
We first show that {\sc MinCVP} is in $\po_{\npo}$.  Given an instance $\pair{C_x,1^n,1^k}$ of {\sc MinCVP}, it is possible to calculate the outcome of $C_x$ by evaluating each gate of $C_x$ one by one within polynomial time. Hence, {\sc MinCVP} is in $\po_{\npo}$.

Next, we shall  show the $\EXreduces^{\nc{1}}$-hardness of $\mathrm{MinCVP}$ for $\po_{\npo}$. Let $P=(I,SOL,m,goal)$ be any optimization problem in $\po_{\npo}$. Let $x$ be any instance of $P$. Since $P$ is $\p$-solvable, we take a function $h\in\fp$ satisfying $h(x)\in SOL^*(x)$ for all $x\in (I\circ SOL)^{\exists}$. For each $x$, we pad extra $k'_x$ zeros so that $k'_x\geq0$ and $|h(x)10^{k'_x}|=k(n)$ for a certain absolute polynomial $k$. For convenience, we set $h'(x)=h(x)10^{k'_x}$.  It is known that any function in $\fp$ can be computed by a certain $\nc{1}$-uniform family of polynomial-size Boolean circuits (see, e.g., \cite{DK00}).

Let us define an $\EXreduces^{\nc{1}}$-reduction $(f,g)$ from $P$ to $\mathrm{MinCVP}$ as follows.
Since $h'\in\fp$, take  an $\nc{1}$-uniform family $\{C_{n}\}_{n\in\nat^{+}}$  of polynomial-size Boolean circuits $C_n$ of $n$ inputs and $k(n)$ outputs that computes $h'$. Define $f(x)=\pair{C_x,1^{|x|},1^{k(|x|)}}$, where $C_x$ is obtained from $C_{|x|}$ by replacing $|x|$ input variables with $|x|$ constant bits $x\in\{0,1\}^{|x|}$.
Since $\{C_n\}_{n\in\nat^{+}}$ is $\nc{1}$-uniform, there is an $\nc{1}$-circuit family to construct each $C_x$. Hence, $f$ must be in $\fnc{1}$. If $y$ is of the form $z10^{k'}$, then we write $\tilde{y}$ to denote $z$. Define $g(x,y) = \tilde{y}$. It is not difficult to show that if $y=C_x$ then $g(x,y) = \tilde{y} = h(x)$. Hence, $P$ is $\EXreduces^{\nc{1}}$-reducible to $\mathrm{MinCVP}$.  The transformation $f$ is computed by a certain family of $\nc{1}$-circuits.

The case for {\sc MinMCV2} is similar, however, by using restricted monotone circuits for $h'$.
\end{proof}

%%%%%%%%%%%%%%%%%%

Using Lemma \ref{MinCVP-complete}, we want to derive Proposition \ref{min-st-cut-is-po}.

\begin{proofof}{Proposition \ref{min-st-cut-is-po}}
It is known that {\sc Min st-cut} is in $\po_{\npo}$.
Thus, it suffices to reduce $\text{\sc MinMCV2}$ to $\text{\sc Min st{-}Cut}$ by an appropriate $\EXreduces^{\nc{1}}$-reduction $(f,g)$. For convenience, we set $\text{\sc MinMCV2}=(I_1,SOL_1,m_1,\text{\sc min})$. 
Let $\pair{C_x,1^n,1^k}$ be any instance of MinMCV2, where $C_x$ is expressed as a sequence $(\alpha_{n},\alpha_{n-1},\ldots,\alpha_0)$.

First, we consider the case where $k=1$. The function $f$ is of the form  $f(\pair{C_x,1^n,1^k})=\pair{G,s,t,c}$, where $G=(V,E)$ is defined as follows. Let $V=\{i\mid 0\leq i\leq n\}\cup\{s,t\}$. Let $c(s,i)=b\cdot 2^i$ if $\alpha_i$ is an input $b\in\{0,1\}$. If $\alpha_i$ is of the form $AND(j,k)$ or $OR(j,k)$, then let $c(j,i)=2^{j}$ and $c(k,i)=2^k$. Moreover, when  $\alpha_i$ is $AND(j,k)$, let $c(i,t)=2^j+2^k-d\cdot 2^i$; when $\alpha_i$ is $OR(j,k)$, let $c(i,s)=2^j+2^k-d\cdot 2^i$, where $d$ is the fan-out of $\alpha_i$. Now, let $E=\{(i,j)\mid c(i,j) \text{ is defined }\}$.  Note that $c(i,j)\leq 2^{n+1}\leq 2^{|V|}$ for any $(i,j)\in E$. For any $s$-$t$ cut $(S'_0,S'_1)$, we set $g(\pair{G,s,t,c},(S'_0,S'_1)) = 1$ if the cut capacity of $(S'_0,S'_1)$ is odd; $0$ otherwise. Note that $g$ can be computed by certain $\nc{1}$-circuits. Let $x=\pair{C_x,1^n,1}$. If $(S'_0,S'_1)$ is an $s$-$t$ cut, then $m_1(x,g(x,(S'_0,S'_1)))=m_1^*(x)=1$.  

Here, we define a special $s$-$t$ cut $(S_0,S_1)$ inductively. The sink $t$ and all vertices corresponding to inputs with $0$ given to $C_x$ are in $S_1$. For a vertex $i$ corresponding to a gate $AND$ or $OR$, if this gate outputs $0$, then the vertex $i$ is in $S_1$. The set $S_0$ is defined as $S_0 = V-S_1$. It is possible to verify that $C_n$ outputs $1$ iff the cut capacity of $(S_0,S_1)$ is odd. Moreover, we can prove that $(S_0,S_1)$ is minimal. 

Let $k\geq2$. Given $C_x$ with $n$ input bits and $k$ output bits, we first make $k$ copies $C_{x,1},C_{x,2},\ldots,C_{x,k}$ of $C_x$ and, for each $C_{x,i}$ ($i\in[k]$), we add the following new nodes: an input $\beta$ having $0$, two gates $AND(\beta,\gamma_1,\gamma_2,\ldots,\gamma_{i-1}, \gamma_{i+1},\ldots,\gamma_{k})$ and $OR(\gamma_i,\zeta)$, where $\gamma_j$ corresponds to the $j$th output bit of $C_x$ and $\zeta$ is the index corresponding to this $AND$ gate and by renaming all 
vertices $v$ of $C_{x,i}$ as $v^{(i)}$.
Let $\tilde{C}_{x,i}$ denote the circuit obtained from $C_{x,i}$ by adding these extra gates. Clearly, the output of $C_{x}$ coincides with the $k$-bit string $r_1r_2\cdots r_n$, where $r_i$ is the output of $C_{x,i}$.
For each $\tilde{C}_{x,i}$, we define a graph $G_i=(V_i,E_i)$ and a weight function $c_i$ as in the case of $k=1$.

Finally, we add new source $\tilde{s}$ and new sink $\tilde{t}$, and then 
define $\tilde{c}$ as follows.  We set $\tilde{c}(\tilde{s},s_i)=2^{n^2}$ and $\tilde{c}(t_i,\tilde{t})=2^{n^2}$ for every $i\in[k]$.  Moreover, we define $\tilde{c}(\ell_1,\ell_2) = 2^{in}c_{i}(\ell'_1,\ell'_2)$ if $\ell'_1,\ell'_2$ are in the same $\tilde{C}_{x,i}$; $0$ otherwise.
A new graph $\tilde{G}=(\tilde{V},\tilde{E})$ is defined similarly as before: let $\tilde{V}=\bigcup_{i\in[k]}V_i\cup \{\tilde{s},\tilde{t}\}$ (provided that $V_i\cap V_j=\setempty$ for any two distinct $i,j\in[k]$)  and $\tilde{E}=\{(i,j)\mid \tilde{c}(i,j)\;\text{is defined}\}$. Let $x=\pair{C_x,1^n,1^k}$. Note that $m_1(x,g(x,(\tilde{S}_0,\tilde{S}_1)))=m_1^*(x)=y_1y_2\cdots y_k$. Hence, $(f,g)$ reduces {\sc MinMCV2} to {\sc Min Weight-st-Cut}. 
Given a minimal $s^{(i)}$-$t^{(i)}$ cut $(S_{0,i},S_{1,i})$ of $G_i$ for each $i\in[k]$, define $\tilde{S}_0 = \{\tilde{s}\}\cup (\bigcup_{i\in[k]}S_{0,i})$  and $\tilde{S}_1 = \{\tilde{t}\}\cup (\bigcup_{i\in[k]}S_{1,i})$.  The cut  $(\tilde{S}_0,\tilde{S}_1)$ is a minimal $\tilde{s}$-$\tilde{t}$ cut of $G$.
We define $g$ as follows: the $i$th bit of $g(\pair{\tilde{G},\tilde{s},\tilde{t},\tilde{c}},(\tilde{S}_0,\tilde{S}_1))$ is $1$ iff the $s^{(i)}$-$t^{(i)}$ cut $(S_{0,i},S_{1,i})$ has odd capacity. This equivalence gives rise to the desired equivalence: $C_x$ outputs $y_1y_2\cdots y_k$ iff  $g(x,(\tilde{S}_0,\tilde{S}_1)) = y_1y_2\cdots y_k$.
\end{proofof}

%%%%%%%%%%%%%%%%%%%%%%%%%%
%%%%%%%%%%%%%%%%%%%%%%%%%%
%%%%%%%%%%%%%%%%%%%%%%%%%%
%%%%%%%%%%%%%%%%%%%%%%%%%%
\let\oldbibliography\thebibliography
\renewcommand{\thebibliography}[1]{%
  \oldbibliography{#1}%
  \setlength{\itemsep}{0pt}%
}
\bibliographystyle{plain}

%%%%%%%%%%%%%%%%%%%%%%%%%%%%%%%%%%%%%%%%%%%%%%%%%%
%%%%%%%%%%%%%%%%%%%%%%%%%%%%%%%%%%%%%%%%%%%%%%%%%%%
\end{document}